\documentclass[journal,draftcls,onecolumn]{IEEEtran}

\usepackage[utf8]{inputenc} 
\usepackage[T1]{fontenc}
\usepackage{url}
\usepackage{ifthen}
\usepackage{cite}
\usepackage[cmex10]{amsmath} 
\usepackage{enumitem}

\usepackage{ amsthm}
\newtheoremstyle{mystyle}
  {\topsep} 
  {\topsep} 
  {\normalfont} 
  {\parindent} 
  {\itshape} 
  {:} 
  { } 
  {} 

\theoremstyle{mystyle}

\newtheorem{prop}{Proposition}
\newtheorem{defn}{Definition}
\newtheorem{thm}{Theorem}
\newtheorem{lemma}{Lemma}

\newtheorem{claim}{Claim}

\newtheorem{cor}{Corollary}
\newtheorem{remark}{Remark}
\newtheorem{example}{Example}

\usepackage{booktabs}
\usepackage{tabularx}
\usepackage{array}
\usepackage{multirow}

\usepackage{soul}

\usepackage{cite}
\usepackage{amsmath,amssymb,amsfonts}

\usepackage{dsfont}
\usepackage{algorithmic}
\usepackage{graphicx}
\usepackage{textcomp}
\usepackage{xcolor}
\usepackage{braket}
\usepackage{hyperref}
\usepackage{amsthm}
\usepackage{url}

\usepackage{amsmath}

\usepackage{graphicx}
\usepackage{subcaption}

\newcommand{\bfC}{\mathbf{C}}

\newcommand{\bfr}{\mathbf{r}}
\newcommand{\bfZ}{\mathbf{Z}}
\newcommand{\bfD}{\mathbf{D}}
\newcommand{\bfQ}{\mathbf{Q}}
\newcommand{\bfE}{\mathbf{E}}
\newcommand{\bfF}{\mathbf{F}}
\newcommand{\bfH}{\mathbf{H}}

\newcommand{\bfX}{\mathbf{X}}
\newcommand{\bfY}{\mathbf{Y}}

\newcommand{\bfA}{\mathbf{A}}
\newcommand{\bfT}{\mathbf{T}}

\newcommand{\bfg}{\mathbf{g}}

\newcommand{\bfB}{\mathbf{B}}
\newcommand{\bfI}{\mathbf{I}}

\newcommand{\bfM}{\mathbf{M}}
\newcommand{\bfR}{\mathbf{R}}
\newcommand{\bfK}{\mathbf{K}}

\newcommand{\bfv}{\mathbf{v}}
\newcommand{\bfP}{\mathbf{P}}

\newcommand{\cF}{\mathcal{F}}
\newcommand{\calD}{\mathcal{D}}

\DeclareMathOperator*{\argmin}{arg\,min}

\newcommand{\floor}[1]{\left\lfloor #1\right\rfloor}

\usepackage{algorithm}

\title{Communication-Efficient Approximate Gradient Coding} 

\author{Sifat Munim and Aditya Ramamoorthy%
\thanks{This paper was presented in part at the IEEE International Symposium on Information Theory (ISIT), Ann Arbor, MI, USA, 2025. This work was supported in part by the National Science Foundation (NSF) under grants CCF-2523473 and CCSS-2503640.\\
The authors are with the Department of Electrical and Computer Engineering, Iowa State University, Ames, IA 50011 USA (e-mail: \{smunim, adityar\}@iastate.edu).}
}

\begin{document}
\maketitle
\begin{abstract}
Large-scale distributed learning aims at minimizing a loss function $L$ that depends on a training dataset with respect to a $d$-length parameter vector. The distributed cluster typically consists of a parameter server (PS) and multiple workers. Gradient coding is a technique that makes the learning process resilient to straggling workers. It introduces redundancy within the assignment of data points to the workers and uses coding theoretic ideas so that the PS can recover $\nabla L$ exactly or approximately, even in the presence of stragglers. Communication-efficient gradient coding allows the workers to communicate vectors of length smaller than $d$ to the PS, thus reducing the communication time. While there have been schemes that address the exact recovery of $\nabla L$ within communication-efficient gradient coding, to the best of our knowledge the approximate variant has not been considered in a systematic manner. In this work we present constructions of communication-efficient approximate gradient coding schemes. Our schemes use structured matrices that arise from bipartite graphs, combinatorial designs and strongly regular graphs, along with randomization, and algebraic techniques and constraints. We derive analytical upper bounds on the approximation error of our schemes that are tight in certain cases. Moreover, we derive a corresponding worst-case lower bound on the approximation error of any scheme. For a large class of our methods, under reasonable probabilistic worker failure models, we show that, after appropriate rescaling, the expected value of the computed gradient equals the true gradient. This in turn allows us to establish convergence of the training to a stationary point. Numerical experiments corroborate our theoretical findings.
\end{abstract}

   

\section{introduction}
Large-scale distributed learning is the workhorse of modern-day machine learning (ML) algorithms. The sheer size of the data and the corresponding computational needs necessitate the usage of huge clusters for the purpose of parameter fitting in most ML problems of practical interest. Such problems are essentially a guided search over a very large space of parameters; the number of parameters can even be in the billions. Examples of such problems include the training of deep neural networks \cite{GoodBengCour16} that have made huge advances in speech and image recognition, and the training of large language models (LLMs) \cite{zhao2024surveylargelanguagemodels}.


At the top level, distributed machine learning operates by partitioning the relevant dataset into data subsets. The workers are assigned a subset of the data subsets and each worker is only responsible for processing its own data subsets. A prototypical example of this scenario is the training of deep neural networks. In this case each worker is responsible for computing the gradient (with respect to the parameter at the current iteration) of the appropriately defined loss function over its assigned data subsets. These gradients are then communicated to the central parameter server (PS) that coordinates the training, i.e., it aggregates the received gradients and determines the parameter for the next iteration. The PS broadcasts the new parameter and the iterations continue thereafter.

Distributed training is a key enabler of several ML technologies. Nevertheless, distributed clusters come with attendant challenges. A major issue is that large-scale clusters, especially those deployed within cloud platforms (e.g., Amazon Web Services (AWS)) are often heterogeneous in nature. These clusters often suffer from the problem of stragglers (slow or failed workers). Depending upon how the job is distributed amongst the workers, it is possible that the overall job execution time is limited by the speed of the slowest worker; this is clearly undesirable. For instance, the work of \cite{tandon_gradient} showed that stragglers can be up to $5\times$ slower than average workers on Amazon EC2. 
\begin{figure}[t!]
        \centering
        \includegraphics[scale = 1.0]{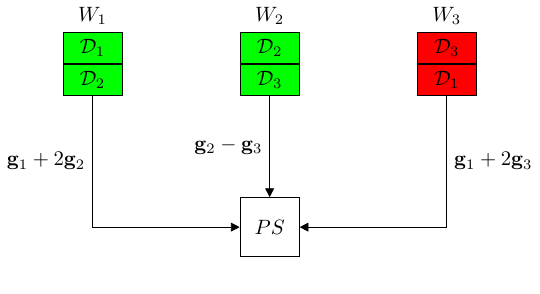}
        \caption{\label{fig:original_GC} {\small Worker $W_3$ has failed. $\mathbf{g}_i$ refers to the partial gradient on data subset $\calD_i$. Based on work completed by $W_1$ and $W_2$,  $\mathbf{g}_{1}+ \mathbf{g}_{2}+\mathbf{g}_{3}$ can be computed.}}
\end{figure}%

Distributing the dataset reduces the per-worker computational load. However, another relevant issue is the communication cost from the workers to the parameter server. In particular, the training process requires to-and-fro communication of high-dimensional vectors between the PS and the workers. When the number of parameters is very high, the communication cost of training can also be prohibitive.


Gradient coding, introduced in the work of \cite{tandon_gradient}, addresses worker slowdowns/failures by introducing redundancy within the assignment of data points to the workers (see Fig. \ref{fig:original_GC} where each data subset is replicated twice in the cluster). The workers transmit linear combinations of the gradients that are computed by them to the PS. As shown in Fig. \ref{fig:original_GC} an exact gradient coding solution allows the PS to precisely recover $\nabla L$ even in the presence of worker failures. It is straightforward to verify that if we insist on exact gradient recovery in the presence of any $s$ worker failures, then each data subset must be replicated at least $(s+1)$ times across the cluster. Thus, exact gradient recovery can be expensive. Moreover, in many parameter training scenarios, the exact gradient may not be needed for the convergence to the appropriate set of parameters, e.g., when the data distribution across the workers is i.i.d.

Accordingly, the approximate gradient coding variant \cite{dimakis_cyclic_mds} considers a setting where the gradient is recovered only approximately. The quality of the gradient is measured by its distance from the true gradient. The approximate gradient with rigorous guarantees on the gradient quality can typically be recovered even if the number of worker failures is much higher.

Communication-efficient gradient coding \cite{YeA18} considers an interesting point in the underlying design space. Specifically, it allows for trading off the redundancy in the data subset assignment for protecting against worker failures and also for communicating shorter length vectors from the workers to the PS. An example of communication-efficient gradient coding is illustrated in Fig. \ref{fig:comm_eff_GC}. However, we point out that most known constructions of communication-efficient gradient coding only consider the case of exact gradient recovery. We elaborate on these different variants of the gradient coding problem more formally in Section \ref{sec:bground}.

\begin{figure}[t!]        
        \centering
        \includegraphics[scale = 1.0]{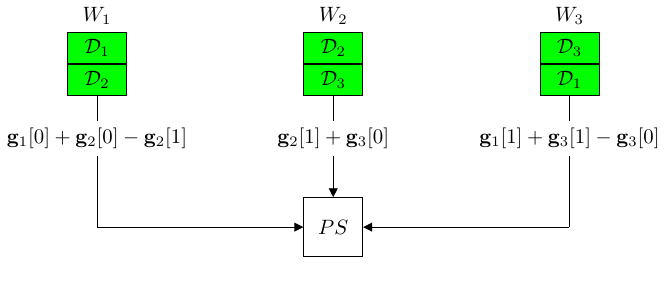}
        \caption{\label{fig:comm_eff_GC} {\small Gradient vectors are split into half \cite{YeA18}, i.e., $\mathbf{g}_{i} = [\mathbf{g}_{i}[0]^T~ \mathbf{g}_{i}[1]^T]^T$.  If all workers return their results, i.e., if there are no failures and worker $W_i$ sends $\tilde{\mathbf{g}}_i$ (labeled on the outgoing edge from $W_i$), the PS can calculate $\sum_{i=1}^3 \mathbf{g}_{i}[0] = \tilde{\mathbf{g}}_1 + \tilde{\mathbf{g}}_2$, and $\sum_{i=1}^3 \mathbf{g}_{i}[1] = \tilde{\mathbf{g}}_2 + \tilde{\mathbf{g}}_3$, i.e., the exact gradient. }}
\end{figure}

In this work, we propose approximate gradient coding techniques that are also communication-efficient. For our techniques, we provide upper bounds on the approximation error for a given communication reduction factor. Moreover, we provide a lower bound on the approximation error for any scheme with a fixed number of stragglers and communication reduction factor. Furthermore, under appropriate worker failure models, we establish a stationary-point convergence guarantee for the gradient descent algorithm. Our results are supported by numerical experiments. Table~\ref{tab:main_results} summarizes the proposed constructions and the main theoretical results.



\begin{table*}[t]
\centering
\caption{Summary of the proposed constructions and main theoretical results.}
\label{tab:main_results}
\renewcommand{\arraystretch}{1.25}
\setlength{\tabcolsep}{4pt}

\begin{tabularx}{\textwidth}{
    >{\raggedright\arraybackslash}p{3.0cm}
    >{\raggedright\arraybackslash}X
    >{\centering\arraybackslash}p{2.3cm}
    >{\centering\arraybackslash}p{2.1cm}
    >{\centering\arraybackslash}p{2.2cm}
}
\toprule
\textbf{Construction/result}
&
\textbf{Underlying structure or scope}
&
\textbf{Exact recovery when no workers straggle}
&
\textbf{Convergence guarantee}
&
\textbf{Analytical error guarantee}
\\
\midrule

\multicolumn{5}{l}{\textit{Proposed constructions}}
\\
\midrule

\multirow{3}{*}{%
\begin{tabular}[c]{@{}l@{}}
Post multiplication\\ by diagonal
matrices
\end{tabular}}
&
Balanced incomplete block design
&
No
&
Yes
&
Upper bound
\\
\cline{2-5}

&
Strongly regular graph
&
No
&
Yes, for the vertex-transitive case
&
Upper bound
\\
\cline{2-5}

&
Coset bipartite graph
&
Yes
&
Yes
&
Upper bound
\\
\midrule

Hadamard product with null-space constraints
&
Bi-regular bipartite graph or balanced incomplete block design
&
Yes
&
No
&
Upper bound
\\
\midrule

\multicolumn{5}{l}{\textit{General result}}
\\
\midrule

Worst-case lower bound
&
All communication-efficient approximate gradient coding schemes
&
---
&
---
&
Lower bound
\\
\bottomrule
\end{tabularx}
\end{table*}
The remainder of this paper is organized as follows. In Section \ref{sec:bground}, we present the relevant background on gradient coding, discuss related work, and summarize our main contributions. Section \ref{sec:prob_form} introduces the communication-efficient approximate gradient coding model and defines the approximation error metric. Section \ref{sec:structured_matrices} presents the structured assignment matrices used throughout the paper. Section \ref{sec:construction_RDM} presents our first method based on post-multiplication by diagonal matrices and derives corresponding upper bounds on the approximation error for several structured assignment matrices. Section \ref{sec:construction_RHPNSC} develops a second method based on Hadamard products with null-space constraints, which achieves exact recovery in the absence of stragglers and admits tractable error upper bounds. Section \ref{sec:lowerbound_proof} establishes a lower bound on the approximation error that applies to any communication-efficient approximate gradient coding scheme. Section \ref{sec:convergence} analyzes convergence of gradient descent under our proposed schemes and suitable straggler models. Finally, Section \ref{sec:numerical_exp} provides numerical experiments validating our theoretical bounds and comparing our constructions to baseline approaches.


\section{Background, Related Work and Summary of Contributions}
\label{sec:bground}
Let $\mathcal{D} = \{ (\mathbf{x}_i , y_i )\}_{i = 1}^N$ be a dataset where $(\mathbf{x}_i, y_i)$ are feature-label pairs. Let $L(\mathbf{w}) = \frac{1}{N} \sum_{i = 1}^N l(\mathbf{x}_i, y_i; \mathbf{w})$ be a loss function. Within learning, we wish to minimize $L$ with respect to the parameter $\mathbf{w} \in \mathbb{R}^d$ by gradient descent. Here, $l$ is a function that measures the prediction error for each point. Let $\mathbf{w}^{(t)}$ be the state of the parameter at iteration $t$. The parameter is updated as
\[ \mathbf{w}^{(t+1)} = \mathbf{w}^{(t)} - \eta^{(t)}\mathbf{g}^{(t)}, \]
where $\mathbf{g}^{(t)} := \frac{1}{N} \sum_{i = 1}^N \nabla l(\mathbf{x}_i, y_i, \mathbf{w}^{(t)})$ is the gradient of the loss function $L$ and $\eta^{(t)} \geq 0$ is the learning rate at iteration $t$.

The central goal is the computation of $\mathbf{g}^{(t)}$ in a distributed manner.
A distributed learning setup involves a PS and $n$ workers denoted $W_1, W_2, \dots, W_n$. The dataset $\mathcal{D}$ is divided into $k$ disjoint data subsets of equal size denoted by $\mathcal{D}_1, \mathcal{D}_2, \dots, \mathcal{D}_k$ which are distributed among the workers; each worker computes the gradients on its assigned data subsets. The PS receives the gradients calculated by the workers and aggregates them to find the overall gradient and hence the updated parameter.

Let $[a] \triangleq \{1,2, \dots, a\}$ for any $a \in \mathbb{N}$, and for a matrix $\bfM$, let $\text{supp}(\bfM)$ denote the set of indices of the nonzero entries of $\bfM$. 
For $i \in [n]$, let the computation load at worker $W_i$ be $\delta_i$, i.e., $W_i$ is assigned $\delta_i$ data subsets. For $j \in [k]$, the number of workers to which $\mathcal{D}_j$ is assigned is called its replication factor, denoted by $\gamma_j$. Throughout our work, we assume a regular assignment, i.e., $\delta_i = \delta$ for all $i \in [n]$, and $\gamma_j = \gamma$  for all $j \in [k]$. 



For a given scheme we define the $(n,k,\delta,\gamma)$ assignment matrix $\mathbf{A} \in \mathbb{R}^{k \times n}$ which is such that $\mathbf{A}(i,j) \ne 0$ \footnote{We use MATLAB notation at various places in this paper.} if and only if worker $W_j$ is assigned the data subset $\calD_i$. Also, note that the assignment matrix stays the same over the iterations. Moreover, our schemes for recovering $\nabla L$ will not be time dependent. Accordingly, henceforth we drop the dependence of the gradient on $t$. 

Let $\mathbf{g}_ i := \frac{1}{N} \sum_{(\mathbf{x}_j, y_j) \in \calD_i} \nabla l(\mathbf{x}_j, y_j; \mathbf{w})$ so that $\nabla L = \mathbf{g} = \sum_{i = 1}^k \mathbf{g}_i$. In the basic gradient coding protocol, worker $W_j$ calculates $\mathbf{g}_{i}$ for all $i \in \text{supp}(\mathbf{A}(:,j))$ (nonzero entries in the $j$-th column of $\mathbf{A}$) and linearly combines them to obtain 
 \begin{align}
     \tilde{\mathbf{g}}_j &= \sum_{i=1}^k \mathbf{A}(i,j) \mathbf{g}_{i}.\label{eq:linearly_comb_grad}
 \end{align}
It subsequently transmits $\tilde{\mathbf{g}}_j$ to the PS. 
The PS wants to decode $\nabla L$, so it picks a decoding vector $\mathbf{r}$ which is such that $\mathbf{r}(j) =0$ if worker $j$ is straggling (i.e., if $\tilde{\mathbf{g}}_j$ is unavailable) and calculates 
\begin{align}
\sum_{j=1}^n \mathbf{r}(j) \tilde{\mathbf{g}}_j &= \sum_{i=1}^k  \left( \sum_{j=1}^n   \mathbf{A}(i,j) \mathbf{r}(j) \right) {\mathbf{g}}_{i}, \label{eq:grad_cod_eq}
\end{align}
where we emphasize that the decoding vector $\mathbf{r}$ depends on the set of non-straggling workers. We now discuss three main threads of work within the gradient coding area.

\noindent {\bf Exact gradient coding:} 
Let $\mathbf{1}_a$ and $\mathbf{0}_a$ denote the all-ones and all-zeros column vector of size $a$ respectively. In exact gradient coding, we want to ensure that $\mathbf{A}\mathbf{r} = \mathbf{1}_k$ under any choice of at most $s$ stragglers, so that the PS can obtain $\mathbf{g}$. For exact gradient coding, each $\calD_i$ needs to be replicated at least $s+1$ times across the cluster.
The exact setting was the one that was discussed in the original work \cite{tandon_gradient} and some of the initial works in the area \cite{dimakis_cyclic_mds,reedsolomon_GC18,10.5555/3327345.3327412,tieredGC20,dynamicClusteringGC23,numerically_stableGC20}. 

\noindent {\bf Approximate gradient coding:} There are other scenarios in which the full gradient is not required, e.g., i.i.d. data distribution or too costly to work with because of the high replication factor needed. Indeed, many ML algorithms work with mini-batch SGD \cite{bottou_optimization} where the gradients are only calculated on a random subset of the data points. 
Thus, the fundamental problem within approximate gradient coding \cite{dimakis_cyclic_mds, charles2017approximate,wang2019fundamental} is to design the assignment matrix $\mathbf{A}$ such that $|| \mathbf{A} \mathbf{r}- \mathbf{1}_k||_2$ ($\ell_2$-norm) can be bounded as a function of the number of straggling workers by an appropriate choice of $\mathbf{r}$. Furthermore, one needs to understand if and under what conditions the iterative algorithm converges to a stationary point of the loss function $L(\cdot)$ in this setting. 




\noindent {\bf Communication-efficient gradient coding:} 
Within distributed training, a significant time cost is associated with the transmission of the computed gradients (vectors of length-$d$) by the workers to the PS, e.g., deep learning usually operates in the highly over-parameterized regime \cite{Belkin_2019} (many more parameters than data points). In real-world experiments, the work of \cite{alistarh2017qsgd} demonstrates that with increasing number of workers, the proportion of time spent on communication actually increases within distributed deep learning. 

For a specified assignment of data subsets to workers, if the number of failures is small, the extra redundancy can be judiciously used to transmit shorter vectors from the workers to the PS; this is known as communication-efficient gradient coding. 
For exact gradient coding, in the regular assignment setting if the replication factor is $s + m$, then the dimension of the transmitted vectors from the workers to the PS can be lowered to $d/ m$ \cite{YeA18} (see also \cite{tayyebehM21}) while still protecting against $s$ failures, i.e., one can trade off communication for computation. In this case, we call $m$ the communication reduction factor. Thus, if we know in advance that fewer failures are expected, the system can save communication time by operating in a communication-efficient mode.

\vspace{-0.1in}
\subsection{Discussion of Related Work} 
Ideas from coding theory have been the topic of intense investigation within the broad area of large-scale distributed computing within the past several years (see, e.g.,  \cite{yu2020straggler,saurav_g_analytics,DBLP:journals/corr/abs-1806-00939, LiMA16, dutta2016short, dutta_et_al20,tandon_gradient,dimakis_cyclic_mds}). 

Exact gradient coding was introduced in the work of \cite{tandon_gradient}, which showed that for regular assignments, the replication factor $\gamma \geq s+1$ and demonstrated constructions that met this bound. These include a scheme based on fractional repetition \cite{el2010fractional,olmezR16} that requires the number of workers $n$ to be a multiple of $(s+1)$, and a more general cyclic assignment-based scheme. Several variants of the exact gradient coding problem have been examined \cite{10.5555/3327345.3327412,tieredGC20,dynamicClusteringGC23,numerically_stableGC20,byzantine_GC23}. The work of \cite{dimakis_cyclic_mds} demonstrated the intimate link of exact gradient coding with coding theory (see also \cite{reedsolomon_GC18}) and introduced the formulation of approximate gradient coding. We note that the performance of the exact gradient coding schemes proposed in \cite{tandon_gradient} deteriorates in the approximate setting. For approximate gradient coding, \cite{dimakis_cyclic_mds} showed connections with the spectral properties of graphs and constructed schemes from expander graphs.  The work of \cite{charles2017approximate} (see also \cite{charles2018gradient}) observed that Ramanujan graphs (expanders with the largest spectral gap) only exist in restrictive settings and considered the usage of sparse random graphs instead. The work of \cite{glasgowW21} also presented graph based schemes in this setting. Connections of approximate gradient coding with block designs were considered in \cite{kadheKR19} and subsequently in \cite{soft_BIBD_GC22}. Using convex optimization-based techniques, \cite{11195507} derived an improved upper bound on the approximation error compared with \cite{dimakis_cyclic_mds}. Fundamental limits of approximate gradient coding were examined in \cite{wang2019fundamental}. Furthermore, several variants of the original gradient coding algorithm have been examined (see \cite{BitarSGC19, treeGC19, partial_recoveryGC23,10.1109/TIT.2024.3420222, krishnan2023sequential} among others).

Communication-efficient gradient coding was introduced in \cite{YeA18} (in the exact setting). The work of \cite{tayyebehM21} considered arbitrary assignment matrices and provided converse results. Crucially, both these works use polynomial interpolation in their solution. This leads to significant numerical instability \cite{Pan16} (see discussion in \cite{ramamoorthyMG24}) to the extent that their solution is essentially unusable for systems with twenty or more workers. This issue was considered in part in the exact setting by \cite{kadheKR20}; however, their solution works only for a restrictive assignment setting, i.e., it requires certain divisibility conditions on the problem parameters to work (similar to the fractional repetition approach of \cite{tandon_gradient}). While the work of \cite{tayyebehM21} does discuss extensions of their Lagrange Interpolation based idea to the approximate case, there are no numerical results reported in their work. 
Broadly speaking, the design of ``communication-efficient approximate gradient coding schemes'' is an open problem.

The work of \cite{ramamoorthyMG24} presented a different class of gradient coding protocols that leverage a small amount of feedback between the workers and the parameter server. This allows for more effective usage of slow (as against failed) workers while remaining numerically stable.

\subsection{Main contributions}
In this work, we present the first systematic approaches for the construction of communication-efficient approximate gradient coding schemes.

Our first method uses structured matrices that arise from (i) incidence matrices of combinatorial designs, (ii) adjacency matrices of strongly regular graphs, and (iii) bi-adjacency matrices of bipartite graphs. For the bipartite graph setting, we consider an algebraic construction that we refer to as coset bipartite graphs. These matrices are used to define the underlying assignment matrix. The first construction is obtained by vertically stacking $m$ copies of the assignment matrix after post-multiplication by diagonal matrices with random entries. For coset bipartite graphs, we also consider a deterministic construction based on Hadamard matrices. Our second method uses the Hadamard product with appropriate vectors that additionally satisfy certain null-space conditions.
Both methods allow the workers to transmit gradients of length $d/m$. The PS solves a least-squares problem to decode the approximate gradient from the non-straggling workers.

It is well recognized in the literature \cite{charles2017approximate} that analyzing approximate gradient coding schemes is challenging as one needs to analyze the corresponding least-squares solution over all possible straggler sets. For our constructions, we provide analytical upper bounds on the approximation error of our schemes that are tight in certain cases. In addition, for several of our constructions, we demonstrate that the condition numbers of the recovery matrices are low; this implies numerically stable decoding of the gradient. Furthermore, we derive a corresponding worst-case lower bound on the approximation error of any scheme. We also demonstrate that under reasonable worker failure models, after appropriate rescaling, the expected value of the computed gradient equals the true gradient. Together with a bounded-variance result, this allows us to show that the expected squared gradient norm at a randomly selected iterate converges to zero. We also present numerical experiments that corroborate our theoretical findings.


\section{Problem Formulation}
\label{sec:prob_form}
 
As we consider the communication-efficient variant in this work, we express it in terms of the formalism developed thus far.  We consider a scenario where the assignment of subsets to workers is fixed. Depending on the operating conditions, the cluster can decide to operate in a communication-efficient mode with communication reduction factor $m$. Thus, the communication-efficient schemes we consider in this work operate with a given assignment of subsets to workers. In Table~\ref{tab:notation}, we summarize the notations used in this paper. 

\begin{table}[!h]
\centering
\caption{Summary of the main notation.}
\label{tab:notation}
\renewcommand{\arraystretch}{1.12}
\begin{tabular}{@{}c p{0.73\columnwidth}@{}}
\hline
\textbf{Symbol} & \textbf{Description} \\
\hline
$n$ & Number of workers. \\

$k$ & Number of data subsets. \\

$s$ & Number of straggling workers. \\

$m$ & Communication reduction factor; each worker transmits a vector of length $d/m$. \\

$\gamma$ & Replication factor of each data subset. \\

$\delta$ & Computation load at each worker. \\

$\mathcal{F}$ & Index set of non-straggling workers, where $|\mathcal{F}|=n-s$. \\

$d$ & Dimension of the gradient vector. \\

$\mathbf{A}\in\mathbb{R}^{k\times n}$ &
Assignment matrix specifying the data subsets assigned to each worker. \\

$\mathbf{B}\in\mathbb{R}^{mk\times n}$ &
Encoding matrix used to generate the worker transmissions. \\

$\mathbf{D}_i\in\mathbb{R}^{n\times n}$ &
Random diagonal matrix used in the construction in~\eqref{eq:effAD1AD2}. \\

$\mathbf{g}_i\in\mathbb{R}^{d}$ &
Gradient corresponding to data subset $\mathcal{D}_i$. \\
$\mathbf{Z}\in\mathbb{R}^{\frac{d}{m}\times mk}$ &
Gradient block matrix formed from $\mathbf{g}_1,\ldots,\mathbf{g}_k$. \\

$\mathbf{R}\in\mathbb{R}^{n\times m}$ &
Decoding matrix used by the PS; its nonzero rows correspond to workers in $\mathcal{F}$. \\

$\mathbf{F}\in\mathbb{R}^{mk\times m}$ &
Target matrix corresponding to the $m$ blocks of the exact gradient. \\
\hline
\end{tabular}
\end{table}

Let $\bfA \in \mathbb{R}^{k \times n}$ be a $(n,k,\delta, \gamma)$ assignment matrix. For designing a scheme with communication reduction factor $m$, we define the encoding matrix denoted by $\mathbf{B} \in \mathbb{R}^{mk \times n}$ as follows.
\begin{align*}
\mathbf{B}^T = [\bfA_1^T ~|~ \dots ~|~ \bfA_m^T],
\end{align*}
where $\text{supp}(\bfA_i) \subseteq  \text{supp}(\bfA)$ for $i \in [m]$. This constraint ensures that the communication-efficient scheme operates with the same assignment of subsets to workers.

For $i \in [k]$, let $\bfg_i$ be partitioned into $m$ equal sized blocks (with zero-padding if required) so that $\bfg_i^T = [\bfg_i[1]^T ~|~ \dots ~|~\bfg_i[m]^T]$. Also, for $u \in [m]$, let $\mathbf{G}[u] \in \mathbb{R}^{\frac{d}{m} \times k}$ be such that $\mathbf{G}[u] := [\mathbf{g}_1[u] ~|~\mathbf{g}_2[u] ~|~ \dots ~|~\mathbf{g}_k[u]].$
Now define a matrix $\mathbf{Z} \in \mathbb{R}^{\frac{d}{m} \times mk}$ as
\[\mathbf{Z} := [ \mathbf{G}[1] ~|~ \mathbf{G}[2] ~|~\dots~|~ \mathbf{G}[m] ].\]


Worker $W_i$ then sends $\mathbf{Z}\mathbf{B}(:,i) \in \mathbb{R}^{\frac{d}{m}}$, once it has finished processing all its assigned subsets. Otherwise, it does not send anything. For $i \in [m]$, let $\mathbf{f}_i^T := [\underbrace{\mathbf{0}_k^T ~|~ \dots ~|~ \mathbf{0}_k^T}_{(i-1) \text{blocks}} ~|~ \mathbf{1}_k^T ~|~\underbrace{\mathbf{0}_k^T~|~ \dots ~|~ \mathbf{0}_k^T}_{(m-i) \text{blocks}}]$ of length $mk$. Next, let 
$\bfF := [\mathbf{f}_1 ~|~ \dots ~|~ \mathbf{f}_m]$. It can be observed that the exact gradient can be obtained from $\mathbf{ZF}$, since
\begin{align*}
\mathbf{ZF} &= [ \sum_{i = 1}^k \mathbf{g}_i[1] ~|~ \sum_{i = 1}^k \mathbf{g}_i[2] ~|~ \dots ~|~ \sum_{i = 1}^k \mathbf{g}_i[m]].
 \end{align*}

Suppose that there are $s$ stragglers, and let $\mathcal{F} \subseteq [n]$ be a set such that $i \in \mathcal{F}$ if and only if $W_i$ is not a straggler. Let $\mathbf{R} \in \mathbb{R}^{n \times m}$ be the decoding matrix such that $\mathbf{R}(:,i) = \mathbf{r}_i$ and $\text{supp}(\mathbf{r}_i) \subseteq \mathcal{F}$ for each $i \in [m]$. The PS then computes  $\mathbf{ZBR}$ and it follows that if $\mathbf{BR} = \mathbf{F}$, the PS computes the gradient exactly. For a matrix $\bfM$, let $\|\bfM\|_F$ and $\|\bfM\|_2$ be the Frobenius norm and the spectral norm of $\bfM$, respectively. 
It is not too hard to see that
\begin{align*}
\|\mathbf{ZBR} - \mathbf{ZF}\|_F^2 
& \leq \|\mathbf{Z}\|_2^2 \|\mathbf{BR-F}\|_F^2.
\end{align*}


Note that $\|\mathbf{Z}\|_2^2$ depends on the actual gradients and we do not have any control over them.
Thus, for a given set of non-stragglers corresponding to $\mathcal{F}$, we seek to minimize $\|\mathbf{BR-F}\|^2_F$. This motivates the following definition. 

\begin{defn}
For a given set of non-stragglers corresponding to $\cF \subseteq [n]$ of size $n-s$, and a given encoding matrix $\bfB$, the approximation error $\text{Err}_{\cF}(\bfB)$ is defined as
\begin{equation} \label{eq:ApproxError}
\text{Err}_{\cF}(\bfB) := \min\limits_{\substack{\mathbf{R}\in \mathbb{R}^{n \times m} \\ \text{supp}(\mathbf{R}(:,i)) \subseteq \mathcal{F},  \forall i \in [m]}}  {\|\mathbf{B} \mathbf{R}- \mathbf{F}\|_F^2}.
\end{equation}
\end{defn}

For a matrix $\mathbf{M} \in \mathbb{R}^{a \times b}$, let $\mathcal{H} \subseteq [a]$ and $\mathcal{K} \subseteq [b]$ be sets corresponding to the rows and columns of $\mathbf{M}$ respectively. We denote the submatrix of $\mathbf{M}$ with rows and columns corresponding to $\mathcal{H}$ and $\mathcal{K}$ by $\mathbf{M}(\mathcal{H}, \mathcal{K})$. Also, denote $\mathbf{M}_{\mathcal{H}} := \mathbf{M}(:, \mathcal{H})$. Now, for a given set of non-stragglers corresponding to $\mathcal{F}$, if $\mathbf{B}_{\mathcal{F}}^T\mathbf{B}_{\mathcal{F}}$ is invertible, then the RHS of \eqref{eq:ApproxError} can be expressed as
\begin{align}
&\sum_{i =1}^m \min\limits_{\substack{\mathbf{r}_i\in \mathbb{R}^{n} \\ \text{supp}(\mathbf{r}_i) \subseteq \mathcal{F} }}   \|\bfB\mathbf{r}_i - \mathbf{f}_i\|_2^2 
 = \sum_{i =1}^m \min\limits_{\substack{\mathbf{r}_i\in \mathbb{R}^{n-s} }}   \|\bfB_{\mathcal{F}}\mathbf{r}_i - \mathbf{f}_i\|_2^2 \nonumber\\
& = \sum_{i =1}^m\mathbf{f}_i^T \mathbf{f}_i - \mathbf{f}_i^T \bfB_{\cF}(\bfB_{\cF}^T\bfB_{\cF})^{-1}\bfB_{\cF}^T\mathbf{f}_i.\label{eq:opt_dec_error}
\end{align}
Here, the first equality follows from observation, and the second follows from an analysis of the least-squares error \cite{ekf}. The expression in \eqref{eq:opt_dec_error} characterizes the optimal decoding error. In practice, the PS can compute the decoding matrix by solving the corresponding least-squares problem using a standard QR or SVD based solver, without explicitly forming 
$(\mathbf{B}_{\mathcal{F}}^{T}\mathbf{B}_{\mathcal{F}})^{-1}$. 
Using a QR based solver, the decoding complexity is at most 
$\mathcal{O}\bigl(mk(n-s)^2\bigr)$ \cite{golub2013matrix}. 
The numerical stability of this decoding step depends on the condition number of 
$\mathbf{B}_{\mathcal{F}}^{T}\mathbf{B}_{\mathcal{F}}$. For a positive definite matrix $\bfM$, its condition number is $\kappa(\bfM) = \frac{\lambda_{max}(\bfM)} {\lambda_{min}(\bfM)}$ \cite{horn2012matrix}. Condition number bounds on the system of equations encountered in the decoding process for several of the proposed schemes are provided in the forthcoming sections.


The following example illustrates the encoding, worker transmissions, and least-squares recovery for a small communication-efficient scheme.

\begin{example}
Consider $k=2$ data subsets and $n=4$ workers with
\[
\mathbf{A}=
\begin{bmatrix}
1&0&1&0\\
0&1&0&1
\end{bmatrix},
\qquad m=2.
\]
Thus, $W_1,W_3$ are assigned $\mathcal{D}_1$, while $W_2,W_4$ are
assigned $\mathcal{D}_2$. Partition
$\mathbf{g}_i^T=[\mathbf{g}_i[1]^T\mid\mathbf{g}_i[2]^T]$, $i\in[2]$,
so that
\[
\mathbf{Z}
=
[\mathbf{g}_1[1]\mid\mathbf{g}_2[1]\mid
 \mathbf{g}_1[2]\mid\mathbf{g}_2[2]].
\]
Choose
\[
\mathbf{A}_1=\mathbf{A},
\qquad
\mathbf{A}_2=
\begin{bmatrix}
1&0&-1&0\\
0&1&0&-1
\end{bmatrix}.
\]
Since $\operatorname{supp}(\mathbf{A}_1)
=\operatorname{supp}(\mathbf{A}_2)
=\operatorname{supp}(\mathbf{A})$, the encoding matrix is
\[
\mathbf{B}
=
\begin{bmatrix}
\mathbf{A}_1\\
\mathbf{A}_2
\end{bmatrix}
=
\begin{bmatrix}
1&0&1&0\\
0&1&0&1\\
1&0&-1&0\\
0&1&0&-1
\end{bmatrix}.
\]
The workers transmit, respectively,
\[
\begin{aligned}
W_1&:\mathbf{g}_1[1]+\mathbf{g}_1[2],&
W_2&:\mathbf{g}_2[1]+\mathbf{g}_2[2],\\
W_3&:\mathbf{g}_1[1]-\mathbf{g}_1[2],&
W_4&:\mathbf{g}_2[1]-\mathbf{g}_2[2],
\end{aligned}
\]
where each transmitted vector has length $d/2$. Suppose that $W_4$ is a straggler, so
$\mathcal{F}=\{1,2,3\}$. Then,
\[
\mathbf{F}=
\begin{bmatrix}
1&0\\
1&0\\
0&1\\
0&1
\end{bmatrix},
\qquad
\mathbf{R}^{\star}
=
\argmin_{\mathbf{R}\in\mathbb{R}^{3\times2}}
\|\mathbf{B}_{\mathcal{F}}\mathbf{R}-\mathbf{F}\|_F^2
=
\frac{1}{2}
\begin{bmatrix}
1&1\\
1&1\\
1&-1
\end{bmatrix}.
\]
Consequently, the recovered blocks are
\[
\begin{aligned}
\widehat{\mathbf{g}}[1]
&=\mathbf{g}_1[1]
+\frac{1}{2}\bigl(\mathbf{g}_2[1]+\mathbf{g}_2[2]\bigr),\\
\widehat{\mathbf{g}}[2]
&=\mathbf{g}_1[2]
+\frac{1}{2}\bigl(\mathbf{g}_2[1]+\mathbf{g}_2[2]\bigr).
\end{aligned}
\]
The PS concatenates these blocks to obtain the approximate gradient.
\end{example}

\begin{remark}
\label{remark:opt_vs_fixed}
    As pointed out in \cite{charles2017approximate}, the analysis of the approximate gradient coding error is challenging, since the error expression involves the inverse of $\bfB_{\cF}^T\bfB_{\cF}$. This quantity needs to be bounded over all possible $\cF \subset [n], |\cF| = n-s$. For the case of $m=1$, prior work \cite{dimakis_cyclic_mds} has aimed at upper bounding this error by using ``fixed-decoding'' approaches, where the decoding vector is determined in advance, rather than by solving a least-squares problem. This typically results in loose upper bounds. The work of \cite{glasgowW21} proposed schemes that use optimal decoding (again for $m=1$). In this work, our goal is not only to design schemes for $m > 1$ that have low approximation error, but also to provide techniques to analytically or numerically upper bound the least-squares error ({\it cf.} \eqref{eq:opt_dec_error}). 
\end{remark}

\section{Structured Assignment Matrices}

\label{sec:structured_matrices}

To facilitate the analysis of the approximate decoding error, we will focus our attention on certain classes of assignment matrices. In particular, we will need to use structured matrices that arise from different areas such as graph theory and combinatorial design theory (see \cite{godsil-royle-algebraic},\cite{Brouwer_SRG}, \cite{StinsonBook}). In this section, we briefly review the relevant notions. In what follows, let $\mathbf{I}_a$ and $\mathbf{J}_a$ be the $a \times a$ identity matrix and all ones matrix, respectively. Also, let $\mathbf{0}_{a \times a}$ be the $a \times a$ all zeros matrix.

\begin{defn} {[{\it Balanced Incomplete Block Design (BIBD) and its Incidence Matrix}]} Let $(X, \mathcal{A})$ be a pair where $X$ is a set of elements called points and $\mathcal{A}$ is a collection of nonempty subsets of $X$ called blocks. A $(n,k, \gamma, \delta, \lambda)$- balanced incomplete block design (BIBD) is a pair $(X, \mathcal{A})$ that satisfies the following properties: {\it (i)} $X$ has size $n$, $\mathcal{A}$ has size $k$, each block in $\mathcal{A}$ has $\gamma$ points, and every point in $X$ is contained in exactly $\delta$ blocks. {\it (ii)} Any pair of distinct points is contained in exactly $\lambda$ blocks. The incidence matrix of a BIBD is a $n \times k$ binary matrix $\mathbf{E}$ such that $\mathbf{E}(i,j) = 1$ if $x_i \in A_j$ and $\mathbf{E}(i,j) = 0$ otherwise. 
     
\end{defn}

Let $\bfE$ be the incidence matrix of a $(n,k,\gamma, \delta, \lambda)$ BIBD and $\bfM = \bfE^T$. Then,  $\mathbf{M} \mathbf{1}_n = \gamma \mathbf{1}_k$ and $\mathbf{M}^T \mathbf{1}_k = \delta \mathbf{1}_n$. In addition, the inner product, $\langle \mathbf{M}(:,i), \mathbf{M}(:,j)\rangle = \lambda$, for $i \ne j$ and $i, j \in [n]$ (see  \cite{StinsonBook}). Consequently, $\bfM^T\bfM = (\delta - \lambda)\bfI_n + \lambda \mathbf{J}_n$ and, 
\begin{equation} \label{eq: bibd_prop}
\bfM_{\cF}^T\bfM_{\cF} = (\delta - \lambda)\mathbf{I}_{n-s} + \lambda \mathbf{J}_{n-s},
\end{equation}
since $|\mathcal{F}| = n-s$.

\begin{defn} {[{\it Strongly Regular Graph (SRG)}]} Let $G$ be a graph that is neither complete nor empty. Then, $G$ is a strongly regular graph (SRG) with parameters $(n, \delta, \lambda, \mu)$ if it has $n$ vertices where each vertex has degree $\delta$, any pair of adjacent vertices has $\lambda$ common neighbors, and any pair of 
nonadjacent vertices has $\mu$ common neighbors. 
\end{defn}

If $\bfM$ is the adjacency matrix of a $(n, \delta, \lambda, \mu)$ SRG, then
$\mathbf{M} = \mathbf{M}^T$ and $\mathbf{M} \mathbf{1}_n = \delta \mathbf{1}_n$  and furthermore, $\mathbf{M}^2 = \delta \mathbf{I}_n + \lambda \mathbf{M} + \mu (\mathbf{J}_n - \mathbf{I}_n - \mathbf{M})$ (see \cite{Brouwer_SRG}). Consequently, 
\begin{equation}\label{eq: srg_prop}
\mathbf{M}_{\cF}^T\mathbf{M}_{\cF}  = \delta \mathbf{I}_{n-s} + \lambda \mathbf{M}(\cF, \cF) + \mu (\mathbf{J}_{n-s} - \mathbf{I}_{n-s} - \mathbf{M}(\cF, \cF)).
\end{equation}
\begin{defn} {[{\it Bi-regular Bipartite Graph and Bi-adjacency Matrix}]}
Let $G = (L \cup R, E)$ be a graph with $L = [k]$, $R = [n]$, $L \cap R = \emptyset$, where each vertex in $L$ has degree $\gamma$ and each vertex in $R$ has degree $\delta$ such that the edges only exist between vertices in $L$ and $R$. Then, $G$ is a $(n,k,\delta, \gamma)$ bi-regular bipartite graph. The bi-adjacency matrix of a bi-regular bipartite graph is a matrix $\mathbf{M}$ such that for $i \in [k]$, $j \in [n]$, $\mathbf{M}(i,j) = 1$ if $i$ and $j$ are adjacent and $\mathbf{M}(i,j) = 0$ otherwise.   
\end{defn}

If $\bfM$ is the bi-adjacency matrix of a $(n,k,\delta, \gamma)$ bi-regular bipartite graph, then $\mathbf{M} \mathbf{1}_n = \gamma \mathbf{1}_k$ and $\mathbf{M}^T \mathbf{1}_k = \delta \mathbf{1}_n$. Next, we discuss the construction of a special class of bi-regular bipartite graphs that we refer to as a coset bipartite graph.  

\subsection{Coset Bipartite Graph}
Let $k,m,\delta$ be positive integers. A $(k, m, \delta)$ coset bipartite graph is constructed as follows. Set $n = mk$. 
We construct a bi-regular bipartite graph $G = (L \cup R, E)$ with $|L| = k$ and $|R| = n = mk$.\\
\textit{Vertex Sets:}
Let $H$ be an order $m$ subgroup of $\mathbb{Z}_{mk}$ (group operation is addition modulo $mk$) such that $H = \{0, k, 2k, \dots, (m-1)k\}$. For $i \in \mathbb{Z}_{mk}$, the coset $i+H$ of $H$ is defined as $i + H = \{ (i + a) \bmod mk \ | \ a \in H\}$. Note that there are $k$ distinct cosets of $H$.
We define the vertex sets corresponding to the distinct cosets of $H$ as $L = \{\, i+H : i=0,1,\dots,k-1 \,\}$ and  
$R = \mathbb{Z}_{mk}$. So, $|L| = k$ and $|R| = mk$.\\
\textit{Edges:} Let $B$ be a subset of the set $\{0,1, \dots, k-1\}$ of size $\delta$. Let $S = \bigcup_{b \in B} (b+H)$. For $i \in \{0,1,\dots,k-1\}$, we place an edge between $i+H \in L$ and $x \in R$ if and only if
$x \in i+S$. We refer to the set $B$ as the generating set.\\
\textit{Bi-regularity:} Each right vertex $x$ lies in exactly one $H$-coset, and since $|B|=\delta$, it is adjacent
to exactly $\delta$ left vertices. Hence, each right vertex has degree $\delta$.
Because $S$ consists of $\delta$ full cosets of $H$, each left vertex has degree
$
\delta \cdot |H|=\delta \cdot m = m\delta.
$
Thus $G$ is bi-regular where each vertex in $L$ has degree $m\delta$ and each vertex in $R$ has degree $\delta$.

If $\bfM$ is the bi-adjacency matrix of a $(k,m,\delta)$ coset bipartite graph, then with $n=mk$, $\mathbf{M} \mathbf{1}_n = m \delta \mathbf{1}_k$ and $\mathbf{M}^T \mathbf{1}_k = \delta \mathbf{1}_n$.

\subsection{Existence and Parameter Feasibility}

The structured assignment matrices discussed above are not available for
arbitrary system parameters. For a $(n,k,\gamma,\delta,\lambda)$ BIBD,
the necessary parameter relations are $n\delta=k\gamma$ and
$\delta(\gamma-1)=\lambda(n-1)$. These conditions are not
sufficient. 

Similarly, the parameters of a $(n,\delta,\lambda,\mu)$ SRG must satisfy
$(n-\delta-1)\mu=\delta(\delta-\lambda-1)$, and the multiplicities of the
two eigenvalues other than $\delta$ must be nonnegative integers. Consequently, SRGs exist only for specific parameter tuples.

In contrast, a $(k,m,\delta)$ coset bipartite graph can be constructed for
any positive integers $k$ and $m$ and any $1\leq\delta\leq k$. However, its
system parameters necessarily satisfy $n=mk$ and $\gamma=m\delta$. In the sequel we need additional assumptions: $k=p^a$, where $p$ is prime and
$a\in\mathbb{Z}_{\geq 1}$, and that $p\nmid\delta$. They are used later to establish
the almost-sure invertibility of the corresponding encoding matrix. 

In summary, the
choice of the structured assignment matrix is constrained by the desired
system parameters, and all three structures need not be simultaneously
available under a common parameter setting.

\section{Construction based on Post Multiplication by diagonal matrices}
\label{sec:construction_RDM}



Let $\mathbf{A}$ be a $(n,k, \delta, \gamma)$ assignment matrix.  Let $\epsilon \in [0,1)$ and define the intervals $I_1 := [ 1-\epsilon  ,  1+\epsilon ]$, and 
$I_2 := [ -1-\epsilon , -1+\epsilon ]$. Let $I := I_1 \cup I_2$. 
For $i \in [m]$, let $\mathbf{D}_i \in \mathbb{R}^{n \times n}$ be diagonal matrices where the diagonal entries are distributed i.i.d. uniformly over $I$. We construct the encoding matrix $\mathbf{B}$ such that
\begin{equation}\label{eq:effAD1AD2}
\mathbf{B}^T = [ \bfD_1 \bfA^T ~|~ \bfD_2 \bfA^T ~|~ \dots ~|~ \bfD_m \bfA^T].
\end{equation}
For this construction, we upper-bound the expected approximation error for a given non-straggler set $\mathcal{F} \subseteq [n]$; throughout this section the expectation is taken over the randomness in the diagonal matrices $\mathbf{D}_i$ for $i \in [m]$.


\begin{example}
Consider the $(3,3,2,2,1)$ BIBD whose incidence matrix is
\[
\mathbf{M}
=
\begin{bmatrix}
1 & 1 & 0\\
1 & 0 & 1\\
0 & 1 & 1
\end{bmatrix}.
\]
We choose the assignment matrix $\bfA = \bfM^T$. Let $m=2$ and $\epsilon=0$. Consider the realization
\[
\mathbf{D}_1=\mathbf{I}_3,
\qquad
\mathbf{D}_2= \begin{bmatrix}
1 & 0 & 0\\
0 & -1 & 0\\
0 & 0 & 1
\end{bmatrix}.
\]
Using the construction in~\eqref{eq:effAD1AD2}, we obtain 

\[
\mathbf{B}^T
=
\begin{bmatrix}
(\mathbf{A}\mathbf{D}_1)^T & (\mathbf{A}\mathbf{D}_2)^T
\end{bmatrix}
=
\begin{bmatrix}
1 & 1 & 0 & 1 & 1 & 0\\
1 & 0 & 1 & -1 & 0 & -1\\
0 & 1 & 1 & 0 & 1 & 1
\end{bmatrix}.
\]

Since the diagonal entries of $\mathbf{D}_1$ and $\mathbf{D}_2$ are nonzero,
\[
\operatorname{supp}(\mathbf{A}\mathbf{D}_1)
=
\operatorname{supp}(\mathbf{A}\mathbf{D}_2)
=
\operatorname{supp}(\mathbf{A}),
\]
and hence the encoding matrix preserves the original data-subset assignment.
\end{example}


\begin{thm}  
\label{thm:random_diag}
For the construction as described above, let $\mathcal{F} \subseteq [n]$ be a set of size $n-s$ corresponding to non-stragglers and suppose that $\mathbf{B}_{\mathcal{F}}^T\mathbf{B}_{\mathcal{F}}  \in \mathbb{R}^{(n-s) \times (n-s)}$ is invertible. Furthermore, let $c := \mathbb{E}[X^2]\mathbb{E}[\frac{1}{X^2}]$, where $X$ is a random variable that is distributed uniformly  over $I$. Then, 
\begin{equation}
\label{eq:AD1AD2err}
 \mathbb{E}[\text{Err}_{\mathcal{F}}(\mathbf{B})] \le mk - m\mathbf{1}_k^T \bfA_{\cF} \left(\mathbf{K}\right)^{-1} \bfA_{\cF}^T\mathbf{1}_k, 
\end{equation} 
where $\mathbf{K}: = \mathbf{\Delta}_{\mathcal{F}} + c(m-1)\mathbf{\Delta}^{(\text{diag})}_{\mathcal{F}}$, $\mathbf{\Delta}_{\mathcal{F}} := \mathbf{A}_{\mathcal{F}}^T\mathbf{A}_{\mathcal{F}} \in \mathbb{R}^{(n-s) \times (n-s)}$ and  $\mathbf{\Delta}^{(\text{diag})}_{\mathcal{F}}$ is a diagonal matrix such that $\mathbf{\Delta}^{(\text{diag})}_{\mathcal{F}}(i,i) := \mathbf{\Delta}_{\mathcal{F}}(i,i)$ 
for $i \in [n-s]$. 
\end{thm}

\begin{proof}
For $i \in [m]$, let $\tilde{\mathbf{D}}_i$ be a submatrix of $\mathbf{D}_i$ such that $\tilde{\mathbf{D}}_i :=\mathbf{D}_i(\mathcal{F}, \mathcal{F})$. We have, 
\begin{align*}
\bfB_{\cF}^T &= [\tilde{\bfD}_1 \bfA_{\cF}^T ~|~ \tilde{\bfD}_2 \bfA_{\cF}^T ~|~ \dots ~|~\tilde{\bfD}_m\bfA_{\cF}^T], \text{~so that}\\
\bfB_{\cF}^T\bfB_{\cF} & =  \sum_{j = 1}^m \tilde{\bfD}_j\bfA_{\cF}^T\bfA_{\cF}\tilde{\bfD}_j = \sum_{j = 1}^m \tilde{\bfD}_j \mathbf{\Delta}_{\mathcal{F}} \tilde{\bfD}_j .
\end{align*}
For $i \in [m]$, observe that $\mathbf{f}_i^T\mathbf{f}_i = k$ and $\mathbf{B}_{\mathcal{F}}^T \mathbf{f}_i = \tilde{\mathbf{D}_i}\mathbf{A}_{\mathcal{F}}^T\mathbf{1}_k$. Therefore, 
\begin{align*}
&\mathbf{f}_i^T \mathbf{f}_i - \mathbf{f}_i^T \bfB_{\cF}(\bfB_{\cF}^T\bfB_{\cF})^{-1}\bfB_{\cF}^T\mathbf{f}_i\\
& = k - \mathbf{1}_k^T \bfA_{\cF} \tilde{\bfD}_i\Bigg(\sum_{j = 1}^m \tilde{\bfD}_j\mathbf{\Delta}_{\mathcal{F}}\tilde{\bfD}_j\Bigg)^{-1}\tilde{\bfD}_i \bfA_{\cF}^T \mathbf{1}_k\\
& = k - \mathbf{1}_k^T \bfA_{\cF} \Bigg(\tilde{\bfD}^{-1}_i\Bigg(\sum_{j = 1}^m \tilde{\bfD}_j\mathbf{\Delta}_{\mathcal{F}}\tilde{\bfD}_j\Bigg)\tilde{\bfD}^{-1}_i\Bigg)^{-1} \bfA_{\cF}^T \mathbf{1}_k\\
& = k -  \mathbf{1}_k^T \bfA_{\cF} \left( \mathbf{\Delta}_{\mathcal{F}}  + \sum_{\substack{j = 1 \\ j \neq i}}^m\tilde{\bfD}^{-1}_i\tilde{\bfD}_j\mathbf{\Delta}_{\mathcal{F}}\tilde{\bfD}_j\tilde{\bfD}^{-1}_i\right)^{-1}\bfA_{\cF}^T\mathbf{1}_k.
\end{align*}


Let
\begin{equation}\label{K_i}
\mathbf{K}_i := \mathbf{\Delta}_{\mathcal{F}}  +  \sum_{\substack{j = 1 \\ j \neq i}}^m\tilde{\bfD}^{-1}_i\tilde{\bfD}_j\mathbf{\Delta}_{\mathcal{F}}\tilde{\bfD}_j\tilde{\bfD}^{-1}_i.\end{equation}

Note that
\[
\mathbf{K}_i
= \tilde{\bf D}_i^{-1}
\left(\sum_{j=1}^m \tilde{\bf D}_j \mathbf{\Delta}_{\mathcal{F}} \tilde{\bf D}_j\right)
\tilde{\bf D}_i^{-1}
= \tilde{\bf D}_i^{-1}
(\mathbf{B}_{\mathcal{F}}^{\top}\mathbf{B}_{\mathcal{F}})
\tilde{\bf D}_i^{-1}.
\]

Since $\mathbf{B}_{\mathcal{F}}^{T}\mathbf{B}_{\mathcal{F}}$ is positive semidefinite and invertible (by assumption), it is positive definite. Also, since $\tilde{\bf D}_i$ is diagonal with nonzero entries, it is also invertible. Consequently,
$\mathbf{K}_i$ is positive definite.

Let $r_1, r_2 \in [m]$ and $r_1 \ne r_2$. Let $u,v \in [n-s]$.  Note that, 
\begin{align*}
&\mathbb{E}[\tilde{\bfD}^{-1}_{r_1}\tilde{\bfD}_{r_2} \mathbf{\Delta}_{\mathcal{F}} \tilde{\bfD}_{r_2} \tilde{\bfD}^{-1}_{r_1}](u,v) \\
&= \mathbb{E}[\mathbf{\Delta}_{\mathcal{F}}(u,v)  \frac{\tilde{\bfD}_{r_2}(u,u)\tilde{\bfD}_{r_2}(v,v)} {\tilde{\bfD}_{r_1}(u,u)\tilde{\bfD}_{r_1}(v,v)}].
\end{align*}
If $u = v$, then 
\begin{align*}
&\mathbb{E}[\mathbf{\Delta}_{\mathcal{F}}(u,v)  \frac{\tilde{\bfD}_{r_2}(u,u)\tilde{\bfD}_{r_2}(v,v)} {\tilde{\bfD}_{r_1}(u,u)\tilde{\bfD}_{r_1}(v,v)}]\\ 
&= \mathbb{E}[\mathbf{\Delta}_{\mathcal{F}}(u,u)  \frac{(\tilde{\bfD}_{r_2}(u,u))^2} {(\tilde{\bfD}_{r_1}(u,u))^2}]\\
&= {\mathbf{\Delta}_{\mathcal{F}}(u,u)}\mathbb{E}[(\tilde{\bfD}_{r_2}(u,u))^2] 
\mathbb{E}\left[\frac{1}{(\tilde{\bfD}_{r_1}(u,u))^2}\right]\\
&= c\mathbf{\Delta}_{\mathcal{F}}(u,u)  = c\mathbf{\Delta}^{(\text{diag})}_{\mathcal{F}}(u,u).
\end{align*}

The third equality holds since the diagonal entries of $\tilde{\bfD}_{r_2}$ and $\tilde{\bfD}_{r_1}$ are i.i.d. and uniformly distributed over the set $I$. Now, if $u \ne v$, then 
\begin{align*}
&\mathbb{E}[\mathbf{\Delta}_{\mathcal{F}}(u,v)  \frac{\tilde{\bfD}_{r_2}(u,u)\tilde{\bfD}_{r_2}(v,v)} {\tilde{\bfD}_{r_1}(u,u)\tilde{\bfD}_{r_1}(v,v)}]\\
& \overset{(a)}= \mathbf{\Delta}_{\mathcal{F}}(u,v)\mathbb{E}[\tilde{\bfD}_{r_2}(u,u)]\mathbb{E}[\tilde{\bfD}_{r_2}(v,v)]\\
& \quad \quad \quad \quad \quad \quad \mathbb{E}\left[\frac{1}{\tilde{\bfD}_{r_1}(u,u)}\right] 
\mathbb{E}\left[\frac{1}{\tilde{\bfD}_{r_1}(v,v)}\right]\\
& \overset{(b)}= \mathbf{\Delta}_{\mathcal{F}}(u,v) \cdot 0 = 0 =c\mathbf{\Delta}^{(\text{diag})}_{\mathcal{F}}(u,v).
\end{align*}
Here, $(a)$ holds since the diagonal entries of $\tilde{\bfD}_{r_1}$ and $\tilde{\bfD}_{r_2}$ are i.i.d. and $(b)$ holds since the expected value of each diagonal entry is zero, and the expected value of the reciprocal of the diagonal entry is also zero. Thus, 
\begin{align*}
\mathbb{E}[\tilde{\bfD}^{-1}_{r_1}\tilde{\bfD}_{r_2} \mathbf{\Delta}_{\cF} \tilde{\bfD}_{r_2} \tilde{\bfD}^{-1}_{r_1}] 
= c\mathbf{\Delta}^{(\text{diag})}_{\mathcal{F}}.
\end{align*}

Using linearity of expectation, 
\begin{align*}\label{E[k_i]}
\mathbb{E}[\mathbf{K}_i] &= \mathbb{E}\left[ \mathbf{\Delta}_{\mathcal{F}}  +  \sum_{\substack{j = 1 \\ j \neq i}}^m\tilde{\bfD}^{-1}_i\tilde{\bfD}_j\mathbf{\Delta}_{\mathcal{F}}\tilde{\bfD}_j\tilde{\bfD}^{-1}_i \right]\\
& = \mathbf{\Delta}_{\mathcal{F}} + c(m-1) \mathbf{\Delta}^{(\text{diag})}_{\mathcal{F}} = \mathbf{K}.
\end{align*}
Finally,  \nopagebreak[4]


\begin{align*}
&\mathbb{E}[\text{Err}_{\mathcal{F}}(\mathbf{B})] \\
& \overset{(c)}{=} \mathbb{E} \left[ \sum_{i =1}^m\mathbf{f}_i^T \mathbf{f}_i - \mathbf{f}_i^T \bfB_{\cF}(\bfB_{\cF}^T\bfB_{\cF})^{-1}\bfB_{\cF}^T\mathbf{f}_i \right]\\
& \overset{(d)}{=} \sum_{i =1}^m \mathbb{E} \left[\mathbf{f}_i^T \mathbf{f}_i - \mathbf{f}_i^T \bfB_{\cF}(\bfB_{\cF}^T\bfB_{\cF})^{-1}\bfB_{\cF}^T\mathbf{f}_i \right]\\
& \overset{(e)}{=}  \sum_{i =1}^m \mathbb{E} \left[k -  \mathbf{1}_k^T \bfA_{\cF} \mathbf{K}^{-1}_i\bfA_{\cF}^T\mathbf{1}_k\right]\\
& \overset{(f)}{=}  \sum_{i =1}^m k -  \mathbf{1}_k^T \bfA_{\cF} \mathbb{E} \left[\mathbf{K}^{-1}_i\right]\bfA_{\cF}^T\mathbf{1}_k\\
&\overset{(g)}{\le}  \sum_{i =1}^m k -  \mathbf{1}_k^T \bfA_{\cF} \left(\mathbb{E} \left[\mathbf{K}_i\right]\right)^{-1}\bfA_{\cF}^T\mathbf{1}_k\\
& \overset{(h)}{=} \sum_{i = 1}^m  k - \mathbf{1}_k^T \bfA_{\cF} \left(\mathbf{K}\right)^{-1} \bfA_{\cF}^T\mathbf{1}_k\\
& = mk - m\mathbf{1}_k^T \bfA_{\cF} \left(\mathbf{K}\right)^{-1} \bfA_{\cF}^T\mathbf{1}_k.
\end{align*}
Here, $(c)$ follows from $\eqref{eq:opt_dec_error}$, $(d)$ follows by linearity of expectation. 
$(e)$ follows from $\eqref{K_i}$. $(f)$ follows since the expectation is over the random diagonal 
matrices $\mathbf{D}_i, i \in [m]$. $(g)$ follows since for any positive definite matrix 
$\mathbf{X}$, the operation $\mathbf{X} \mapsto \mathbf{X}^{-1}$ is matrix convex and thus, 
$\mathbb{E}[\mathbf{X}^{-1}] \succeq (\mathbb{E}[\mathbf{X}])^{-1}$ (see \cite{bhatia2009positive}). 
(For matrices $\mathbf{X}, \mathbf{Y}$, $\mathbf{X} \succeq \mathbf{Y}$ means $\mathbf{X}- \mathbf{Y}$ 
is positive semidefinite). $(h)$ follows since $\mathbb{E}\left[\mathbf{K}_i\right] = \mathbf{K}$.
\end{proof}

We now consider constructions where $\bfA$ is chosen as an instance of a structured matrix presented in Section \ref{sec:structured_matrices}.

\begin{remark}\label{rem:c_equals_one}
Note that, by the Cauchy--Schwarz inequality, $
c=\mathbb{E}[X^2]\mathbb{E}\left[\frac{1}{X^2}\right]
\geq \left(\mathbb{E}\left[X\left(\frac{1}{X}\right)\right]\right)^2=1
$. When $\epsilon=0$, $X$ takes values in $\{-1,1\}$ with equal probability, and hence $c=1$, which is the smallest possible value of $c$. 


Let $c_2\geq c_1\geq 1$, and let $\mathbf{K}_1$ and $\mathbf{K}_2$ denote the matrices obtained by substituting $c_1$ and $c_2$, respectively, for the parameter $c$ in the definition of $\mathbf{K}$ in Theorem \ref{thm:random_diag}. Since $\mathbf{\Delta}_{\mathcal{F}}=\mathbf{A}_{\mathcal{F}}^T\mathbf{A}_{\mathcal{F}}$ is positive semidefinite, its diagonal entries are nonnegative. Hence, $\mathbf{\Delta}_{\mathcal{F}}^{(\text{diag})}\succeq\mathbf{0}$. Then,
\[
\mathbf{K}_2-\mathbf{K}_1
=(c_2-c_1)(m-1)\mathbf{\Delta}_{\mathcal{F}}^{(\mathrm{diag})}
\succeq \mathbf{0}.
\]
Suppose $c_1$ and $c_2$ satisfy the invertibility assumption on $\mathbf{K}$ in Theorem ~\ref{thm:random_diag}. Then, $\mathbf{K}_1$ and $\mathbf{K}_2$ are positive definite. Therefore, the upper bound in \eqref{eq:AD1AD2err} increases with $c_i$.  Since $c\geq 1$, with equality when $\epsilon=0$, choosing $\epsilon=0$ yields the smallest upper bound.

We note that the applicability of Theorem \ref{thm:random_diag} requires $\mathbf{B}_{\mathcal{F}}^T 
\mathbf{B}_{\mathcal{F}}$ to be invertible. This condition may not hold for all values of $\epsilon$ depending on the assignment matrix $\mathbf{A}$. For the BIBD and SRG constructions considered below, 
$\mathbf{B}_{\mathcal{F}}^T \mathbf{B}_{\mathcal{F}}$ is invertible for $\epsilon = 0$ and every 
nonempty $\mathcal{F}$, as established in the proofs of Corollaries~1
and~2. Therefore, for these constructions, the choice $\epsilon = 0$ both satisfies 
the assumptions of Theorem \ref{thm:random_diag} and yields the smallest upper bound.

\end{remark}

 
\begin{cor} \label{cor:BIBD}
Suppose that the assignment matrix $\mathbf{A}$ is the transpose of the incidence matrix of a $(n,k,\gamma,\delta,\lambda)$ BIBD. 
Then, for the construction as in $\eqref{eq:effAD1AD2}$ with $\epsilon  = 0$, the expected error can be upper bounded as
\begin{equation} \label{eq:AD1AD2err_BIBD}
\mathbb{E}[\text{Err}_{\mathcal{F}}(\mathbf{B})] \le mk -   \frac{m\delta^2(n-s)} {m\delta + (n-s-1)\lambda}.
\end{equation}
\end{cor}
\begin{proof} 
By \eqref{eq: bibd_prop}, $\mathbf{\Delta}_{\mathcal{F}} = \mathbf{A}_{\mathcal{F}}^T\mathbf{A}_{\mathcal{F}} =  (\delta - \lambda) \mathbf{I}_{n-s} + \lambda \mathbf{J}_{n-s} $. Observe that $\mathbf{\Delta}_{\mathcal{F}}$ has eigenvalue $\delta - \lambda$ with multiplicity $n-s-1$ and eigenvalue $\delta - \lambda + \lambda(n-s)$ with multiplicity 1. Since $\delta  > \lambda$ for a BIBD \cite{StinsonBook}, the eigenvalues of $\mathbf{\Delta}_{\cF}$ are strictly positive. Since $\mathbf{\Delta}_{\mathcal{F}}$ is symmetric and all the eigenvalues are strictly positive, it is positive definite. Since by construction $\tilde{\bfD}_j$ is invertible, $\tilde{\bfD}_j \mathbf{\Delta}_{\cF} \tilde{\bfD}_j$ is positive definite for each $j \in [m]$.  Since the sum of positive definite matrices is positive definite,  it follows that $\mathbf{B}_{\mathcal{F}}^T\mathbf{B}_{\mathcal{F}}  =  \sum_{j = 1}^m \tilde{\bfD}_j\mathbf{\Delta}_{\mathcal{F}}\tilde{\bfD}_j$ is positive definite and hence invertible.
Now, $\mathbf{\Delta}^{(\text{diag})}_{\mathcal{F}} = \delta \mathbf{I}_{n-s}$. 
Hence $\mathbf{K} = \mathbf{\Delta}_{\mathcal{F}} + (m-1)\mathbf{\Delta}^{(\text{diag})}_{\mathcal{F}} = (m\delta - \lambda) \mathbf{I}_{n-s} + \lambda \mathbf{J}_{n-s} $. Also, note that $\mathbf{A}_{\mathcal{F}}^T\mathbf{1}_k = \delta \mathbf{1}_{n-s}$ in this case. Observe that $\mathbf{K}$ has eigenvalue $m\delta - \lambda$ with multiplicity $n-s-1$, and an eigenvalue $m\delta - \lambda + \lambda(n-s)$ with multiplicity 1 with corresponding eigenvector $\mathbf{1}_{n-s}$. 
Let $\mathbf{K} = \sum_{i = 1}^{n-s} \lambda_i \mathbf{v}_i\mathbf{v}_i^T$ (eigen-decomposition). Therefore,
\begin{align*}
\mathbf{K}^{-1} &= \sum_{i = 1}^{n-s}\frac{1}{\lambda_i}\mathbf{v}_i\mathbf{v}_i^T
= \frac{1}{m\delta + \lambda(n-s-1)} \frac{\mathbf{1}_{n-s}\mathbf{1}^T_{n-s}}{n-s} + \sum_{i = 2}^{n-s}\frac{1}{\lambda_i}\mathbf{v}_i\mathbf{v}_i^T.    
\end{align*}

Since $ \mathbf{A}_{\mathcal{F}}^T\mathbf{1}_k = \delta \mathbf{1}_{n-s}$, it follows that $\mathbf{1}_k^T \mathbf{A}_{\mathcal{F}} \mathbf{K}^{-1} \mathbf{A}_{\mathcal{F}}^T\mathbf{1}_k =  \frac{\delta^2(n-s)} {m\delta + (n-s-1)\lambda}$
since $\mathbf{v}_i^T \mathbf{1}_{n-s} = 0$, for $i \ne 1$ and $i \in [n-s]$. The result follows from \eqref{eq:AD1AD2err}.
\end{proof}
\begin{remark} When $\mathbf{A}$ is chosen as above, the work of \cite{kadheKR19} considered the case of $m=1$ and showed that for any set of non-straggling workers $\mathcal{F}$ of size $n-s$, $\text{Err}_{\mathcal{F}}(\mathbf{B}) = \text{Err}_{\mathcal{F}}(\mathbf{A}) = k - \frac{\delta^2(n-s)}{\delta +  (n-s-1)\lambda}$.

Our work is a generalization of their result for the communication-efficient case when $m > 1$. Indeed, a naive application of their scheme is to simply construct $\mathbf{B}$ by stacking $\mathbf{A}$ vertically $m$ times, i.e., 
\begin{align}
  \bfB^T &= [ \bfA^T~|~ \dots ~|~ \bfA^T].  \label{eq:trivial_cons}
\end{align}

In this case, their error expression becomes $\text{Err}_{\mathcal{F}}(\mathbf{B}) = mk - \frac{\delta^2(n-s)}{\delta + (n-s-1)\lambda}$ (follows from Theorem \ref{thm:random_diag}).
Our construction reduces the expected error to substantially below this.
\end{remark}



\begin{cor}
\label{cor:SRG}
Suppose that the assignment matrix $\mathbf{A}$ is the adjacency matrix of a $(n, \delta, \lambda, \mu)$ SRG with $0 <\delta \ne \mu$. 
Then, for the construction as in $\eqref{eq:effAD1AD2}$, with $\epsilon  = 0$, the expected error can be upper bounded as follows.
\begin{equation}
\mathbb{E}[\text{Err}_{\mathcal{F}}(\mathbf{B})] \le mk -   \frac{m\delta^2(n-s)} {(m\delta - \mu)+ \mu(n-s) + (\lambda - \mu)\theta},
\end{equation}
where
\begin{align*}
    \theta = \begin{cases}
        \delta & \text{~if~} \lambda \geq \mu,\\
        \frac{1}{2} \left[ (\lambda - \mu) - \sqrt{(\lambda - \mu)^2 + 4 (\delta - \mu)} \right] & \text{~otherwise.}
    \end{cases}
\end{align*}

\end{cor}

\begin{proof}

It is known that $\mathbf{A}$ has only three distinct eigenvalues given by $\delta , \tilde{r}, \tilde{s}$ \cite{Brouwer_SRG} where 
\[ 
\tilde{r} = \frac{1}{2}[ \lambda - \mu + \sqrt{(\lambda - \mu)^2 + 4(\delta - \mu)}],\]
and \[ 
\tilde{s} =  \frac{1}{2}[ \lambda - \mu - \sqrt{(\lambda - \mu)^2 + 4(\delta - \mu)}].\]

Since $0<\delta \ne \mu$, all the eigenvalues of $\mathbf{A}$ are nonzero.
Consequently, $\mathbf{A}^T\mathbf{A} = \mathbf{A}^2$ will have all strictly positive eigenvalues. Since $\mathbf{\Delta}_{\mathcal{F}} = \bfA_{\mathcal{F}}^T \bfA_{\mathcal{F}} = (\mathbf{A}^T\mathbf{A})(\mathcal{F}, \mathcal{F})$, by Cauchy's interlacing theorem (Theorem 4.3.17 in \cite{horn2012matrix}),  $\mathbf{\Delta}_{\mathcal{F}}$ will have all  strictly positive eigenvalues. Since $\mathbf{\Delta}_{\mathcal{F}}$ is symmetric and all the eigenvalues are strictly positive, it is positive definite. Since by construction $\tilde{\bfD}_j$ is invertible, $\tilde{\bfD}_j \mathbf{\Delta}_{\cF} \tilde{\bfD}_j$ is positive definite for each $j \in [m]$. Since the sum of positive definite matrices is positive definite, it follows that $\mathbf{B}_{\mathcal{F}}^T\mathbf{B}_{\mathcal{F}}  =  \sum_{j = 1}^m \tilde{\bfD}_j\mathbf{\Delta}_{\mathcal{F}}\tilde{\bfD}_j$ is positive definite and is hence invertible. By \eqref{eq: srg_prop},
\begin{align*}
\mathbf{\Delta}_{\mathcal{F}} &= \mathbf{A}_{\mathcal{F}}^T \mathbf{A}_{\mathcal{F}}\\
& = (\delta - \mu) \mathbf{I}_{n-s} + \mu \mathbf{J}_{n-s} + (\lambda - \mu) \mathbf{A}(\mathcal{F}, \mathcal{F}).
\end{align*}

Therefore, 
\begin{align*}
\mathbf{\Delta}^{(\text{diag})}_{\mathcal{F}} & = \delta \mathbf{I}_{n-s}, \text{~and}\\
\mathbf{K}& = \mathbf{\Delta}_{\mathcal{F}} + (m-1) \mathbf{\Delta}^{(\text{diag})}_{\mathcal{F}} = (m\delta-\mu) \bfI_{n-s} + \mu \mathbf{J} _{n-s} + (\lambda - \mu) \bfA(\mathcal{F},\mathcal{F}).
\end{align*}


Since both $\mathbf{\Delta}_{\mathcal{F}}$ and $\mathbf{\Delta}^{(\text{diag})}_{\mathcal{F}}$ are positive definite, it follows that $\mathbf{K}$ is positive definite and hence invertible. 
Note that the eigenvalues of $(m\delta-\mu) \bfI_{n-s} + \mu \mathbf{J} _{n-s}$ are $m\delta-\mu + \mu (n-s)$ (with multiplicity 1) and $m\delta-\mu$ (with multiplicity $n-s-1$). Also, by Cauchy's interlacing theorem, the maximum eigenvalue of $\mathbf{A}(\mathcal{F}, \mathcal{F})$, $\lambda_{\text{max}}(\mathbf{A}(\mathcal{F}, \mathcal{F})) \le \delta$ and the minimum eigenvalue of $\mathbf{A}(\mathcal{F}, \mathcal{F})$, $\lambda_{\text{min}}(\mathbf{A}(\mathcal{F}, \mathcal{F})) \ge \tilde{s}$. Now, by applying Weyl's inequality (Theorem 4.3.1 in \cite{horn2012matrix}) to $(m\delta-\mu) \bfI_{n-s} + \mu \mathbf{J} _{n-s}$ and $(\lambda - \mu) \bfA(\mathcal{F},\mathcal{F})$, we get
\begin{align*}
    \lambda_{\max} (\mathbf{K}) &\leq m\delta-\mu + \mu (n-s) + (\lambda - \mu) \theta \label{eq:upper_bd_lmax};
\end{align*}
where $\theta = \delta$ if $\lambda \ge \mu$ and otherwise $\theta = \tilde{s} = \frac{1}{2}[ \lambda - \mu - \sqrt{(\lambda - \mu)^2 + 4(\delta - \mu)}]$

Consequently, 
\[\lambda_{\text{min}}(\mathbf{K}^{-1}) \ge \frac{1}{ m\delta-\mu + \mu (n-s) + (\lambda - \mu) \theta}. \]

Note that, $\mathbf{A}_{\mathcal{F}}^T\mathbf{1}_k = \delta \mathbf{1}_{n-s}$. Let $ \mathbf{u} := \frac{1}{\sqrt{n-s}}\mathbf{1}_{n-s}$ so that $\|\mathbf{u}\|_2 = 1$. Then, 
 $\mathbf{u}^T \mathbf{K}^{-1} \mathbf{u} \ge \lambda_{\text{min}}(\mathbf{K}^{-1})$ by Rayleigh Quotient Theorem (see Theorem 4.2.2 in \cite{horn2012matrix}).
 By Theorem \ref{thm:random_diag}, 

\begin{align*}
\mathbb{E}[\text{Err}_{\mathcal{F}}(\mathbf{B})] 
& \le  mk - m\delta\mathbf{1}_{n-s}^T (\mathbf{K})^{-1} \delta \mathbf{1}_{n-s}\\
& = mk - m\delta^2(n-s)\mathbf{u}^T(\mathbf{K})^{-1}\mathbf{u}\\
& \le mk - m\delta^2(n-s)\lambda_{\text{min}}(\mathbf{K}^{-1})\\
& \le mk - \frac{m\delta^2(n-s)}{m\delta-\mu + \mu (n-s) + (\lambda - \mu) \theta}. 
\end{align*}
\end{proof}

A desirable property of a scheme is that it recovers the exact gradient when there are no stragglers, i.e., when $s=0$. The construction in  \eqref{eq:effAD1AD2} with assignment matrices from BIBDs and SRGs has this property only when $m=1$. Indeed, exact recovery requires the existence of a decoding matrix $\mathbf{R}$ such that $\mathbf{B}\mathbf{R}=\mathbf{F}$. From the block structures of $\mathbf{B}$ and $\mathbf{F}$, for every $i \in [m]$, exact recovery requires
$
\mathbf{A}\mathbf{D}_i\mathbf{r}_i=\mathbf{1}_k
$ and $
\mathbf{A}\mathbf{D}_j\mathbf{r}_i=\mathbf{0}_k
$ for every $j\neq i$, $j\in [m]$. For the BIBD construction, $\mathbf{A}$ has full column rank, whereas for the SRG construction with $\delta\neq\mu$, $\mathbf{A}$ is invertible. Since each $\mathbf{D}_j$ is invertible, every matrix $\mathbf{A}\mathbf{D}_j$ has a trivial null space. Thus, when $m>1$, the condition $
\mathbf{A}\mathbf{D}_j\mathbf{r}_i=\mathbf{0}_k
$ for some $j\neq i$ implies $\mathbf{r}_i=\mathbf{0}_n$, which contradicts
$
\mathbf{A}\mathbf{D}_i\mathbf{r}_i=\mathbf{1}_k.
$ Thus, although the approximation error is significantly reduced compared to the naive scheme for general values of $s$, it remains strictly positive when $s=0$ and $m>1$.


We address this issue by considering assignment matrices given by the
bi-adjacency matrix of a $(k,m,\delta)$ coset bipartite graph.
A sufficient condition for exact gradient recovery in the absence of
stragglers is that the encoding matrix $\bfB$ be invertible.
The following proposition shows that this condition holds for coset
bipartite graphs with appropriate parameters.

\begin{prop}{\label{prop: cos_inv}}
Suppose that the assignment matrix $\bfA$ is the bi-adjacency matrix of a $(k, m, \delta)$ coset bipartite graph such that $k = p^a$, where $p$ is a prime, $a \in \mathbb{Z}_{\ge 1}$ and $p \nmid \delta$. Then, for the construction as in $\eqref{eq:effAD1AD2}$ with $\epsilon > 0$, the encoding matrix $\bfB$ is almost surely invertible. 
\end{prop}

\begin{proof}
By construction, for $i \in \{0, \dots, k-1\}$ and $x \in \{0 , \dots, mk-1\}$, $\bfA(i+1, x+1) = 1$ if and only if $x \in i+ S$. By definition, for $y \in \{0, \dots, k-1\}$, $y \in i+S$ if and only if $y + tk \in i + S$ for any $t \in \{0, \dots, m-1\}$. Consequently, $\bfA(i+1, y+1) = 1$ if and only if $\bfA(i+1, tk + y+1) = 1$. Therefore, $\bfA$ can be written as $m$ equal column blocks such that 
\[ \bfA = \begin{bmatrix} \bfC & \dots & \bfC \end{bmatrix},\]
 where $\bfC \in \mathbb{R}^{k \times k}$. 

Now we show that $\bfC$ is circulant. For $i,j \in \{0, \dots, k-1\}$, by definition $\bfC(i+1, j+1) = 1$ if and only if $j \in i + S$, which implies that $j = i + u \bmod mk$ with $u \in S$. This implies that
\begin{align*}
    j-i &= u \bmod mk\\
    j-i &= u \bmod k \text{~~(since $k ~|~ mk$)}.
\end{align*}
Next, since $u \in S$, there exists $b \in B$ and $h \in H$ such that $u = b+h$ and since $h$ is a multiple of $k$, we have $u = b \bmod k$. We conclude that $(j-i) \bmod k \in B$. The reverse implication is easy to see. We conclude that $\bfC(i+1,j+1)$ only depends on $(j-i) \bmod k$. Therefore, each row of $\bfC$ is obtained from the preceding row by a cyclic shift of one position to the right. Thus, $\bfC$ is a circulant matrix.

Note that in this case $n = mk$. For $t \in \{1,\dots,m \}$, let
\[
\bfX_t := \mathrm{diag}\big(x^{(t)}_1,\dots,x^{(t)}_n\big)
\]
be diagonal matrices whose entries are independent indeterminates. Define $\bfB_{\bfX}$ such that
\begin{equation}\label{eq:poly}
\mathbf{B}_{\bfX}^T := [ \bfX_1 \bfA^T ~|~ \bfX_2 \bfA^T ~|~ \dots ~|~ \bfX_m \bfA^T].
\end{equation}

Let $\text{det}(\bfM)$ denote the determinant of a matrix $\bfM$. We will show that $\text{det}(\bfB_X)$ is not an identically zero polynomial, so that the Schwartz-Zippel lemma \cite{zippel1979probabilistic} can be applied to conclude that the encoding matrix $\mathbf{B}$ is invertible almost surely. For $i \in [m]$,
let $\mathbf{e}_i \in \mathbb{R}^m$ be the $i$-th standard basis vector.
Define
\[
\bfX_i' := (\mathbf{e}_i \mathbf{e}_i^T) \otimes \bfI_k \in \mathbb{R}^{mk\times mk},
\]
where $\otimes$ denotes the Kronecker product. Let 
\begin{equation}
\mathbf{B}_{\bfX'}^T = [ \bfX_1' \bfA^T ~|~ \bfX_2' \bfA^T ~|~ \dots ~|~ \bfX_m' \bfA^T].
\end{equation}

Since $\bfX_i'$ selects the $i$-th $k\times k$ block and each block of $\bfA$
equals $\bfC$, we have $\mathbf{B}_{\bfX'} = \bfI_m \otimes \bfC$. So, the determinant of $\mathbf{B}_{\bfX'}$, $\text{det}(\mathbf{B}_{\bfX'}) = (\text{det} (\bfC))^m$.
We show that the $k\times k$ circulant matrix $\bfC$ is invertible when
$k=p^{a}$ for a prime $p$ and $p \nmid \delta$. Let the first row of $\bfC$ be the
indicator of a set $T\subseteq \{0, \dots, k-1 \}$ with $|T|=\delta$, where index $j \in \{0, \dots,k-1\}$ corresponds to the $(j+1)$-th column of $\bfC$ following our convention on the labeling of the group elements in $\mathbb{Z}_{k}$. Define the mask
polynomial,
\[
p_T(x)\;:=\;\sum_{j\in T} x^{j} \in \mathbb{Z}[x]~~ (\text{ring of polynomials with coefficients in $\mathbb{Z}$}).
\]
Let $\omega:=e^{-2\pi i/k}$. It is well known that for a circulant matrix the eigenvalues are
\begin{equation}\label{eq:eigs-circ}
\lambda_r \;=\; p_T(\omega^{r}) \;=\; \sum_{j\in T}\omega^{rj}, \qquad r=0,1,\dots,k-1.
\end{equation}
In particular, $\lambda_0=p_T(1)=\delta\neq 0$. Assume for contradiction that $\bfC$ is singular. Then $\lambda_r=0$ for some
$r\in\{1,\dots,k-1\}$, i.e.,
\begin{equation}\label{eq:vanish}
\sum_{j\in T} (\omega^{r})^{j}=0.
\end{equation}
Let $q:=\gcd(k,r)$ and set $t=\frac{k}{q}$. Then $\omega^{r}$ is a primitive $t$-th root
of unity. Since $k=p^{a}$ we have $t=p^{b}$ for some $1\le b\le a$. Let $\Phi_t(x)$ denote the $t$-th cyclotomic
polynomial, i.e., the minimal polynomial over $\mathbb{Q}$ of a primitive $t$-th root of unity. For $t=p^b$, we have $\Phi_{t}(x)=1+x^{p^{b-1}}+\cdots+x^{(p-1)p^{b-1}}$, hence
$\Phi_{t}(1)=p$.
Since $\omega^r$ is a primitive $t$-th root and $p_T(\omega^r)=0$, it follows that there exists $g(x) \in \mathbb{Q}[x]$ such that 
\[
\Phi_t(x)g(x) = p_T(x).
\]
Suppose $g(x)$ is a constant. Then, $g(x) = 1$, since $\Phi_t(x)$ and $p_T(x)$ are both monic polynomials. Now, $\delta = p_T(1) =\Phi_t(1)g(1) = p$, contradicting the
assumption $p\nmid \delta$. Thus, $g(x)$ cannot be a constant. Next, suppose that $g(x)$ is not a constant polynomial. Since the coefficients of $\Phi_t(x)$ are coprime, by Gauss's Lemma (see Theorem 17.14 in \cite{judson2021abstract}) there exists  $h(x)\in \mathbb{Z}[x]$ such that $\Phi_t(x)h(x) = p_T(x)$. So, $\delta = p_T(1) = \Phi_t(1) h(1) = p\cdot h(1)$, i.e., $p \mid \delta$ which contradicts the assumption $p \nmid \delta$. Hence $\lambda_r\neq 0$ for all $r$, so $\bfC$ has no zero eigenvalues and is invertible.  


Consequently, $\text{det}(\mathbf{B}_{\bfX'}) \ne 0$ and  $\text{det}(\bfB_X)$ is not an identically zero polynomial. Also, since for the encoding matrix $\bfB$ the diagonal entries of $\bfD_i$ are i.i.d. continuous (as $\epsilon > 0)$, it follows from Schwartz-Zippel lemma \cite{zippel1979probabilistic} that, for the encoding matrix $\mathbf{B}$ generated by this construction, 
\begin{align*}
\mathbb{P}(\text{det}(\mathbf{B}) \ne 0) &= 1.
\end{align*}


\end{proof}

We now turn to the error analysis. Using the invertibility of the encoding
matrix established in Proposition~\ref{prop: cos_inv}, we derive an
upper bound on the expected error.


\begin{cor}  \label{cor:bipartite}
Suppose that the assignment matrix $\bfA$ is the bi-adjacency matrix of a $(k,m,\delta)$ coset bipartite graph such that $k = p^a$, where $p$ is a prime, $a \in \mathbb{Z}_{\ge 1}$ and $p \nmid \delta$. Then, for the construction as in $\eqref{eq:effAD1AD2}$ with $\epsilon > 0$,
the expected error can be upper bounded as
\begin{equation} \label{eq:AD1AD2err_Coset}
\mathbb{E}[\text{Err}_{\mathcal{F}}(\mathbf{B})] \le mk - \frac{m\delta^2(mk-s)}{m\delta^2 + c(m-1)\delta}.
\end{equation}
\end{cor}
\begin{proof}

By Proposition \ref{prop: cos_inv}, $\bfB$ is almost surely invertible. Consequently, for any $\cF \subseteq [n]$ of size $n-s$, $\mathbf{B}_{\mathcal{F}}^T\mathbf{B}_{\mathcal{F}}$ is invertible and thus we can apply Theorem $\ref{thm:random_diag}$. 
Note that in this case $n = mk$. We now calculate the maximum eigenvalue of $\mathbf{\Delta}_{[n]} + c(m-1)\mathbf{\Delta}^{(\text{diag})}_{[n]}$. Note that $\mathbf{\Delta}_{[n]} = \bfA^T \bfA$ by definition and since each column of $\bfA$ has $\delta$ nonzero entries that are $1$, $\bfA^T \mathbf{1}_k = \delta \mathbf{1}_n$ and $\mathbf{\Delta}^{(\text{diag})}_{[n]} = \delta \mathbf{I}_n$. Also, $\bfA \mathbf{1}_n = m\delta \mathbf{1}_k$. We have, 

\begin{align*}
(\mathbf{\Delta}_{[n]} + c(m-1)\mathbf{\Delta}^{(\text{diag})}_{[n]})\mathbf{1}_n &= (\mathbf{A}^T\mathbf{A} + c(m-1)\delta \mathbf{I}_n) \mathbf{1}_n\\
&= \mathbf{A}^T\mathbf{A} \mathbf{1}_n + c(m-1) \delta \mathbf{1}_n\\
&= \mathbf{A}^T m\delta \mathbf{1}_k + c(m-1) \delta \mathbf{1}_n\\
& = m\delta^2 \mathbf{1}_n + c(m-1)\delta \mathbf{1}_n\\
&= (m\delta^2 + c(m-1) \delta) \mathbf{1}_n.
\end{align*}

So, $\mathbf{1}_n$ is an eigenvector of $\mathbf{\Delta}_{[n]} + c(m-1)\mathbf{\Delta}^{(\text{diag})}_{[n]}$ with eigenvalue $m\delta^2 + c(m-1)\delta$. Since $\mathbf{\Delta}_{[n]} + c(m-1)\mathbf{\Delta}^{(\text{diag})}_{[n]}$ is a matrix with non-negative entries, by a consequence of the Perron–Frobenius theorem for nonnegative matrices (see Theorem 8.3.4 in \cite{horn2012matrix}) 
$\lambda_{max}\left(\mathbf{\Delta}_{[n]} + c(m-1)\mathbf{\Delta}^{(\text{diag})}_{[n]}\right) = m\delta^2 + c(m-1)\delta$. 
By Cauchy's interlacing theorem (Theorem 4.3.17 in \cite{horn2012matrix}), 
\begin{align*}
\lambda_{max}\left(\bfK = \mathbf{\Delta}_{\mathcal{F}} + c(m-1)\mathbf{\Delta}^{(\text{diag})}_{\mathcal{F}}\right) \le m\delta^2 + c(m-1)\delta.
\end{align*}

For $m>1$, since $\mathbf{\Delta}_{\cF}$ is positive semidefinite and $c(m-1) \mathbf{\Delta}^{(\text{diag})}_{\cF} = c(m-1) \delta \bfI_{n-s}$ is positive definite, $\bfK$ is positive definite. For $m=1$, $\mathbf K=\mathbf{\Delta}_{\mathcal F}$, which is
positive definite since $\mathbf A_{\mathcal F}$ has full column rank.
Hence, $\mathbf K$ is invertible. Consequently, 
\begin{equation}\label{eq:lambda_min}
\lambda_{min}\left(\bfK^{-1}\right) \ge \frac{1}{m\delta^2 + c(m-1)\delta}.
\end{equation}

By Theorem \ref{thm:random_diag}, 
\begin{align*}
\mathbb{E}[\text{Err}_{\mathcal{F}}(\mathbf{B})] &\le mk - m\mathbf{1}_k^T \bfA_{\cF} \left({\mathbf{K}}\right)^{-1} \bfA_{\cF}^T\mathbf{1}_k\\
& = mk - m\delta^2(n-s)\frac{\mathbf{1}_{n-s}^T}{\sqrt{n-s}}\left({\mathbf{K}}\right)^{-1} \frac{\mathbf{1}_{n-s}}{\sqrt{n-s}}\\
&\le  mk - m\delta^2(n-s)\lambda_{min}\left({\mathbf{K}}^{-1}\right)\\
& \le mk - \frac{m\delta^2(n-s)}{m\delta^2 + c(m-1)\delta} = mk - \frac{m\delta^2(mk-s)}{m\delta^2 + c(m-1)\delta}.
\end{align*} 

Here, the second inequality follows from  Rayleigh Quotient Theorem (Theorem 4.2.2 in \cite{horn2012matrix}). The third inequality follows from $\eqref{eq:lambda_min}$.
\end{proof}

\remark{For the coset bipartite graph construction, $\epsilon>0$ ensures that the diagonal entries of $\mathbf{D}_i$ are continuously distributed, which is needed to establish almost-sure invertibility in Proposition~1; consequently, $c>1$. Under the parameter settings of Proposition~1, the encoding matrix is almost surely invertible, so the approximation error is zero in the absence of stragglers even when $m>1$, and the exact gradient can be recovered.}



\remark{The derived upper bounds illustrate the tradeoff between communication efficiency and approximation accuracy. For the BIBD and SRG constructions, with the remaining parameters fixed, each upper bound is the difference between the linearly increasing term $mk$ and a term that converges to the finite limit $\delta(n-s)$ as $m$ increases. Consequently, the upper bounds increase approximately linearly with $m$ for large $m$, indicating a weaker approximation guarantee. This behavior is expected, since increasing $m$ reduces the length of the vector transmitted by each worker ($d/m$), thereby improving communication efficiency at the expense of approximation accuracy. For the coset bipartite graph based construction, the number of workers $n=mk$, so varying $m$ also changes the number of workers. However, for fixed $k$, $s$, $\delta$, and $c$ the corresponding upper bound also grows linearly with $m$ for large $m$, implying a similar communication--accuracy tradeoff.}

A direct comparison of the three upper bounds under identical system
parameters is generally not possible. In particular, the SRG based
construction uses an $n\times n$ assignment matrix and hence requires
$k=n$, whereas the coset bipartite graph based construction requires $n=mk$.
Therefore, for $m>1$, these two constructions cannot be instantiated
with common values of $n$ and $k$. Furthermore, a fair comparison would
require the computation load $\delta$ to be fixed, while the existence
conditions discussed in Section~IV-B may prevent a common admissible
parameter choice across all three structures. Accordingly, a direct
parameter-matched comparison of all three upper bounds is generally not
possible.

\subsection{Discussion about condition number of the recovery system of equations}
For the BIBD and SRG based constructions, explicit upper bounds on $\kappa(\mathbf{B}_{\mathcal{F}}^T\mathbf{B}_{\mathcal{F}})$ are provided in Appendix \ref{sec:appendix_cond_number}. These bounds show that in these cases the condition number is low for many practical cases of interest. Thus, we have numerically stable gradient recovery at the PS.

For the construction above based on coset bipartite graph, however, a similar condition number bound is hard to obtain when considering post-multiplication by random diagonal matrices. Instead for this case, we consider post-multiplication by diagonal matrices that are obtained from columns of Hadamard matrices.

\begin{defn} {[{\it Hadamard matrix}]}
A Hadamard matrix of order $\ell$ is a matrix
\[
\mathbf{H}_\ell \in \{\pm 1\}^{\ell\times \ell}
\]
whose rows are pairwise orthogonal. Equivalently,
\[
\mathbf{H}_\ell \mathbf{H}_\ell^{T} = \mathbf{H}_\ell^{T} \mathbf{H}_\ell = \ell I_\ell.
\]    
For $\ell>2$, a necessary condition for the existence of a Hadamard matrix
is that $\ell$ be divisible by $4$. Hadamard matrices are known to exist for many
orders, including every order which is a power-of-2.
\end{defn}
To obtain a construction based on coset bipartite graph for which the condition number can be bounded, consider a $(k,m,\delta)$ coset bipartite graph such that $k = p^a$, where $p$ is a prime, $a \in \mathbb{Z}_{\ge 1}$ and $p \nmid \delta$. Its bi-adjacency matrix can be written as $\mathbf{A}=[\mathbf{C}\ \cdots\ \mathbf{C}],$ where $\mathbf{C}\in\mathbb{R}^{k\times k}$ is circulant. Let $\mathbf{H}_m$ be the Hadamard matrix of order $m$, and define the encoding matrix,
\begin{equation} \label{eq:hada_coset_cons} \mathbf{B}=\mathbf{H}_m^T\otimes\mathbf{C}. \end{equation} 
For this construction, using the Cauchy interlacing theorem, we can conclude that $\kappa(\mathbf{B}_{\mathcal{F}}^T\mathbf{B}_{\mathcal{F}}) \le \kappa(\bfB^T\bfB)$. Furthermore, $\kappa(\bfB^T\bfB) = \kappa(\bfH_m\bfH^T_m) \kappa (\bfC^T\bfC) =  \kappa (\bfC^T\bfC)$.

\remark{For example, consider the $(27,2,5)$ coset bipartite graph with generating set $B = \{0,1,4,9,15\}$. For the corresponding circulant matrix $\mathbf{C}$,
$
\kappa(\mathbf{C}^T\mathbf{C}) \approx 110.49 $. Therefore, for the construction in \eqref{eq:hada_coset_cons},
$
\kappa(\mathbf{B}_{\mathcal{F}}^T\mathbf{B}_{\mathcal{F}})
\leq 110.5
$, for every non-straggler set $\mathcal{F}$.}

\section{Construction based on Hadamard product with null-space constraints}
\label{sec:construction_RHPNSC}

We now propose another construction based on bi-adjacency matrices of bipartite graphs or incidence matrices of BIBDs that achieves zero approximation error in the case of zero stragglers. 




In the following construction, let $\mathbf{A}$ be either the $k\times n$ bi-adjacency matrix of a bi-regular bipartite graph or the incidence matrix of a $(k,n,\gamma,\delta,\lambda)$ BIBD,\footnote{This differs from Section \ref{sec:construction_RDM}, where we considered a $(n,k,\gamma,\delta,\lambda)$ BIBD and took $\mathbf{A}$ to be the transpose of its $n\times k$ incidence matrix.} where $n=mk$. Thus, in the BIBD case, the design has $k$ points and $n$ blocks, and its incidence matrix has dimensions $k\times n$. Note that, for $m>1$, we have $k<n$, so $\mathbf{A}$ has more columns than rows.


Let $\bfC \circ \bfD$ denote the Hadamard product (element-wise product) of matrices $\bfC$ and $\bfD$ of the same dimension. At a top-level, the construction proceeds by first selecting a vector $\bfv_1 \in \mathbb{R}^n$ with no zero entries. Following this, we first construct a matrix $\bfA_1$ by setting its $i$-th row, $\bfA_1(i,:)$ to be the Hadamard product $\bfA(i,:) \circ \bfv_1$ and then normalizing it by $||\bfA_1(i,:)||_2^2$. Since $\bfA_1$ has a non-trivial null space (since $k < n$ for $m>1$), we then choose $\bfv_2$ from $\text{Null}(\bfA_1)$ with no zero entries and construct $\bfA_2$ in the same manner. Thus, at the $j$-th iteration, $j \leq m$ we choose $\bfv_j$ with no zero entries such that $ \bfv_j \in \text{Null}(\bfA_1) \cap \dots \cap \text{Null}(\bfA_{j-1})$ (the corresponding null space is nontrivial since $k(j-1) < n$) and construct $\bfA_j$ in a similar manner. At the end of this process, we set the encoding matrix $\bfB$ by vertically stacking the $\bfA_i$'s so that

\begin{equation}\label{eq:recons}
\bfB^T = [\bfA_1^T ~|~ \dots ~|~\bfA_m^T].    
\end{equation}

The requirement that every entry of $\mathbf{v}_j$ be nonzero ensures that
$
\left\|\mathbf{v}_j(\mathcal{H}_i)\right\|_2^2 > 0$ for all  $i \in [k]$, where $\mathcal{H}_i =  \text{supp}(\mathbf{A}(i,:))$. Consequently, the normalization used in the construction is well-defined. Moreover, this condition ensures that 
$\text{supp}(\mathbf{A}_i)
=
\text{supp}(\mathbf{A})$ for all $ i \in [m]$. At the $j$-th iteration, the vector $\mathbf{v}_j$ is selected from the null space of the matrix obtained by vertically stacking $\mathbf{A}_1,\ldots,\mathbf{A}_{j-1}$ (see  Algorithm  \ref{alg:example} for a formal description).

We emphasize that although the rank-nullity theorem guarantees that this null space is nontrivial, it does not, in general, guarantee the existence of a null-space vector with all entries nonzero. Accordingly, the construction is stated under the assumption that such $\bfv_i$'s can be found. There are examples of assignment matrices $\mathbf{A}$, where this can be guaranteed; we discuss this in the upcoming Proposition~\ref{prop:coset-full-support}.


\begin{example}
Consider a $(4,2,1,2)$ bi-regular bipartite graph. The bi-adjacency matrix is 
\[
\mathbf{A}=
\begin{bmatrix}
1&0&1&0\\
0&1&0&1
\end{bmatrix}.
\]
Choose $\mathbf{v}_1=[1\ 1\ 1\ 1]^T$. Algorithm~\ref{alg:example} gives
\[
\mathbf{A}_1=\frac{1}{2}
\begin{bmatrix}
1&0&1&0\\
0&1&0&1
\end{bmatrix}.
\]
Next, choose $\mathbf{v}_2=[1\ 1\ {-1}\ {-1}]^T$. Since
$\mathbf{A}_1\mathbf{v}_2=\mathbf{0}_2$, we have
$\mathbf{v}_2\in\operatorname{Null}(\mathbf{A}_1)$, and hence
\[
\mathbf{A}_2=\frac{1}{2}
\begin{bmatrix}
1&0&-1&0\\
0&1&0&-1
\end{bmatrix}.
\]
Therefore,
\[
\mathbf{B}
=
\begin{bmatrix}
\mathbf{A}_1\\
\mathbf{A}_2
\end{bmatrix}
=
\frac{1}{2}
\begin{bmatrix}
1&0&1&0\\
0&1&0&1\\
1&0&-1&0\\
0&1&0&-1
\end{bmatrix},
\]
with
$\text{supp}(\mathbf{A}_1)
=\text{supp}(\mathbf{A}_2)
=\text{supp}(\mathbf{A})$.
\end{example}

Our scheme allows for exact gradient recovery when $s=0$. Furthermore, we provide an upper bound on $\operatorname{Err}_{\mathcal{F}}(\mathbf{B})$ that can be computed numerically and analyzed for certain choices of the vectors $\mathbf{v}_j$. These results are summarized in the following lemma.

\begin{algorithm}[t!]
\caption{Algorithm to construct $\mathbf{B}$}
\label{alg:example}
\begin{algorithmic}[1] 
\REQUIRE Assignment matrix $\mathbf{A} \in \mathbb{R}^{k \times n}$ such that $n = mk$, communication reduction factor $m >1$.
\ENSURE Matrices $\mathbf{A}_1, \mathbf{A}_2, \dots, \mathbf{A}_m \in \mathbb{R}^{k \times n}$ such that for any $i \in [m]$, $\text{supp}(\mathbf{A}_i) = \text{supp}(\mathbf{A})$.
\STATE Initialize $\mathbf{A}_1, \dots, \mathbf{A}_m$ as zero matrices.  
\STATE Pick a vector $\mathbf{v}_1 \in \mathbb{R}^n$ such that $\mathbf{v}_1(u) \ne 0$ for every $u \in [n]$.
\FOR{$i = 1$ to $k$}
 \STATE $\mathcal{H} \gets \text{supp}(\mathbf{A}(i,:))$.
 \STATE $\mathbf{A}_1(i, \mathcal{H}) \gets \frac{(\mathbf{v}_1(\mathcal{H}))^T}{\|\mathbf{v}_1(\mathcal{H})\|_2^2}$.
\ENDFOR
\FOR{$j = 2$ to $m$}
\STATE $\mathbf{H} \gets \begin{bmatrix} \mathbf{A}_1 \\ \vdots \\ \mathbf{A}_{j-1}\end{bmatrix}$.
\STATE Pick a vector $\mathbf{v}_j \in \mathbb{R}^n$ such that $\mathbf{v}_j \in \text{Null}(\mathbf{H})$ and $\mathbf{v}_j(u) \ne 0$ for every $u \in [n]$. 
\FOR{$i = 1$ to $k$}
 \STATE $\mathcal{H} \gets \text{supp}(\mathbf{A}(i,:))$.
 \STATE $\mathbf{A}_j(i, \mathcal{H}) \gets\frac{(\mathbf{v}_j(\mathcal{H}))^T}{\|\mathbf{v}_j(\mathcal{H})\|_2^2}$.
\ENDFOR
\ENDFOR

\RETURN $\mathbf{A}_1, \mathbf{A}_2, \dots, \mathbf{A}_m$.
\end{algorithmic}
\end{algorithm}

\begin{lemma}
\label{lemma:second_cons}
The construction that uses Algorithm 1 recovers the exact gradient when $s = 0$. For a given set of non-stragglers corresponding to $\mathcal{F}$, let $\mathbf{\Sigma}_{\mathcal{F}} := \mathbf{B}_{\mathcal{F}}^T\mathbf{B}_{\mathcal{F}}$, and $\tilde{\mathbf{\Sigma}}_{\mathcal{F}}$ be a $(n-s) \times (n-s)$ diagonal matrix such that for $u \in [n-s]$, 
 $\tilde{\mathbf{\Sigma}}_{\mathcal{F}}(u,u) := \mathbf{\Sigma}_{\mathcal{F}}(u,u) + \sum_{j \ne u} |\mathbf{\Sigma}_{\mathcal{F}}(u,j)|$. Suppose that $\mathbf{\Sigma}_{\mathcal{F}}$ is invertible. Then,
\begin{equation}\label{eq:bounddiagdom}
\text{Err}_{\mathcal{F}}(\mathbf{B}) \le mk - \sum_{i = 1}^m \mathbf{1}_k^T {\mathbf{A}_i}_{\mathcal{F}} (\tilde{\mathbf{\Sigma}}_{\mathcal{F}})^{-1} {\mathbf{A}_i}_{\mathcal{F}}^T \mathbf{1}_k.  
\end{equation}
\end{lemma}

\begin{proof}
 Let $\mathbf{v}_1, \mathbf{v}_2, \dots, \mathbf{v}_m \in \mathbb{R}^n$ be vectors as generated by algorithm 1. For any $j \in [m]$ and $u\in [k]$, let $\text{supp}(\mathbf{A}_j(u, :)) = \mathcal{H}_u$. We emphasize that $\mathcal{H}_u$ is the same for all $\bfA_j$. Now, 

\begin{align*}(\mathbf{A}_j\mathbf{v}_j)(u) = \mathbf{A}_j(u,:)\mathbf{v}_j &= \mathbf{A}_j(u,\mathcal{H}_u)  \mathbf{v}_j({\mathcal{H}_u})\\
&= \frac{\mathbf{v}_j(\mathcal{H}_u)^T \mathbf{v}_j(\mathcal{H}_u)}{\|\mathbf{v}_j(\mathcal{H}_u)\|_2^2} = 1.
\end{align*}
So, $\mathbf{A}_j\mathbf{v}_j = \mathbf{1}_k$ for all $j \in [m]$. 
Note that by construction, $\bfv_j \in \cap_{i=1}^{j-1} \text{Null}(\bfA_i)$ (Step 9 of Algorithm \ref{alg:example}). Next we show that $\bfv_j \in \text{Null}(\bfA_{j+1}) \cap \dots \cap \text{Null}(\bfA_{m})$.

Consider $i$ such that $j < i \leq m$. Now,
\begin{align*}(\mathbf{A}_i\mathbf{v}_j)(u) = \mathbf{A}_i(u,:)\mathbf{v}_j &= \mathbf{A}_i(u,\mathcal{H}_u)  \mathbf{v}_j({\mathcal{H}_u})\\
&= \frac{\mathbf{v}_i(\mathcal{H}_u)^T \mathbf{v}_j(\mathcal{H}_u)}{\|\mathbf{v}_i(\mathcal{H}_u)\|_2^2}.
\end{align*}

We note here that by construction $\bfv_i \in \text{Null}(\bfA_j)$ so that  $\mathbf{A}_j\mathbf{v}_i = \mathbf{0}_k$. However, this means that for $u \in [k]$, we have $(\mathbf{A}_j\mathbf{v}_i)(u) = 0 \implies \frac{\mathbf{v}_j(\mathcal{H}_u)^T \mathbf{v}_i(\mathcal{H}_u)}{\|\mathbf{v}_j(\mathcal{H}_u)\|_2^2 } = 0$. This in turn implies that $\mathbf{v}_i(\mathcal{H}_u)^T \mathbf{v}_j(\mathcal{H}_u) = 0$ so that we can conclude that $\bfA_i \bfv_j = 0_k$ for $j < i \leq m$. It follows that $\bfB \bfv_i = \mathbf{f}_i$ for $i \in [m]$, i.e., the exact gradient recovery is possible when $s=0$.

Now we prove the upper bound on $\text{Err}_{\mathcal{F}}(\mathbf{B})$. Let $\mathbf{Q} := \tilde{\mathbf{\Sigma}}_{\mathcal{F}} - \mathbf{\Sigma}_{\mathcal{F}}$. $\bfQ$ is evidently Hermitian. Now, for any $u \in [n-s]$, we have
\begin{align*}
&\mathbf{Q}(u,u) = \tilde{\mathbf{\Sigma}}_{\mathcal{F}}(u,u) - \mathbf{\Sigma}_{\mathcal{F}}(u,u)\\ 
&= \mathbf{\Sigma}_{\mathcal{F}}(u,u) + \sum_{j \ne u} |\mathbf{\Sigma}_{\mathcal{F}}(u,j)| - \mathbf{\Sigma}_{\mathcal{F}}(u,u) \\
&=  \sum_{j \ne u} |\mathbf{\Sigma}_{\mathcal{F}}(u,j)|.
\end{align*}

On the other hand, for $j \ne u$ and $j \in [n-s]$,
\begin{align*}
\mathbf{Q}(u,j) &= -\mathbf{\Sigma}_{\mathcal{F}}(u,j) \text{, so that}\\
\sum_{j \ne u} |\mathbf{Q}(u,j)| &=  \sum_{j \ne u} |\mathbf{\Sigma}_{\mathcal{F}}(u,j)|.
\end{align*}

This shows that $\mathbf{Q}$ is a diagonally dominant matrix and hence positive semidefinite \cite{horn2012matrix}.


Thus, 
$\mathbf{Q} \succeq 0 \implies \tilde{\mathbf{\Sigma}}_{\mathcal{F}} - \mathbf{\Sigma}_{\mathcal{F}} \succeq 0 \implies \tilde{\mathbf{\Sigma}}_{\mathcal{F}} \succeq \mathbf{\Sigma}_{\mathcal{F}} \implies (\tilde{\mathbf{\Sigma}}_{\mathcal{F}})^{-1} \preceq (\mathbf{\Sigma}_{\mathcal{F}})^{-1}$ (see \cite{bhatia2009positive}). Finally,
\begin{align*}
\text{Err}_{\mathcal{F}}(\mathbf{B}) 
& = \sum_{i =1}^m \mathbf{f}_i^T \mathbf{f}_i - \mathbf{f}_i^T \bfB_{\cF}(\bfB_{\cF}^T\bfB_{\cF})^{-1}\bfB_{\cF}^T\mathbf{f}_i \\
& = mk -  \sum_{i =1}^m \mathbf{1}_k^T {\mathbf{A}_i}_{\mathcal{F}}(\mathbf{\Sigma}_{\mathcal{F}})^{-1}  {\mathbf{A}_i}_{\mathcal{F}}^T\mathbf{1}_k\\
& \le mk -  \sum_{i =1}^m \mathbf{1}_k^T {\mathbf{A}_i}_{\mathcal{F}}(\tilde{\mathbf{\Sigma}}_{\mathcal{F}})^{-1}  {\mathbf{A}_i}_{\mathcal{F}}^T\mathbf{1}_k.
\end{align*}
\end{proof}


\begin{remark}
    The matrix $\tilde{\mathbf{\Sigma}}_{\mathcal{F}}$ is a diagonal matrix and hence the bound in \eqref{eq:bounddiagdom} is easier to calculate as compared to calculating the actual least-squares error. The presented bound depends on the set $\cF$. However, it may be possible to arrive at weaker upper bounds that hold for all $\cF$ with $|\cF|=n-s$, by considering matrices $\tilde{\mathbf{\Sigma}}'_{\mathcal{F}} \succeq \tilde{\mathbf{\Sigma}}_{\mathcal{F}}$ and using them instead in the RHS of \eqref{eq:bounddiagdom}.
\end{remark}
It turns out that if we choose $\mathbf{A}$ as the bi-adjacency matrix of a $(k,m,\delta)$ coset bipartite graph, there is a choice of the $\mathbf{v}_i$'s so that the above construction is guaranteed to work. In particular, consider pairwise orthogonal vectors $\mathbf{q}_1, \dots, \mathbf{q}_m \in \mathbb{R}^{m}$ such that $\mathbf{q}_1 = \mathbf{1}_m$ and every entry of each $\mathbf{q}_j$ is nonzero. We choose $\mathbf{v}_i = \mathbf{q}_i \otimes \mathbf{1}_k$. This is stated as a proposition below; the proof appears in Appendix \ref{app:coset-full-support}. 
\begin{prop}
\label{prop:coset-full-support}
Suppose that $\mathbf{A}$ is the bi-adjacency matrix of a $(k,m,\delta)$ coset bipartite graph. Then, there exist vectors $\mathbf{v}_1,\ldots,\mathbf{v}_m\in\mathbb{R}^{mk}$ such that for $j \in [m]$ every entry of $\mathbf{v}_j$ is nonzero  and for $j \in \{ 2, \dots, m\}$,
\[
\mathbf{v}_j \in
\text{Null}\left(
\left[
\mathbf{A}_1^{T}\ \cdots\ \mathbf{A}_{j-1}^{T}
\right]^{T}
\right),\] 
where $\mathbf{A}_1,\ldots,\mathbf{A}_{j-1}$ are constructed from $\bfA$ according to Algorithm \ref{alg:example}.
\end{prop}

Moreover, the vectors $\mathbf{q}_1,\ldots,\mathbf{q}_m$ can be scaled such that
$\|\mathbf{q}_j\|_2^2=m$ for all $j\in[m]$, while preserving their orthogonality
and nonzero entries. Since $\mathbf{A}=[\mathbf{C}\ \cdots\ \mathbf{C}]$ and each
row of $\mathbf{C}$ has $\delta$ nonzero entries, for
$\mathbf{v}_j=\mathbf{q}_j\otimes\mathbf{1}_k$,
$
\|\mathbf{v}_j(\mathcal{H}_i)\|_2^2=m\delta,
$
for all $i\in[k]$ and $j\in[m]$. Thus, Algorithm~1 gives
$
\mathbf{A}_j=\frac{1}{m\delta}
\big[q_j(1)\mathbf{C}\ \cdots\ q_j(m)\mathbf{C}\big].
$
Let $\mathbf{Q}=[\mathbf{q}_1\ \cdots\ \mathbf{q}_m]$, then the resulting encoding
matrix can be written as
$
\mathbf{B}=\frac{1}{m\delta}\left(\mathbf{Q}^T\otimes\mathbf{C}\right).
$
When $k=p^a$, where $p$ is prime, $a\in\mathbb{Z}_{\geq 1}$, and
$p\nmid\delta$, the matrix $\mathbf{C}$ is invertible. Since
$\mathbf{Q}\mathbf{Q}^T=m\mathbf{I}_m$, it follows that
$
\kappa(\mathbf{B}_{\mathcal{F}}^T\mathbf{B}_{\mathcal{F}})
\leq
\kappa(\mathbf{B}^T\mathbf{B})
=
\kappa(\mathbf{C}^T\mathbf{C}).
$

\section{Lower Bound}
The following lower bound applies to any communication-efficient approximate gradient coding scheme. 
\begin{thm}
\label{thm:lowerbound}
For any communication-efficient approximate gradient coding encoding matrix $\mathbf{B} \in \mathbb{R}^{mk \times n}$ with communication reduction factor $m$, number of stragglers $s$, computation load at most $\delta$, number of workers $n$ and number of partitions $k$, there exists a non-straggler set $\mathcal{C} \subseteq [n]$ of size $n-s$ such that,
\begin{equation}\label{eq:lowerbound}
\text{Err}_{\mathcal{C}}(\mathbf{B})  
\ge \max_{u \in [m]} \min \left\{k,{\floor{\frac{k(s+m-u)}{n\delta}}} \right \}u.
\end{equation}


\end{thm}

\label{sec:lowerbound_proof}
\begin{proof}
Our proof leverages the basic ideas of \cite{dimakis_cyclic_mds}.
 Let $G = (\mathcal{W} \cup \mathcal{D}, \mathcal{E})$ be a bipartite graph where vertex set $\mathcal{W} = \{W_1 , \dots, W_n\}$ corresponds to the set of workers and vertex set $\mathcal{D} = \{ \mathcal{D}_1, \dots , \mathcal{D}_k\}$ corresponds to the set of data subsets. Also, let $(W_i,\mathcal{D}_j) \in \mathcal{E}$ if and only if the worker $W_i$ is assigned the data subset $\mathcal{D}_j$. For a vertex $v$ in graph $G$, let $\text{deg}(v)$ denote its degree. Without loss of generality, let us assume that $\deg(\mathcal{D}_1) \le \deg(\mathcal{D}_2) \le \dots \le \deg(\mathcal{D}_{k})$. Let $d_{j} := \frac{1}{j} \sum_{i = 1}^j \deg(\mathcal{D}_i)$, for any $j \in [k]$. Since each vertex in $\mathcal{W}$ has degree at most $\delta$, it follows that $d_1 \le d_2 \le \dots \le d_{k} \le \frac{n \delta}{k}$. For a vertex set $\mathcal{V}$ in $G$, let $N(\mathcal{V})$ denote the set of vertices that are adjacent to the vertices in $\mathcal{V}$. It follows that for $j \in [k]$, there exists a set of data subsets $\mathcal{P}_{j} \triangleq \{\mathcal{D}_1, \dots, \mathcal{D}_j\}$ of size $j$ for which $|N(\mathcal{P}_j)| \le \sum_{i = 1}^j \deg(\mathcal{D}_i) = jd_j \le {}\frac{jn \delta}{k}$. Fix $u \in [m]$,  and let
$
j_u:=\min\left\{k,
\left\lfloor\frac{k(s+m-u)}{n\delta}\right\rfloor\right\}$. Let $\mathcal{P}_0:=\emptyset$. Also, let $\mathcal{Q}_u:=\mathcal{P}_{j_u}$. Then,
$
|\mathcal{N}(\mathcal{Q}_u)|
\leq \frac{j_u n\delta}{k}
\leq s+m-u.
$
Thus, the data subsets in $\mathcal{Q}_u$ are assigned to at most
$s+m-u$ workers.

Let $\mathcal{S}$ be the set of stragglers, so $|\mathcal{S}| = s$. If $|N(\mathcal{Q}_u)| \ge s$, then choose $\mathcal{S}$ such that $\mathcal{S} \subseteq N(\mathcal{Q}_u)$. In this case, the data subsets in the set $\mathcal{Q}_u$ are assigned to at most $m-u$ non-stragglers. Otherwise, choose $\mathcal{S}$ such that $N(\mathcal{Q}_u)$ is contained in $\mathcal{S}$, i.e., $N(\mathcal{Q}_u) \subseteq \mathcal{S}$. In this case, the data subsets in the set $\mathcal{Q}_u$ are assigned to zero non-stragglers. In both cases, the data subsets in the set $\mathcal{Q}_u$ are assigned to at most $m-u$ non-stragglers. Let us denote the set of non-stragglers by $\mathcal{C} \triangleq [n] \setminus \mathcal{S}$.


For $i \in [k]$, denote $\mathcal{H}_i \triangleq \{ i, k+i, \dots,  (m-1)k + i \}$.
Define a matrix $\mathbf{B}^{(i)}$ such that $\mathbf{B}^{(i)} := \mathbf{B}(\mathcal{H}_i, \mathcal{C}) \in \mathbb{R}^{m \times (n-s)}$. 
 For any $\mathcal{D}_j\in\mathcal{Q}_u$, let $\text{rank}(\mathbf{B}^{(j)}) = \rho$. Since $\mathcal{D}_j$ is assigned to at most $m-u$ non-stragglers, at most $m-u < m$ columns of $\mathbf{B}^{(j)}$ will be nonzero. Hence, $\rho \le m-u$. For any $\mathbf{R} \in \mathbb{R}^{(n-s) \times m}$, $\text{rank}(\mathbf{B}^{(j)}\mathbf{R}) \le \min\{\text{rank}(\mathbf{B}^{(j)}), \text{rank}(\mathbf{R}) \} \le \rho$.
Let $\{\sigma_i = 1\}_{i = 1}^m$ be the singular values of $\mathbf{I}_m$. Thus, for every $\mathcal{D}_j\in\mathcal{Q}_u$,
\begin{align*} \min\limits_{\substack{\mathbf{R}\in \mathbb{R}^{(n-s) \times m} }}  {\|\mathbf{B}^{(j)} \mathbf{R}- \mathbf{I}_m\|_F^2} & \ge \min\limits_{\substack{\tilde{\mathbf{I}}_m\in \mathbb{R}^{m \times m} \\ \text{rank}(\tilde{\mathbf{I}}_m) \le \rho }}  {\|\tilde{\mathbf{I}}_m- \mathbf{I}_m\|_F^2} \\
& = \sum_{ i = \rho+1}^m \sigma_i^2 = m-\rho \ge u.
\end{align*}
The equality holds by the Eckart-Young theorem \cite{Eckart_Young_1936}. The second inequality holds since $\rho \le m-u$. 

 Finally, for a fixed $u \in [m]$,
\begin{align*}
 %
 \text{Err}_{\mathcal{C}}(\mathbf{B})  & =\min\limits_{\substack{\mathbf{R}\in \mathbb{R}^{(n-s) \times m}}} {\|\mathbf{B}_{\mathcal{C}} \mathbf{R}- \mathbf{F}\|_F^2}\\
 &= \min\limits_{\substack{\mathbf{R}\in \mathbb{R}^{(n-s) \times m} }}\sum_{i = 1}^k   {\|\mathbf{B}^{(i)} \mathbf{R}- \mathbf{I}_m\|_F^2}\\
 &\ge \sum_{i = 1}^k \min\limits_{\substack{\mathbf{R}\in \mathbb{R}^{(n-s) \times m} }}\  {\|\mathbf{B}^{(i)} \mathbf{R}- \mathbf{I}_m\|_F^2}\\
 &\ge \sum_{i:\mathcal{D}_i \in \mathcal{Q}_u} \min\limits_{\substack{\mathbf{R}\in \mathbb{R}^{(n-s) \times m} }}  {\|\mathbf{B}^{(i)} \mathbf{R}- \mathbf{I}_m\|_F^2}\\
 & \ge  j_u u.
\end{align*}
The second equality above holds from the argument in Appendix \ref{sec:appendix_claim_frob}. Since $u \in [m]$ was arbitrary, by maximizing over all $u \in [m]$, we have

\begin{align*}
 \text{Err}_{\mathcal{C}}(\mathbf{B}) \ge \max_{u \in [m]} \min\left\{k,
\left\lfloor\frac{k(s+m-u)}{n\delta}\right\rfloor\right\}u.
\end{align*}

\end{proof}



\begin{example}
Suppose that $k=n \ge 5$, $\delta = 4$, $m=2$ and $s = \delta +1 =5$. Then the above bound is achieved at $u=m=2$ and the lower bound is $2$. On the other hand, if $s = \delta-1 = 3$, then the bound is achieved at $u=1$ and takes the value $1$. Thus, the maximum in \eqref{eq:lowerbound} is achieved for different values of $u$ in different ranges of $s$.
\end{example}

\begin{remark}
The lower bound is associated with the worst-case straggler set for a given number of stragglers $s$. Consequently, the approximation error can be lower than this bound for other choices of straggler set $\mathcal{S}$ with $|\mathcal{S}| = s$. 
\end{remark}

\section{Convergence Analysis}
\label{sec:convergence}
Since our schemes propose approximations of the gradient, it is necessary to show that the gradient descent algorithm converges in these cases. In this section, we establish convergence for the construction as in $\eqref{eq:effAD1AD2}$ with BIBD, vertex-transitive SRG and coset bipartite graph assignment matrices, and for the construction as in $\eqref{eq:hada_coset_cons}$.

Let $\partial L(\mathbf{w}^{(t)})$ be the set of subgradients of the loss function $L$ at $\mathbf{w}^{(t)}$. In Stochastic Gradient Descent (SGD) (Section 14.3 in \cite{books/daglib/0033642}), the parameter is updated as 

\[ \mathbf{w}^{(t+1)} = \mathbf{w}^{(t)} - \eta^{(t)} {\mathbf{v}}^{(t)},\]

where $\mathbf{v}^{(t)}$ is a random vector such that $\mathbb{E}[ \mathbf{v}^{(t)} | \mathbf{w}^{(t)}] \in \partial L (\mathbf{w}^{(t)})$. When $L$ is differentiable, this condition reduces to $\mathbb{E}\!\left[\mathbf{v}^{(t)} \mid \mathbf{w}^{(t)}\right]
=
\nabla L\!\left(\mathbf{w}^{(t)}\right)$. In what follows, we show that, after an appropriate rescaling, the computed gradient satisfies this unbiasedness condition. Therefore, the resulting update can be viewed as an instance of SGD. We also establish that the variance of the rescaled computed gradient is bounded by a constant (assuming the spectral norm of $\bfZ$ is bounded), independent of the iteration number. Under this property, Corollary~2.2 in~\cite{doi:10.1137/120880811} provides a stationary-point convergence guarantee for SGD when the loss function is bounded below and has a Lipschitz-continuous gradient, with the learning rate chosen as specified, where the guarantee is stated for $\mathbf{w}^{(R_T)}$ with $R_T$ chosen uniformly at random from $\{1,\ldots,T\}$.





In this part we assume that the arrival of the computed result from each worker within a fixed time interval is modeled by a Bernoulli random variable \cite{DBLP:journals/corr/abs-1901-09671}. In particular, we will assume that within a fixed time interval, each worker sends its computed result independently with probability $1-q$ for $q \in [0, 1)$. Therefore, each worker straggles independently with probability $q$.


For the random diagonal matrix based construction, since in each iteration the encoding matrix varies based on the diagonal matrices $\mathbf{D}_i$, $i \in [m]$, the computed gradient is a function of random diagonal matrices $\mathbf{D}_i$. As before we let $\mathcal{F}$ denote the index set corresponding to non-straggling workers. 
%
As described in Section \ref{sec:prob_form}, for a given encoding matrix $\mathbf{B}$ and a given set of non-stragglers corresponding to $\mathcal{F}$, the computed gradient for our schemes is $\mathbf{Z} \mathbf{B}_{\cF}\mathbf{R}$, where $\bfR \in \mathbb{R}^{(n-s) \times m}$ is the optimal decoding matrix. Let $\mathbf{D} := (\mathbf{D}_1, \mathbf{D}_2, \dots, \mathbf{D}_m)$. 
For any function $\phi$ of $\mathbf{D}$, we use $\mathbb{E}_{\mathbf{D}}[\phi(\mathbf{D})]$ to denote the expectation with respect to the joint distribution of the diagonal entries of $\mathbf{D}_1,\ldots,\mathbf{D}_m$. For any function $\psi$ defined on $2^{[n]}$ (codomain can vary), define 




\[
\mathbb{E}_{\mathcal{F}}[\psi(\mathcal{F})]
:= \sum_{F \subseteq [n]} \psi(F)\, \mathbb{P}(\mathcal{F} = F)
= \sum_{i = 0}^n \ \sum_{\substack{|F| = i}}
\psi(F)\, q^{\,n-i}(1-q)^{\,i}.
\]
In the discussion below, we will show that $\mathbb{E}_{\mathbf{D}}[\mathbb{E}_{\mathcal{F}}[\mathbf{B}_{\cF} \mathbf{R}]]= \zeta \mathbf{F}$, so that the true gradient can be obtained by calculating $\frac{1}{\zeta}\mathbb{E}_{\mathbf{D}}[\mathbb{E}_{\mathcal{F}}[\mathbf{Z}\mathbf{B}_{\cF}\mathbf{R}]]$. We need the following basic definitions.
\begin{defn}{[{\it Permutation}]}
A permutation is a reordering of the elements of $[n]$, denoted by $\sigma : [n] \to [n]$.  
The corresponding permutation matrix
$\bfP_{\sigma} \in \{0,1\}^{n \times n}$ is such that for $i, j \in [n]$, $\bfP_{\sigma}(i,j) = 1$ if $i = \sigma(j)$ and $\bfP_{\sigma}(i,j) = 0$ otherwise. Note that for standard basis vector $\mathbf{e}_i \in \mathbb{R}^n$,  $ \mathbf{e}_{\sigma(i)}^T\bfP_{\sigma} = \mathbf{e}_i ^T$. Also, $\bfP_{\sigma}^{-1} = \bfP_{\sigma}^T$.      
\end{defn}


\begin{defn}{[{\it Graph Automorphism}]}
Let $G=([n], E)$ be a graph with adjacency matrix $\bfA_G$. An automorphism of $G$ is a permutation $\sigma$ (with corresponding matrix $\bfP_{\sigma}$) of the vertex set $[n]$ such that for all $i, j \in [n]$, $(i,j) \in E$ if and only if $(\sigma(i),\sigma(j)) \in E$. In this case, $\bfP_{\sigma} \bfA_G \bfP_{\sigma}^T = \bfA_G$. 
\end{defn}

\begin{defn}{[{\it Vertex Transitive Graph}]}
A graph $G = ([n],E)$ is vertex-transitive if for any $i, j \in [n]$, there exists an automorphism $\sigma$ such that $\sigma(i) = j$.
\end{defn}



\subsection{Unbiasedness for BIBDs}

The following lemma establishes that the computed gradient for the random diagonal matrix-based construction with BIBD assignment matrices is unbiased up to a nonzero scaling factor.

\begin{lemma}
\label{lemma:BIBD_unbiasedness}
Suppose that the assignment matrix $\mathbf{A}$ is the transpose of the incidence matrix of a $(n,k,\gamma,\delta,\lambda)$ BIBD. 
Then, for the construction as in $\eqref{eq:effAD1AD2}$ with $\epsilon  = 0$, we have 

\[ \mathbb{E}_{\mathbf{D}}[\mathbb{E}_{\mathcal{F}}[\mathbf{B}_{\cF} \mathbf{R}]]= \alpha \mathbf{F},\]
where $\alpha$ is a strictly positive constant. 

\end{lemma}

\begin{proof}



Fix $i \in [m]$. Let $\bfr_i* = \text{argmin}_{\mathbf{r}} \|\bfB_{\cF} \mathbf{r} - \mathbf{f}_i \|_2^2$. When $\cF$ is empty, we assume $\mathbf{B}_{\cF} = \mathbf{0}_{mk}$. Thus, $\mathbf{B}_{\cF}\mathbf{r}_i* = \mathbf{0}_{mk}$. For non-empty $\cF$, we have $\mathbf{B}_{\mathcal{F}}^T = [ \tilde{\bfD}_1 \bfA_{\mathcal{F}}^T ~|~ \tilde{\bfD}_2 \bfA_{\mathcal{F}}^T ~|~ \dots ~|~ \tilde{\bfD}_m\bfA_{\mathcal{F}}^T]$. Here, for $t \in [m]$, $\tilde{\mathbf{D}}_t := \mathbf{D}_t(\mathcal{F}, \mathcal{F})$. Let $\mathbf{\Delta}_{\mathcal{F}} := \mathbf{A}_{\mathcal{F}}^T \bfA_{\cF}$. Then, $\mathbf{B}_{\mathcal{F}}^T\mathbf{B}_{\mathcal{F}} = \sum_{t = 1}^m\tilde{\bf{D}}_t  \mathbf{\Delta}_{\cF} \tilde{\bf{D}}_t$. So, for non-empty $\mathcal{F}$ $(s \ne n)$ of size $n-s$,  
\begin{align*} 
\bfr_i* &=  (\bfB_{\cF}^T\bfB_{\cF})^{-1} \bfB_{\cF}^T  \mathbf{f}_i \overset{(a)}=  \tilde{\mathbf{D}}_i^{-1} \left(\mathbf{\Delta}_{\cF}  + \sum_{t = 1, t\ne i}^m\tilde{\bf{D}}_i^{-1}\tilde{\bf{D}}_t \mathbf{\Delta}_{\cF} \tilde{\bf{D}}_t\tilde{\bf{D}}_i^{-1}\right)^{-1}  \mathbf{A}_{\cF}^T \mathbf{1}_k,
\end{align*}
where
$(a)$ holds since $\tilde{\bfD}_i^{-1}\tilde{\bfD}_i = \mathbf{I}_{n-s}$. Let $\bfM_i := \mathbf{\Delta}_{\cF}  + \sum_{t = 1, t\ne i}^m\tilde{\bf{D}}_i^{-1}\tilde{\bf{D}}_t \mathbf{\Delta}_{\cF} \tilde{\bf{D}}_t\tilde{\bf{D}}_i^{-1}$. From the discussion in Corollary \ref{cor:BIBD}, we have that $\bfM_i$ is positive-definite. Now,

\begin{align*}
\mathbf{B}_{\mathcal{F}}\mathbf{r}_i* = \begin{bmatrix} \mathbf{A}_{\cF} \tilde{\bfD}_1\tilde{\mathbf{D}}_i^{-1} \\ \vdots\\ \bfA_{\cF}\\ \vdots \\  \mathbf{A}_{\cF} \tilde{\bfD}_m\tilde{\mathbf{D}}_i^{-1} \end{bmatrix} \bfM_i^{-1}\mathbf{A}_{\cF}^T \mathbf{1}_k.
\end{align*}
For $j \ne i$, $j \in [m]$, let 

\begin{equation}
f_j(\mathbf{D}_i)  := \mathbb{E}_{\mathbf{D}_t, t \in [m], t \ne i}[ \mathbf{A}_{\cF} \tilde{\bfD}_j\tilde{\bf{D}}_i^{-1}\bfM_i^{-1}\mathbf{A}_{\cF}^T \mathbf{1}_k].
\end{equation}

Note that  
$f_j(-\mathbf{D}_i) = - f_j(\mathbf{D}_i)$.
So, $f_j(\mathbf{D}_i)$ is an odd function of $\mathbf{D}_i$. Also, note that the distribution of $\mathbf{D}_i$ is symmetric, i.e., $\bfD_i \overset{d}= - \bfD_i$ (~$\overset{d}=$ denotes distributional equivalence). Consequently, for $j \ne i$, $j \in [m]$, we have $\mathbb{E}_{\mathbf{D}_i}[f_j(\mathbf{D}_i)] = \mathbf{0}_k$. Therefore, 
\begin{align*}\mathbb{E}_{\mathbf{D}} [ \mathbf{A}_{\cF} \tilde{\bfD}_j\tilde{\bf{D}}_i^{-1}\bfM_i^{-1}\mathbf{A}_{\cF}^T \mathbf{1}_k] &= \mathbb{E}_{\mathbf{D}_i}[\mathbb{E}_{\mathbf{D}_t, t \in [m], t \ne i}[ \mathbf{A}_{\cF} \tilde{\bfD}_j\tilde{\bf{D}}_i^{-1}\bfM_i^{-1}\mathbf{A}_{\cF}^T \mathbf{1}_k]]\\
&= \mathbb{E}_{\mathbf{D}_i}[f_j(\mathbf{D}_i)] \\
&=  \mathbf{0}_{k}. 
\end{align*}

Next we calculate $\mathbb{E}_{\bfD} [\mathbf{A}_{\cF}\bfM_i^{-1}\mathbf{A}_{\cF}^T \mathbf{1}_k]$. 
We have,
\begin{align*}
\mathbf{M}_i &= \mathbf{\Delta}_{\cF} + \sum_{t = 1, t\ne i}^m\tilde{\bf{D}}_i^{-1}\tilde{\bf{D}}_t \mathbf{\Delta}_{\cF} \tilde{\bf{D}}_t\tilde{\bf{D}}_i^{-1}\\
&= (\delta - \lambda) \mathbf{I}_{n-s} + \lambda \mathbf{1}_{n-s} \mathbf{1}_{n-s}^T + \sum_{t = 1, t\ne i}^m \tilde{\bf{D}}_i^{-1}\tilde{\bf{D}}_t  \left[ (\delta - \lambda) \mathbf{I}_{n-s} + \lambda \mathbf{1}_{n-s} \mathbf{1}_{n-s}^T \right]\tilde{\bf{D}}_t\tilde{\bf{D}}_i^{-1}\\
& = m(\delta - \lambda) \bfI_{n-s} + \lambda \mathbf{1}_{n-s} \mathbf{1}_{n-s}^T+ \lambda \sum_{t = 1, t\ne i}^m\mathbf{u}_{it} \mathbf{u}_{it}^T,
\end{align*}
where $\mathbf{u}_{it}:= \tilde{\bf{D}}_i^{-1}\tilde{\bf{D}}_t \mathbf{1}_{n-s}$. The final equality holds since for $\epsilon = 0$, the entries of $\tilde{\bfD_{j}}$ are $\pm 1$ and therefore $\tilde{\mathbf{D}}_j^{-1} =\tilde{\mathbf{D}}_j$ for all $j \in [m]$. 
Let $\mathbf{\Pi} \in \mathbb{R}^{(n-s) \times (n-s)}$ be a permutation matrix such that for $u, v \in [n-s]$,  $\mathbf{\Pi}^T \mathbf{e}_u = \mathbf{e}_v$. 
We have,
\begin{align*}
\mathbf{\Pi}\mathbf{M}_i\mathbf{\Pi}^T &= m(\delta - \lambda) \mathbf{\Pi}\mathbf{\Pi}^T + \lambda \mathbf{1}_{n-s} \mathbf{1}_{n-s}^T+ \lambda \sum_{t = 1, t\ne i}^m \mathbf{\Pi}\mathbf{u}_{it} \mathbf{u}_{it}^T \mathbf{\Pi}^T\\
&\overset{d}=\mathbf{M}_i.
\end{align*}
where the second equality above holds by $\mathbf{\Pi}\mathbf{\Pi}^T = \mathbf{I}_{n-s}$, $\mathbf{\Pi} \mathbf{1}_{n-s} = \mathbf{1}_{n-s}$, and since $\mathbf{u}_{it} \;\stackrel{d}{=}\; \mathbf{\Pi}\mathbf{u}_{it}$ because the entries of $\mathbf{u}_{it}$ are i.i.d. and take values $\pm 1$ with equal probability. Therefore, 
\begin{align*}
\mathbf{M}_i^{-1} \mathbf{1}_{n-s} &\overset{d}= (\mathbf{\Pi} \mathbf{M}_i \mathbf{\Pi}^T)^{-1} \mathbf{1}_{n-s} \\
& = \mathbf{\Pi} \mathbf{M}_i^{-1} \mathbf{1}_{n-s}.
\end{align*}
Next,
\begin{align*}
\mathbb{E}_{\mathbf{D}} [ (\mathbf{M}_i^{-1}\mathbf{1}_{n-s})(u)] & = \mathbb{E}_{\mathbf{D}} [\mathbf{e}_u^T \mathbf{M}_i^{-1}\mathbf{1}_{n-s}]\\
& \overset{(a)}{=} \mathbf{e}_u^T \mathbb{E}_{\mathbf{D}} [\mathbf{\Pi} \mathbf{M}_i^{-1}\mathbf{1}_{n-s}]\\
& \overset{(b)}{=} \mathbf{e}_v^T \mathbb{E}_{\mathbf{D}} [ \mathbf{M}_i^{-1}\mathbf{1}_{n-s}]\\
& = \mathbb{E}_{\mathbf{D}} [ (\mathbf{M}_i^{-1}\mathbf{1}_{n-s})(v)],
\end{align*}
where $(a)$ holds since $\mathbf{\Pi}\mathbf{M}_i^{-1} \mathbf{1}_{n-s} \overset{d} = \mathbf{M}_i^{-1} \mathbf{1}_{n-s}$ and $(b)$ holds since $\mathbf{\Pi}^T \mathbf{e}_u = \mathbf{e}_v$.
Since $u,v \in [n-s]$ are arbitrary, it follows that all the entries of $\mathbb{E}_{\mathbf{D}}[\mathbf{M}_i^{-1} \mathbf{1}_{n-s}]$ are equal. Also, since $\mathbf{M}_i$ is positive definite, $\mathbf{M}_i^{-1}$ is positive definite. Hence, $\mathbb{E}_{\mathbf{D}}[\mathbf{M}_i^{-1}]$ is positive definite. So, $\mathbf{1}_{n-s}^T\mathbb{E}_{\mathbf{D}}[\mathbf{M}_i^{-1}] \mathbf{1}_{n-s} = \mathbf{1}_{n-s}^T\mathbb{E}_{\mathbf{D}}[\mathbf{M}_i^{-1} \mathbf{1}_{n-s}] > 0$. So, $\mathbb{E}_{\mathbf{D}}[\mathbf{M}_i^{-1} \mathbf{1}_{n-s}]$ is a nonzero vector. Since $\bfM_i$ is a function of $\mathbf{\Delta}_{\cF}$ and for a BIBD $\mathbf{\Delta}_{\cF}$ depends only on the size of $\cF$, $\mathbb{E}_{\mathbf{D}} [ \mathbf{M}_i^{-1}\mathbf{1}_{n-s}] = \alpha_{i,|\cF|} \mathbf{1}_{n-s}$, where $\alpha_{i,|\cF|} > 0$. So, when $\mathcal{F}$ is non-empty,
\begin{align*}
\mathbb{E}_{\mathbf{D}}\left[\mathbf{B}_{\mathcal{F}}\mathbf{r}_i*\right] & = \mathbb{E}_{\mathbf{D}}\left[\begin{bmatrix} \mathbf{A}_{\cF} \tilde{\bfD}_1\tilde{\bf{D}}_i^{-1} \\ \vdots \\\mathbf{A}_{\cF} \\ \vdots \\  \mathbf{A}_{\cF} \tilde{\bfD}_m\tilde{\bf{D}}_i^{-1}\end{bmatrix} \bfM_i^{-1}\mathbf{A}_{\cF}^T \mathbf{1}_k \right]\\
&=\mathbb{E}_{\mathbf{D}}\left[\begin{bmatrix} \mathbf{A}_{\cF} \tilde{\bfD}_1\tilde{\bf{D}}_i^{-1} \bfM_i^{-1}\mathbf{A}_{\cF}^T \mathbf{1}_k \\ \vdots \\\mathbf{A}_{\cF}\bfM_i^{-1}\mathbf{A}_{\cF}^T \mathbf{1}_k  \\ \vdots \\  \mathbf{A}_{\cF} \tilde{\bfD}_m\tilde{\bf{D}}_i^{-1}\bfM_i^{-1}\mathbf{A}_{\cF}^T \mathbf{1}_k \end{bmatrix} \right]\\
&= \begin{bmatrix} \mathbf{0}_k\\ \vdots \\ \delta \mathbf{A}_{\mathcal{F}}\mathbb{E}_{\mathbf{D}}[ \mathbf{M}_i^{-1}\mathbf{1}_{n-s}] \\ \vdots\\ \mathbf{0}_k \end{bmatrix} =  \begin{bmatrix} \mathbf{0}_k\\ \vdots \\ \delta \alpha_{i, |\cF|} \mathbf{A}_{\mathcal{F}} \mathbf{1}_{n-s} \\ \vdots\\ \mathbf{0}_k \end{bmatrix}.
\end{align*}
For non-empty $\cF$, let $\mathbf{P}_{\mathcal{F}} = \mathbf{I}_n(:, \mathcal{F}) \in \mathbb{R}^{n \times (n-s)}$ be the selection matrix that selects the columns of a matrix corresponding to the set $\mathcal{F}$. So, $\mathbf{A}_{\cF} = \bfA \bfP_{\cF}$. 
Also, let $\mathbf{1}_{n}^{(\cF)} := \bfP_{\cF} \mathbf{1}_{n-s}$, where $\mathbf{1}_{n}^{(\cF)}$ is a vector of length $n$ whose $i$-th entry is $1$ if $i \in \cF$ and zero otherwise. So, $\bfA_{\cF} \mathbf{1}_{n-s} = \bfA \bfP_{\cF} \mathbf{1}_{n-s} = \bfA\mathbf{1}_{n}^{(\cF)}$. Now, if $\cF$ is empty, $\mathbf{1}_{n}^{(\cF)} = \mathbf{0}_n$. So, $\bfA\mathbf{1}_{n}^{(\cF)} = \bfA \mathbf{0}_n = \mathbf{0}_k$. Since for $\cF = \emptyset$, $\mathbb{E}_{\bfD}[\bfB_{\cF} \mathbf{r}_i*] = \mathbf{0}_{mk}$,  we have that for all $\cF$,  
\[ \mathbb{E}_{\bfD}[\mathbf{B}_{\cF}  \mathbf{r}_i*] = \begin{bmatrix} \mathbf{0}_k\\ \vdots \\ \delta \alpha_{i,|\cF|} \mathbf{A} \mathbf{1}_{n}^{(\cF)} \\ \vdots\\ \mathbf{0}_k \end{bmatrix}.\]
Now we compute $\mathbb{E}_{\cF} [\alpha_{i,|\cF|} \mathbf{1}_{n}^{(\cF)}]$. Since the worker failures are i.i.d., all coordinates of $\mathbb{E}_{\cF} [\alpha_{i,|\cF|} \mathbf{1}_{n}^{(\cF)}]$ are equal. So, $\mathbb{E}_{\cF} [\alpha_{i,|\cF|} \mathbf{1}_{n}^{(\cF)}] = \alpha_i \mathbf{1}_n$. Moreover, $\alpha_i>0$ since $\alpha_{i,|\mathcal{F}|}>0$ for every nonempty
$\mathcal{F}$ and each worker belongs to $\mathcal{F}$ with probability
$1-q>0$.
Thus,
\[ 
\mathbb{E}_{\mathcal{F}}[\mathbb{E}_{\mathbf{D}}\left[\mathbf{B}_{\mathcal{F}}\mathbf{r}_i*\right]] = \begin{bmatrix} \mathbf{0}_k\\ \vdots \\  \delta \alpha_i \gamma \mathbf{1}_{k} \\ \vdots\\ \mathbf{0}_k \end{bmatrix} =  \delta \alpha_i \gamma  \mathbf{f}_i
.\]

Since the diagonal matrices $\mathbf{D}_i$ are i.i.d., $\mathbb{E}_{\bfD}[\mathbf{M}_i^{-1} \mathbf{1}_{n-s}] = \mathbb{E}_{\bfD}[\mathbf{M}_j^{-1} \mathbf{1}_{n-s}]$ for all $i,j \in [m]$. Therefore, $\alpha_i = \alpha_ j = \alpha'$ (let) for all $i ,j \in [m]$. So, for each $i \in [m]$,  
\[ 
\mathbb{E}_{\mathcal{F}}[\mathbb{E}_{\mathbf{D}}\left[\mathbf{B}_{\mathcal{F}}\mathbf{r}_i*\right]]=  \delta \alpha' \gamma \mathbf{f}_i
.\]
Since $\bfR = \begin{bmatrix}\mathbf{r}_1* & \mathbf{r}_2* & \dots & \mathbf{r}_m* \end{bmatrix}$, and $\bfF = \begin{bmatrix}\mathbf{f}_1 & \mathbf{f}_2 & \dots & \mathbf{f}_m \end{bmatrix}$, we have 
\[ \mathbb{E}_{\mathcal{F}}[\mathbb{E}_{\mathbf{D}}\left[\mathbf{B}_{\mathcal{F}}\mathbf{R}\right]] = \delta \alpha' \gamma \mathbf{F} = \alpha \mathbf{F},\] 
where $\alpha = \delta \alpha' \gamma $. Since $\alpha' \ne 0$ and $\delta, \gamma > 0$, $\alpha$ is a nonzero constant.

\end{proof}
\subsection{Unbiasedness for Vertex-Transitive SRGs and Coset Bipartite Graphs}

The following lemma establishes that the computed gradient for the random diagonal matrix-based construction is unbiased up to a nonzero scaling factor under an appropriate symmetry condition on the assignment matrix.

\begin{lemma}
{\label{lemma:SRG_COSET_unbiasedness}}
Let $\bfA \in \mathbb{R}^{k \times n}$ be an assignment matrix such that $\mathbf{A}^T\mathbf{1}_k\neq \mathbf{0}_n$ and for any $u,v \in [k]$ there exist permutations $\sigma$ of $[k]$ and $\pi$ of $[n]$ such that for the corresponding permutation matrices $\bfP_{\sigma}$ and $\bfQ_{\pi}$, $\mathbf{e}_{u}^T\bfP_{\sigma} = \mathbf{e}_{v}^T$  and $\bfP_{\sigma} \mathbf{A}\bfQ_{\pi}^T = \bfA$. Suppose the encoding matrix $\bfB$ is constructed as in $\eqref{eq:effAD1AD2}$ and that $\bfB_{\cF}^T \bfB_{\cF}$ is invertible for any non-empty $\cF \subseteq [n]$. Then, 

\[ \mathbb{E}_{\mathbf{D}}[\mathbb{E}_{\mathcal{F}}[\mathbf{B}_{\cF} \mathbf{R}]]= \beta\mathbf{F},\]
where $\beta$ is a nonzero constant.

\end{lemma}

\begin{proof}
Fix $i \in [m]$. Let $\bfr_i* = \text{argmin}_{\mathbf{r}} \|\bfB_{\cF} \mathbf{r} - \mathbf{f}_i \|_2^2$. When $\cF$ is empty, we assume $\mathbf{B}_{\cF} = \mathbf{0}_{mk}$. Thus, $\mathbf{B}_{\cF}\mathbf{r}_i* = \mathbf{0}_{mk}$. For non-empty $\cF$, we have
\[ \bfB_{\cF} =\begin{bmatrix} \mathbf{A} \bfD_1 \\ \bfA \bfD_2\\ \vdots \\ \bfA\bfD_m\end{bmatrix} \mathbf{P}_{\mathcal{F}},\] 
where $\mathbf{P}_{\mathcal{F}} = \mathbf{I}_n(:, \mathcal{F})$ is the selection matrix that selects the columns of $\mathbf{B}$ corresponding to the set $\mathcal{F}$. Also, for non-empty $\cF$, $\bfr_i* = (\bfB_{\cF}^T\bfB_{\cF})^{-1} \bfB_{\cF}^T \mathbf{f}_i$. We have,  
\begin{align*} 
\bfB_{\cF} \bfr_i* &=  \bfB_{\cF}(\bfB_{\cF}^T\bfB_{\cF})^{-1} \bfB_{\cF}^T  \mathbf{f}_i\\ 
& = \begin{bmatrix} \mathbf{A} \bfD_1 \\ \bfA \bfD_2\\ \vdots \\ \bfA\bfD_m\end{bmatrix} \mathbf{P}_{\mathcal{F}} \left(\sum_{t=1}^m\mathbf{P}_{\mathcal{F}}^T \mathbf{D}_t\mathbf{A}^T\mathbf{A}\mathbf{D}_t \mathbf{P}_{\mathcal{F}}\right)^{-1}\mathbf{P}_{\mathcal{F}}^T \mathbf{D}_i \mathbf{A}^T \mathbf{1}_k.\\
\end{align*}

For $j \in [m]$, let 

\[
h_{ji}(\mathbf{A}, \mathcal{F}, \mathbf{D}) :=
\begin{cases}
\mathbf{A}\mathbf{D}_j\mathbf{P}_{\mathcal{F}}
\left(\sum_{t = 1}^m \mathbf{P}_{\mathcal{F}}^T \mathbf{D}_t \mathbf{A}^T \mathbf{A} \mathbf{D}_t \mathbf{P}_{\mathcal{F}}\right)^{-1}
\mathbf{P}_{\mathcal{F}}^T \mathbf{D}_i \mathbf{A}^T, & \mathcal{F} \neq \emptyset, \\[0.5em]
\mathbf{0}_{k \times k}, & \mathcal{F} = \emptyset.
\end{cases}
\]

Also, let 
\[ f_{ji}( \mathbf{A}, \mathbf{D}) := \mathbb{E}_{\mathcal{F}}\left[h_{ji}(\mathbf{A}, \mathcal{F}, \mathbf{D})\right].\]

Now,

\[ \mathbb{E}_{\mathbf{D}} [ \mathbb{E}_{\mathcal{F}}[\mathbf{B}_{\mathcal{F}} \mathbf{r}_i*]] = \mathbb{E}_{\mathbf{D}} \left[ \begin{bmatrix}f_{1i}(\mathbf{A}, \mathbf{D}) \\ f_{2i}(\mathbf{A}, \mathbf{D})\\ \vdots \\ f_{mi}(\mathbf{A}, \mathbf{D}) \end{bmatrix} \mathbf{1}_k\right] .\]
For $j \in [m]$, let $f_j(\mathbf{D}_i)  := \mathbb{E}_{\mathbf{D}_t, t\in [m], t \ne i}[f_{ji}(\mathbf{A}, \mathbf{D}) \mathbf{1}_k] $. Note that if $i \ne j$, 
\[f_j(-\mathbf{D}_i) = - f_j(\mathbf{D}_i).\]
So, $f_j(\mathbf{D}_i)$ is an odd function of $\mathbf{D}_i$. Also, the distribution of $\mathbf{D}_i$ is symmetric, i.e., $\bfD_i \overset{d}= - \bfD_i$. Thus, for $i \ne j$, we have, $\mathbb{E}_{\mathbf{D}_i}[f_j(\mathbf{D}_i)] =  \mathbf{0}_k$.
Therefore, for $i \ne j$, 
\begin{align*}\mathbb{E}_{\mathbf{D}} [f_{ji}(\mathbf{A}, \mathbf{D}) \mathbf{1}_k] &= \mathbb{E}_{\mathbf{D}_i}[\mathbb{E}_{\mathbf{D}_t,   t \in [m], t \ne i}[f_{ji} (\mathbf{A}, \mathbf{D})\mathbf{1}_k]]\\&= \mathbb{E}_{\mathbf{D}_i}[f_j(\mathbf{D}_i)] =  \mathbf{0}_k.\end{align*}

Next we find $\mathbb{E}_{\bfD}[f_{ii}(\bfA, \bfD) \mathbf{1}_k]$. Suppose $u, v \in [k]$. So, by assumption there exist permutations $\sigma$ of $[k]$ and $\pi$ of $[n]$ such that for the corresponding permutation matrices $\bfP_{\sigma}$ and $\bfQ_{\pi}$, $\mathbf{e}_u ^T\bfP_{\sigma}= \mathbf{e}_{v}^T$ and $\bfP_{\sigma} \mathbf{A}\bfQ_{\pi}^T = \bfA$.
We will show that $\mathbf{P}_{\sigma} f_{ii}( \mathbf{A}, \mathbf{D}) \mathbf{P}_{\sigma}^T \overset{d}{=} f_{ii}( \mathbf{A}, \mathbf{D})$.






Let $\mathcal{H}_r$ be the set of subsets of $[n]$ of size $r$, $r \in [n]$. Let $\pi_{\mathcal{H}_r}: \mathcal{H}_r \to \mathcal{H}_r$ be such that for $\cF \in \mathcal{H}_r$, $\pi_{\mathcal{H}_r}(\mathcal{F}) = \{\pi(u)| u \in \cF\}$. Since $\pi$ is a bijection on $[n]$, $\pi_{\mathcal{H}_r}$ is also a bijection on $\mathcal{H}_r$. Also, since $\bfQ_{\pi} \mathbf{e}_w = \mathbf{e}_{\pi(w)}$ for $w \in [n]$, $\bfQ_{\pi} \bfP_{\cF} = \bfP_{\pi_{\mathcal{H}_r}(\cF)}$. Let $\mathbf{Y}(\bfP_{\cF}, \bfD) := \sum_{t=1}^m 
\mathbf{P}_{\cF}^T \mathbf{D}_t \mathbf{A}^T\mathbf{A}\mathbf{D}_t \mathbf{P}_{\cF}$. Let $c_r = q^{n-r} (1-q)^{r}$. Since we assume each worker straggles independently with probability $q$, we have

\begin{align*}
f_{ii}(\mathbf{A}, \mathbf{D})
& \overset{(a)}= q^n  \mathbf{0}_{k \times k} + \sum_{r = 1}^n 
    \sum_{\substack{F \subseteq [n] \\ |F| = r}} c_r \, h_{ii}(\mathbf{A}, F, \mathbf{D})\\
& \overset{(b)}= \sum_{r = 1}^n 
    \sum_{\substack{F \subseteq [n] \\ |F| = r}}
    c_r \, h_{ii}(\mathbf{A}, F, \mathbf{D})\\
& \overset{(c)}= \sum_{r = 1}^n 
    \sum_{\substack{F \subseteq [n] \\ |F| = r}}
    c_r \,
    h_{ii}(\mathbf{A}, \pi_{\mathcal{H}_r}(F), \mathbf{D})\\
& \overset{(d)}= \sum_{r = 1}^n 
    \sum_{\substack{F \subseteq [n] \\ |F| = r}}
    c_r \,
   \mathbf{A}\mathbf{D}_i\mathbf{P}_{\pi_{\mathcal{H}_r}(F)} (\mathbf{Y}(\bfP_{\pi_{\mathcal{H}_r}(F)}, \bfD))^{-1}
    \mathbf{P}_{\pi_{\mathcal{H}_r}(F)}^T  \mathbf{D}_i\mathbf{A}^T\\ 
& \overset{(e)}= \sum_{r = 1}^n 
    \sum_{\substack{F \subseteq [n] \\ |F| = r}}
   c_r \,
   \mathbf{A}\mathbf{D}_i\bfQ_{\pi}\mathbf{P}_{F} (\mathbf{Y}(\bfQ_{\pi}\bfP_{F}, \bfD))^{-1}
    \mathbf{P}_{F}^T \bfQ_{\pi}^T \mathbf{D}_i\mathbf{A}^T.
\end{align*}

Here, $(a)$ holds by definition of $f_{ii}(\bfA, \bfD)$. $(b)$ holds since $h_{ii}(\bfA, \cF, \bfD)$ is a $k \times k$ zero matrix for $|\cF| = 0$. $(c)$ holds since $\pi_{\mathcal{H}_r}$ is a bijection. $(d)$ holds by definition of $h_{ii}(\mathbf{A}, \pi_{\mathcal{H}_r}(F), \mathbf{D})$. $(e)$ holds since  $\bfQ_{\pi} \bfP_{F} = \bfP_{\pi_{\mathcal{H}_r}(F)}$, for each $r \in [n]$. 

Recall that $\mathbf{D} = (\mathbf{D}_1, \dots, \mathbf{D}_m)$. In the expression below, $\mathbf{Q}_{\pi} \mathbf{D} \mathbf{Q}_{\pi}^T$ is used to denote $(\mathbf{Q}_{\pi} \mathbf{D}_1 \mathbf{Q}_{\pi}^T, \dots, \mathbf{Q}_{\pi} \mathbf{D}_m \mathbf{Q}_{\pi}^T)$. We have,
\begin{align*}
&\mathbf{P}_{\sigma} f_{ii}( \mathbf{A}, \mathbf{D}) \mathbf{P}_{\sigma}^T\\
& \overset{(f)}= q^n \mathbf{0}_{k \times k} + \sum_{r = 1}^n 
    \sum_{\substack{F \subseteq [n] \\ |F| = r}} c_r \, \mathbf{P}_{\sigma} h_{ii}(\mathbf{A}, F, \mathbf{D}) \mathbf{P}_{\sigma}^T\\
& \overset{(g)}=  \sum_{r = 1}^n 
    \sum_{\substack{F \subseteq [n]\\ |F| = r}} 
    c_r \,
    \mathbf{P}_{\sigma} \mathbf{A}\mathbf{D}_i\mathbf{P}_{F} (\bfY(\bfP_{F}, \bfD))^{-1}
    \mathbf{P}_{F}^T  \mathbf{D}_i\mathbf{A}^T\mathbf{P}_{\sigma}^T\\[0.5em]
& \overset{(h)}=  \sum_{r = 1}^n 
    \sum_{\substack{F \subseteq [n]\\ |F| = r}} 
    c_r \,
    \mathbf{A}\mathbf{Q}_{\pi}\mathbf{D}_i\mathbf{Q}_{\pi}^T\mathbf{Q}_{\pi}\mathbf{P}_{F} (\bfY(\bfQ_{\pi}\bfP_{\cF}, \bfQ_{\pi}\bfD \bfQ_{\pi}^T))^{-1}
    \mathbf{P}_{F}^T\mathbf{Q}_{\pi}^T\mathbf{Q}_{\pi}  
    \mathbf{D}_i \mathbf{Q}_{\pi}^T\mathbf{A}^T\\[0.5em]
& \overset{(i)}=  f_{ii}(\mathbf{A}, \mathbf{Q}_{\pi}\mathbf{D} \mathbf{Q}_{\pi}^T).
\end{align*}

Here, $(f)$ holds by definition of $f_{ii}(\bfA, \bfD)$. $(g)$ holds by definition of $h_{ii}(\bfA, F, \bfD)$. $(h)$ holds since $\mathbf{Q}_{\pi}^T \mathbf{Q}_{\pi} = \mathbf{I}_n$ and $\bfP_{\sigma} \bfA \bfQ_{\pi}^T = \bfA$. $(i)$ holds as a consequence of $(e)$. Since for each $t \in [m]$, we have $\mathbf{D}_t \overset{d}{=} \mathbf{Q}_{\pi}\mathbf{D}_t \mathbf{Q}_{\pi}^T$, 
$ f_{ii}(\mathbf{A}, \mathbf{Q}_{\pi}\mathbf{D} \mathbf{Q}_{\pi}^T) \overset{d}{=}  f_{ii}(\mathbf{A}, \mathbf{D}) 
.$
Consequently, 
\[ \mathbf{P}_{\sigma} f_{ii}( \mathbf{A}, \mathbf{D}) \mathbf{P}_{\sigma}^T \overset{d}{=} f_{ii}( \mathbf{A}, \mathbf{D}),\]

and thus $\mathbb{E}_{\bfD}[\mathbf{P}_{\sigma} f_{ii}( \mathbf{A}, \mathbf{D}) \mathbf{P}_{\sigma}^T] = \mathbb{E}_{\bfD}[f_{ii}( \mathbf{A}, \mathbf{D})]$. Using this distributional equivalence and the properties of $\mathbf{P}_\sigma$ we have,
\begin{align*}
(\mathbb{E}_{\mathbf{D}} [f_{ii}( \mathbf{A}, \mathbf{D}) \mathbf{1}_k ])(u) 
& = \mathbb{E}_{\mathbf{D}}[\mathbf{e}_u^T f_{ii}( \mathbf{A}, \mathbf{D}) \mathbf{1}_k]\\
& = \mathbf{e}_u^T  \mathbb{E}_{\mathbf{D}}[\mathbf{P}_{\sigma}f_{ii}( \mathbf{A}, \mathbf{D}) \mathbf{P}_{\sigma}^T] \mathbf{1}_k\\
& = \mathbb{E}_{\mathbf{D}}[ \mathbf{e}_v^Tf_{ii}( \mathbf{A}, \mathbf{D}) \mathbf{1}_k]\\
& = (\mathbb{E}_{\mathbf{D}} [f_{ii}( \mathbf{A}, \mathbf{D}) \mathbf{1}_k ])(v).
\end{align*}

By our assumptions $u,v \in [k]$ are arbitrary. Thus, it follows that all the entries of $\mathbb{E}_{\mathbf{D}} [f_{ii}( \mathbf{A}, \mathbf{D}) \mathbf{1}_k ]$ are equal. Moreover, $h_{ii}(\bfA, \cF, \bfD)$ is positive semidefinite for every non-empty $\cF$. For $\mathcal F=[n]$, since $\mathbf{A}^T\mathbf{1}_k \ne \mathbf{0}_n$ and $\mathbf D_i$ is invertible,
\[
\mathbf 1_k^T h_{ii}(\mathbf A,[n],\mathbf D)\mathbf 1_k
=
(\mathbf D_i\mathbf A^T\mathbf 1_k)^T
(\mathbf B^T\mathbf B)^{-1}
(\mathbf D_i\mathbf A^T\mathbf 1_k)>0.
\]
Since $q < 1$, $\cF = [n]$ has nonzero probability. It follows that
$
\mathbf 1_k^T\mathbb{E}_{\mathbf D}
[f_{ii}(\mathbf A,\mathbf D)]\mathbf 1_k>0
$.
Hence, $\mathbb{E}_{\mathbf D}[f_{ii}(\mathbf A,\mathbf D)\mathbf 1_k]$
is a nonzero vector. Consequently,
$
\mathbb{E}_{\mathbf D}[f_{ii}(\mathbf A,\mathbf D)\mathbf 1_k]
=\beta_i\mathbf 1_k
$, where $\beta_i$ is a nonzero constant.  
Therefore, for $i \in [m]$,
\[\mathbb{E}_{\mathbf{D}} [\mathbb{E}_{\mathcal{F}}[ \mathbf{B}_{\mathcal{F}} \mathbf{r}_i*] ] = \begin{bmatrix}\mathbf{0}_k \\ \vdots \\ \beta_i\mathbf{1}_k\\ \vdots \\ \mathbf{0}_k \end{bmatrix} = \beta_i\mathbf{f}_i.\]
Since the diagonal matrices $\mathbf{D}_i$ are i.i.d., $\mathbb{E}_{\mathbf{D}}[ {f}_{ii}(\bfA,\bfD) \mathbf{1}_{k}] = \mathbb{E}_{\mathbf{D}}[{f}_{jj}(\bfA,\bfD) \mathbf{1}_{k}]$ for all $i,j \in [m]$. Therefore, $\beta_i = \beta_ j = \beta$ (let) for all $i ,j \in [m]$.
Since $\bfR = \begin{bmatrix}\mathbf{r}_1* & \mathbf{r}_2* & \dots & \mathbf{r}_m* \end{bmatrix}$, and $\bfF = \begin{bmatrix}\mathbf{f}_1 & \mathbf{f}_2 & \dots & \mathbf{f}_m \end{bmatrix}$, we have  \[\mathbb{E}_{\mathbf{D}} [\mathbb{E}_{\mathcal{F}}[ \mathbf{B}_{\mathcal{F}} \mathbf{R}] ] = \beta\mathbf{F},\]
where $\beta$ is a nonzero constant.
\end{proof}

\begin{cor}
\label{cor:SRG_unbiased}
 Suppose that the assignment matrix $\mathbf{A}$ is the adjacency matrix of a $(n, \delta, \lambda, \mu)$ vertex-transitive strongly regular graph $G = ([n], E)$. 
Then, for the construction as in $\eqref{eq:effAD1AD2}$,

 \[ \mathbb{E}_{\mathbf{D}}[\mathbb{E}_{\mathcal{F}}[\mathbf{B}_{\cF} \mathbf{R}]]= \beta\mathbf{F},\]
where $\beta$ is a nonzero constant.

\end{cor}

\begin{proof}
By definition, for a vertex transitive strongly regular graph, for any $u, v \in [n]$ there exists an automorphism $\sigma$ such that $\sigma(v) = u$. So, for                    the corresponding permutation matrix $\bfP_{\sigma}$, $\mathbf{e}_{u}^T\bfP_{\sigma} = \mathbf{e}_{v}^T$ and $\bfP_{\sigma}\bfA \bfP_{\sigma}^T = \bfA$. Also, as proved in Corollary $\ref{cor:SRG}$, $\bfB_{\cF}^T \bfB_{\cF}$ is invertible for non-empty $\cF \subseteq [n]$. Since $\bfA^T \mathbf{1}_k = \delta \mathbf{1}_n$, the conclusion follows from Lemma $\ref{lemma:SRG_COSET_unbiasedness}$.
\end{proof}

\begin{cor}
\label{cor:COSET_unbiased}
Suppose that the assignment matrix $\mathbf{A}$ is the bi-adjacency matrix of a $(k,m,\delta)$ coset bipartite graph such that $k = p^a$, where $p$ is a prime, $a \in \mathbb{Z}_{\ge 1}$ and $p \nmid \delta$. Then, for the construction as in $\eqref{eq:effAD1AD2}$ with $\epsilon > 0$,
\[ \mathbb{E}_{\mathbf{D}}[\mathbb{E}_{\mathcal{F}}[\mathbf{B}_{\cF} \mathbf{R}]]= \beta\mathbf{F},\]
where $\beta$ is a nonzero constant.
\end{cor}

\begin{proof}
Let $G = (L \cup R, E)$ be a  $(k, m, \delta)$ coset bipartite graph. So, $R = \mathbb{Z}_{mk}$ and $L = \{\, i+H : i=0,1,\dots,k-1 \,\}$ where $H$ is an order $m$ subgroup of $\mathbb{Z}_{mk}$. By construction, $i + H \in L$ is adjacent to $x \in R = \mathbb{Z}_{mk}$ if and only if $x \in i + S$, for $i \in \{0, \dots, k-1\}$ and $x \in \mathbb{Z}_{mk}$. Consequently, for its bi-adjacency matrix $\bfA$, we have $\bfA(i+1, x+1) = 1$ if and only if $x \in i + S $ (note that the sets of row and column indices of $\bfA$ are $[k]$ and $[mk]$ respectively).
For $g \in \{ 0, \dots, k-1\}$, define the row permutation $\sigma_g:[k]\to[k]$ by
\[
\sigma_g(i) =  ((i-1+g) \bmod{k}) + 1,
\]
and the column permutation $\tau_g:[mk]\to [mk]$ by
\[
\tau_g(x) =  ((x-1+g) \bmod{mk}) +1 .
\]
Let $\mathbf{P}_{\sigma_g}\in\{0,1\}^{k\times k}$ and $\mathbf{Q}_{\tau_g}\in\{0,1\}^{(mk)\times(mk)}$ denote the corresponding permutation matrices.
We will first show that for every $g\in \{0 , \dots, k-1\}$,
$\bfP_{\sigma_g}\bfA \mathbf{Q}_{\tau_g}^T= \bfA.$
We have,
\[
(\bfP_{\sigma_g}\bfA \mathbf{Q}_{\tau_g}^T )(i+1,x+1)
= \bfA(\sigma_g^{-1}(i+1),\,\tau_g^{-1}(x+1))
\]
Moreover, by the definition of $\sigma_g$ and $\tau_g$,
\[
\sigma_g^{-1}(y) = ((y-1-g)\bmod{k})+1,
\qquad
\tau_g^{-1}(z) = ((z-1-g)\bmod{mk})+1.
\]
Thus,
\[
(\bfP_{\sigma_g}\bfA \mathbf{Q}_{\tau_g}^T )(i+1,x+1)
= \bfA\!\Big( ((i-g)\bmod{k})+1,\; ((x-g)\bmod{mk})+1 \Big).
\]

But, by definition, $\bfA\big((i-g)\bmod k + 1,\,(x-g)\bmod mk + 1\big)=1$
\text{if and only if }
$(x-g)\bmod mk \in \big((i-g)\bmod k\big)+S$. 
Let $i':=(i-g)\bmod{k}\in\{0,\dots,k-1\}$ and $x':=(x-g)\bmod{mk}\in\{0,\dots,mk-1\}$.
We can view $i',x'$ as elements of $\mathbb{Z}_{mk}$ (via the natural embedding of $\{0,\dots,k-1\}$ into $\mathbb{Z}_{mk}$). 
Since $i' = i-g \pmod{k}$, there exists an integer $t$ such that
\[
i' = i-g + tk.
\]
Also, since $x' = x-g \pmod{mk}$, there exists an integer $u$ such that
\[
x' = x-g + umk.
\]
So, $x-g+umk \in i -g + tk + S$. By definition of $S$, $tk +S = S$. So, $x-g+umk \in i -g + S$. Therefore, by definition of $S$ there exists $b \in S$ such that for some integer $v$, 
\[ x-g + umk = i-g+b + vmk.\] So, $x = i+b + (v-u) mk$. Since $b \in S$, $b + (v-u)mk \in S$. Now, $ i + b + (v-u)mk = (i +  b  + (v-u)mk) \bmod mk$, since $x < mk$. Thus, $i +  b  + (v-u)mk \in i +S$ and therefore, $x \in i +S$. Similarly we can show that if $x \in i +S$, then $x' \in i' +S$. Consequently, 
$\bfA\!\Big( ((i-g)\bmod{k})+1,\; ((x-g)\bmod{mk})+1 \Big)=1$
if and only if  $\bfA(i+1,x+1)=1$. Since the entries of $\bfA$ are either $0$ or $1$,  $(\bfP_{\sigma_g} \bfA \mathbf{Q}_{\tau_g}^{T})(i+1,x+1)= \bfA(i+1,x+1)$ for all $i \in \{0, \dots, k-1\}$ and $x \in \mathbb{Z}_{mk}$. So, $\bfP_{\sigma_g} \bfA \mathbf{Q}_{\tau_g}^{T} = \bfA$. Now, given any pair of rows $i,j \in [k]$, we can choose $g =  (i-j) \bmod k$. Then, $\sigma_g(j) = i$. Thus,
$
\mathbf{e}_i^T\bfP_{\sigma_g}  = \mathbf{e}_j^T.
$
So, for any $i,j \in [k]$, there exist permutation matrices $\bfP_{\sigma_g}$ and $\bfQ_{\tau_g}$ such that $\bfP_{\sigma_g} \bfA \bfQ_{\tau_g}^T = \bfA$ and $\mathbf{e}_i^T \bfP_{\sigma_g} = \mathbf{e}_j^T$. Also, as proved in Proposition $\ref{prop: cos_inv}$, for $k = p^a$ where $p$ is a prime, $a \in \mathbb{Z}_{\ge 1}$ and $p \nmid \delta$, $\bfB$ is almost surely invertible and thus $\bfB_{\cF}^T \bfB_{\cF}$ is invertible for non-empty $\cF \subseteq [n]$. Since $\bfA^T \mathbf{1}_k = \delta \mathbf{1}_n$, the conclusion follows from Lemma $\ref{lemma:SRG_COSET_unbiasedness}$.
\end{proof}

\subsection{Bounded Variance of the Computed Gradient}

For the construction in \eqref{eq:hada_coset_cons}, the corresponding unbiasedness result is established in Lemma \ref{lemma:Cost_Hadamard_Convergence} in Appendix \ref{app:hada_cos_unbiased}.

\begin{thm}
\label{thm:Bound_variance}
Suppose that the spectral norm of $\mathbf{Z}$ is bounded, i.e.,
there exists a constant $c_z>0$ such that $\|\mathbf{Z}\|_2\leq c_z$. Then, for the construction as in $\eqref{eq:effAD1AD2}$ with BIBD, vertex-transitive SRG and coset bipartite graph assignment matrices, and for the construction as in $\eqref{eq:hada_coset_cons}$, the variance of the rescaled computed gradient is bounded.
\end{thm}

\begin{proof}
By Lemma \ref{lemma:BIBD_unbiasedness} for the BIBD construction, Corollaries~\ref{cor:SRG_unbiased} and \ref{cor:COSET_unbiased} for the
vertex-transitive SRG and coset bipartite graph constructions as in
\eqref{eq:effAD1AD2}, respectively, the rescaled computed gradient is unbiased. Thus,
\begin{equation}
\mathbb E_{\mathbf D}
\left[
\mathbb E_{\mathcal F}
\left[
\mathbf B_{\mathcal F}\mathbf R
\right]
\right]
=
\zeta\mathbf F.
\end{equation}
Therefore,
\begin{align}
\mathbb E_{\mathbf D}
\left[
\mathbb E_{\mathcal F}
\left[
\left\|
\mathbf B_{\mathcal F}\mathbf R-\mathbf F
\right\|_F^2
\right]
\right] =
\mathbb E_{\mathbf D}
\left[
\mathbb E_{\mathcal F}
\left[
\left\|
\mathbf B_{\mathcal F}\mathbf R-\zeta\mathbf F
\right\|_F^2
\right]
\right]
+
\left\|
(\zeta-1)\mathbf F
\right\|_F^2.
\end{align}
Consequently,
\begin{align}
\mathbb E_{\mathbf D}
\left[
\mathbb E_{\mathcal F}
\left[
\left\|
\mathbf B_{\mathcal F}\mathbf R-\zeta\mathbf F
\right\|_F^2
\right]
\right]
\leq
\mathbb E_{\mathbf D}
\left[
\mathbb E_{\mathcal F}
\left[
\operatorname{Err}_{\mathcal F}(\mathbf B)
\right]
\right],
\end{align}
where
\begin{equation}
\left\|
\mathbf B_{\mathcal F}\mathbf R-\mathbf F
\right\|_F^2
=
\operatorname{Err}_{\mathcal F}(\mathbf B)
\end{equation}
since $\mathbf R$ is the optimal decoding matrix. Hence,
\begin{align}
\mathbb E_{\mathbf D}
\left[
\mathbb E_{\mathcal F}
\left[
\left\|
\frac{1}{\zeta}
\mathbf Z\mathbf B_{\mathcal F}\mathbf R
-
\mathbf Z\mathbf F
\right\|_F^2
\right]
\right]
&\leq
\frac{\|\mathbf Z\|_2^2}{\zeta^2}
\mathbb E_{\mathbf D}
\left[
\mathbb E_{\mathcal F}
\left[
\left\|
\mathbf B_{\mathcal F}\mathbf R-\zeta\mathbf F
\right\|_F^2
\right]
\right]
\\
&\leq
\frac{c_z^2}{\zeta^2}
\mathbb E_{\mathbf D}
\left[
\mathbb E_{\mathcal F}
\left[
\operatorname{Err}_{\mathcal F}(\mathbf B)
\right]
\right].
\end{align}

Corollaries~\ref{cor:expected_error_BIBD_strr}, \ref{cor:expected_error_SRG_strr}, and \ref{cor:expected_error_coset_strr} in Appendix \ref{sec:exp_app_err_rand_strr} show that for each of the corresponding
constructions, there exists a finite constant $U$ depending only on
the construction and straggler model parameters such that
\begin{equation}
\mathbb E_{\mathbf D}
\left[
\mathbb E_{\mathcal F}
\left[
\operatorname{Err}_{\mathcal F}(\mathbf B)
\right]
\right]
\leq U.
\end{equation}
Therefore,
\begin{equation}
\mathbb E_{\mathbf D}
\left[
\mathbb E_{\mathcal F}
\left[
\left\|
\frac{1}{\zeta}
\mathbf Z\mathbf B_{\mathcal F}\mathbf R
-
\mathbf Z\mathbf F
\right\|_F^2
\right]
\right]
\leq
\frac{c_z^2U}{\zeta^2}.
\end{equation}

Also, for the construction as in \eqref{eq:hada_coset_cons}, by Lemma \ref{lemma:Cost_Hadamard_Convergence}, $\mathbb E_{\mathcal F}
\left[
\mathbf B_{\mathcal F}\mathbf R
\right]
=
\zeta\mathbf F$, and by Lemma \ref{lemma:hada_coset_bound}, $
\mathbb E_{\mathcal F}
\left[
\operatorname{Err}_{\mathcal F}(\mathbf B)
\right] 
\leq  nq$. Therefore, by the same argument as above, $\mathbb E_{\mathcal F}
\left[
\left\|
\frac{1}{\zeta}
\mathbf Z\mathbf B_{\mathcal F}\mathbf R
-
\mathbf Z\mathbf F
\right\|_F^2
\right]
\leq
\frac{c_z^2 nq}{\zeta^2}.$ So the variance bound holds with $U = nq$.

\end{proof}

\subsection{Stationary-Point Convergence}

\begin{cor}
Consider the construction as in $\eqref{eq:effAD1AD2}$ with BIBD, vertex-transitive SRG and coset bipartite graph assignment matrices, and the construction as in $\eqref{eq:hada_coset_cons}$,
and suppose that the spectral norm of $\mathbf{Z}$ is bounded, i.e.,
there exists a constant $c_z>0$ such that $\|\mathbf{Z}\|_2\leq c_z$. Furthermore, assume that $L$ is bounded below by $L_{\inf}$ and has
an $L_s$-Lipschitz-continuous gradient. Let
\begin{equation}
\sigma^2=\frac{c_z^2U}{\zeta^2}
\end{equation}
and define
\begin{equation}
D_L
=
\left(
\frac{2\left(L(\mathbf w^{(1)})-L_{\inf}\right)}
{L_s}
\right)^{1/2}.
\end{equation}
For a given number of iterations $T$, let $\mathbf{v}^{(t)}$ be the computed gradient obtained from the columns of $\frac{1}{\zeta}\mathbf{Z}^{(t)}\mathbf{B}_{\cF}^{(t)}\mathbf{R}^{(t)}$ at iteration $t$.
\begin{equation}
\mathbf w^{(t+1)}
=
\mathbf w^{(t)}
-
\eta
\mathbf{v}^{(t)},
\qquad t=1,\ldots,T,
\end{equation}
where
\begin{equation}
\eta
=
\min\left\{
\frac{1}{L_s},
\frac{\widetilde D}{\sigma\sqrt{T}}
\right\}
\end{equation}
for some $\widetilde D>0$. Let $R_T$ be uniformly distributed over
$\{1,\ldots,T\}$. Then,
\begin{equation}
\frac{1}{L_s}
\mathbb E
\left[
\left\|
\nabla L\left(\mathbf w^{(R_T)}\right)
\right\|_2^2
\right]
\leq
\frac{L_sD_L^2}{T}
+
\left(
\widetilde D+\frac{D_L^2}{\widetilde D}
\right)
\frac{\sigma}{\sqrt{T}}.
\end{equation}
Consequently,
\begin{equation}
\lim_{T\rightarrow\infty}
\mathbb E
\left[
\left\|
\nabla L\left(\mathbf w^{(R_T)}\right)
\right\|_2^2
\right]
=0.
\end{equation}
\end{cor}

\begin{proof}
By Lemmas \ref{lemma:BIBD_unbiasedness} and $\ref{lemma:Cost_Hadamard_Convergence}$ and Corollaries \ref{cor:SRG_unbiased} and \ref{cor:COSET_unbiased}, the
rescaled computed gradient is unbiased. Theorem \ref{thm:Bound_variance} establishes that its variance is bounded. Therefore, under the stated assumptions, the result follows from
Corollary~2.2 in~\cite{doi:10.1137/120880811}.
\end{proof}


\section{Numerical Experiments}
\label{sec:numerical_exp}

In what follows, we refer to the construction obtained by stacking the underlying assignment matrix vertically $m$ times as the baseline construction (similar to  \eqref{eq:trivial_cons}). We present the results of numerical experiments to demonstrate the performance of our constructions and compare them with the baseline construction, along with the lower bound derived in Theorem \ref{thm:lowerbound} (see \cite{cE_agc_repo} for the source code used for the experiments). 


\begin{figure*}[t!]
    \centering

    \begin{subfigure}[t]{0.45\textwidth}
        \centering
        \includegraphics[
            width=\linewidth,
            ,
            clip
        ]{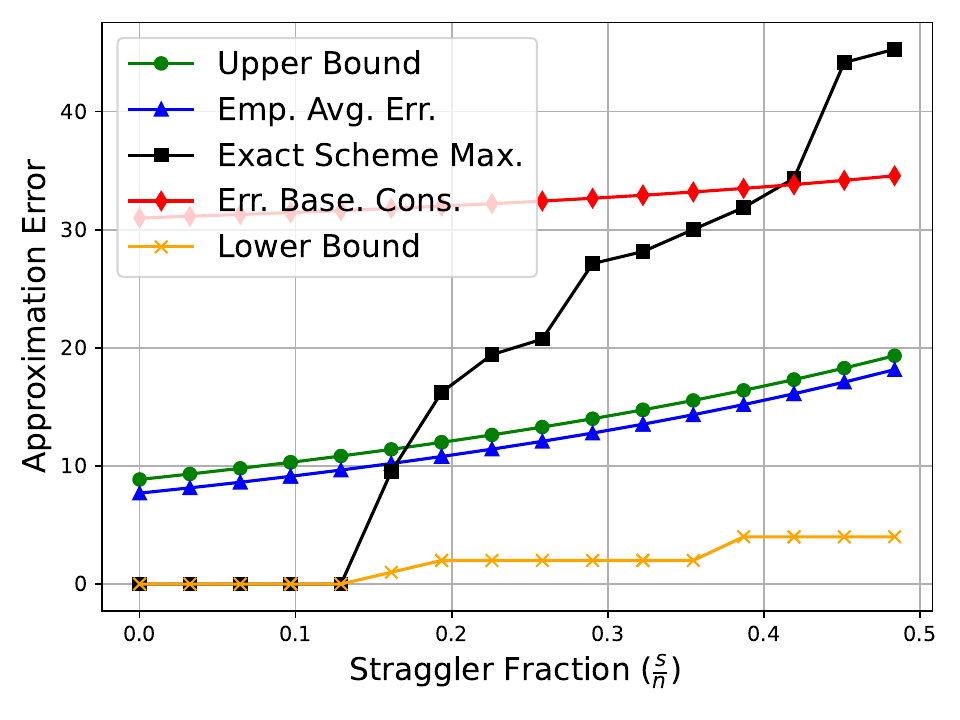}
        \caption{Comparison with the exact communication-efficient gradient coding scheme in \cite{tayyebehM21} for $m =2$.}
        \label{fig:bibd_31_exact}
        
    \end{subfigure}
    \hfill
    \begin{subfigure}[t]{0.45\textwidth}
        \centering
        \includegraphics[
            width=\linewidth,
          ,
            clip
        ]{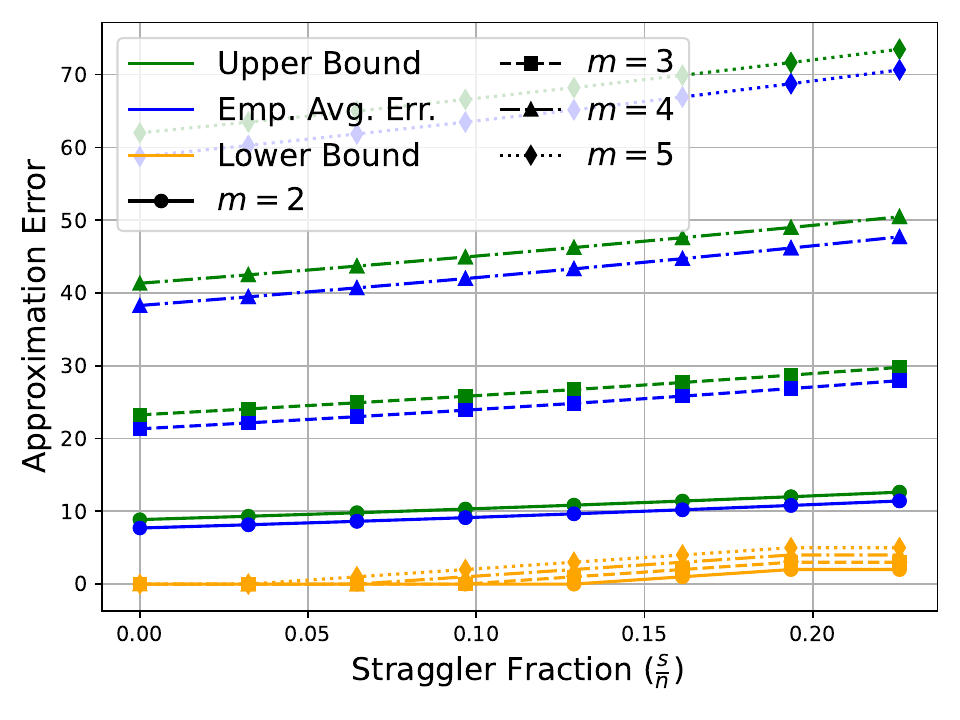}
        \caption{Performance for different values of $m$.}
        \label{fig:different_m}
        
    \end{subfigure}

    \caption{Performance of the construction as in \eqref{eq:effAD1AD2} with
    $\epsilon=0$ and a $(31,31,6,6,1)$ BIBD.}
    \label{fig:bibd_comparison}
\end{figure*}

Firstly, we chose a $(31,31,6,6,1)$ BIBD as the assignment matrix and $m = 2$ to evaluate the construction discussed in Section \ref{sec:construction_RDM} (see \eqref{eq:effAD1AD2}). To compare with the upper bound \eqref{eq:AD1AD2err_BIBD}, the empirical average error was calculated as follows. For a given number of stragglers $s$, $100$ different realizations of the encoding matrix $\mathbf{B}$ were generated by generating the matrices $\mathbf{D}_1, \mathbf{D}_2$ each time from the underlying distribution with $\epsilon = 0$. For each of these realizations, the approximation error \eqref{eq:ApproxError} was calculated for each of $1000$ random choices of stragglers. Finally, the average of the errors $(100 \times 1000)$ was calculated (see \autoref{fig:bibd_31_exact}). Notice that this average error lies close to the upper bound \eqref{eq:AD1AD2err_BIBD}. Also, the upper bound is significantly lower than the error for the baseline construction. We additionally compare this construction against the exact communication-efficient gradient coding scheme in~\cite{tayyebehM21}. For the same set of parameters, the exact scheme is 
designed to tolerate $s=\gamma-m=4$ stragglers. Since its polynomial encoder 
is linear, we represent it in encoding matrix form and evaluate its error using 
the same optimal least-squares decoder. For each value of $s$, we report the maximum approximation error of the exact scheme over the considered straggler sets. 
We point out that the condition number of the corresponding matrix is very high and can lead to numerical instability in scenarios with higher number of workers.
As shown in \autoref{fig:bibd_31_exact}, 
it achieves zero error within its designed tolerance, but its error increases 
sharply beyond this threshold, whereas the proposed construction exhibits a 
more gradual increase.

We further evaluate this construction for $m\in\{2,3,4,5\}$ using the same procedure. As shown in \autoref{fig:different_m}, the empirical average error remains close to the corresponding upper bound for all the considered values of $m$, while the lower bound is also shown for comparison.



\begin{figure*}[t!]
    \centering

    \begin{subfigure}[t]{0.45\textwidth}
        \centering
        \includegraphics[
            width=\linewidth,
            ,
            clip
        ]{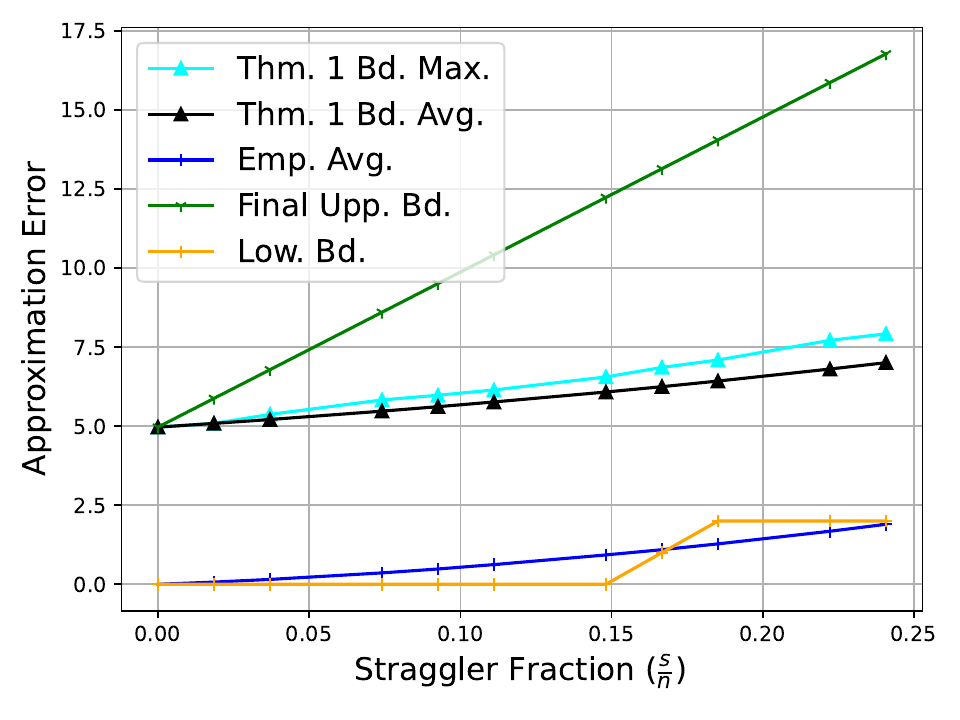}
        \caption{Empirical average error and corresponding bounds.}
        \label{fig:rdm_cos_bound}
        
    \end{subfigure}
    \hfill
    \begin{subfigure}[t]{0.45\textwidth}
        \centering
        \includegraphics[
            width=\linewidth,
          ,
            clip
        ]{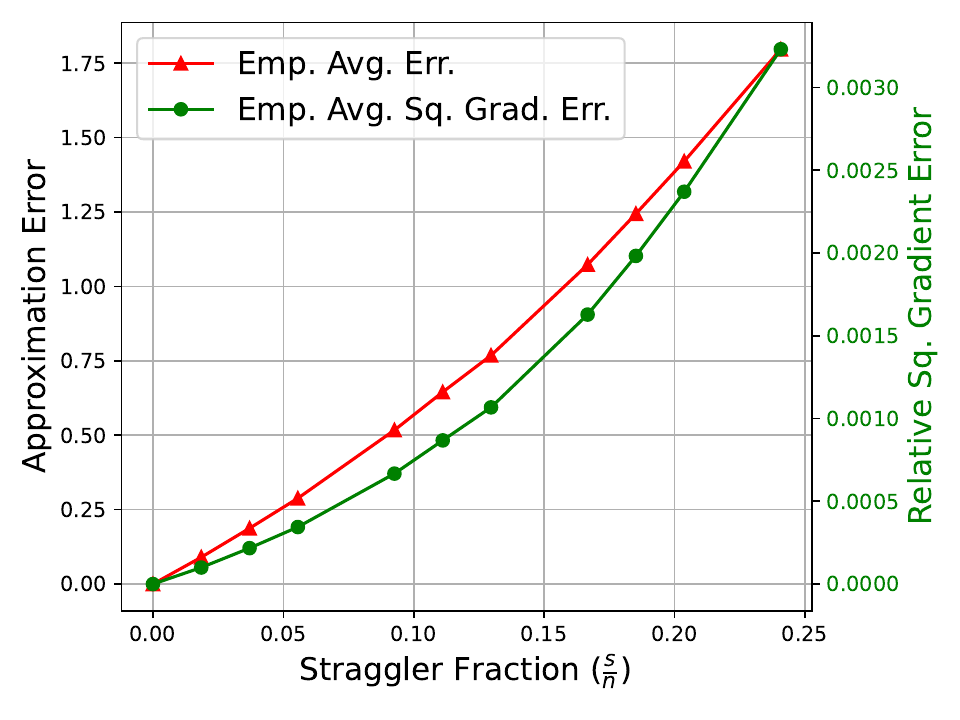}
        \caption{Comparison of empirical average error and relative squared gradient error.}
        \label{fig:err_gradd_err}
        
    \end{subfigure}

    \caption{{\small Performance of the construction as in \eqref{eq:effAD1AD2} for $\epsilon  = 0.1$, and $m = 2$ with a $(27, 2, 5) $ coset bipartite graph.}}
    \label{fig:coset_rdm}
\end{figure*}

\begin{figure*}[t!]
    \centering

    \begin{subfigure}[t]{0.45\textwidth}
        \centering
        \includegraphics[
            width=\linewidth,
            ,
            clip
        ]{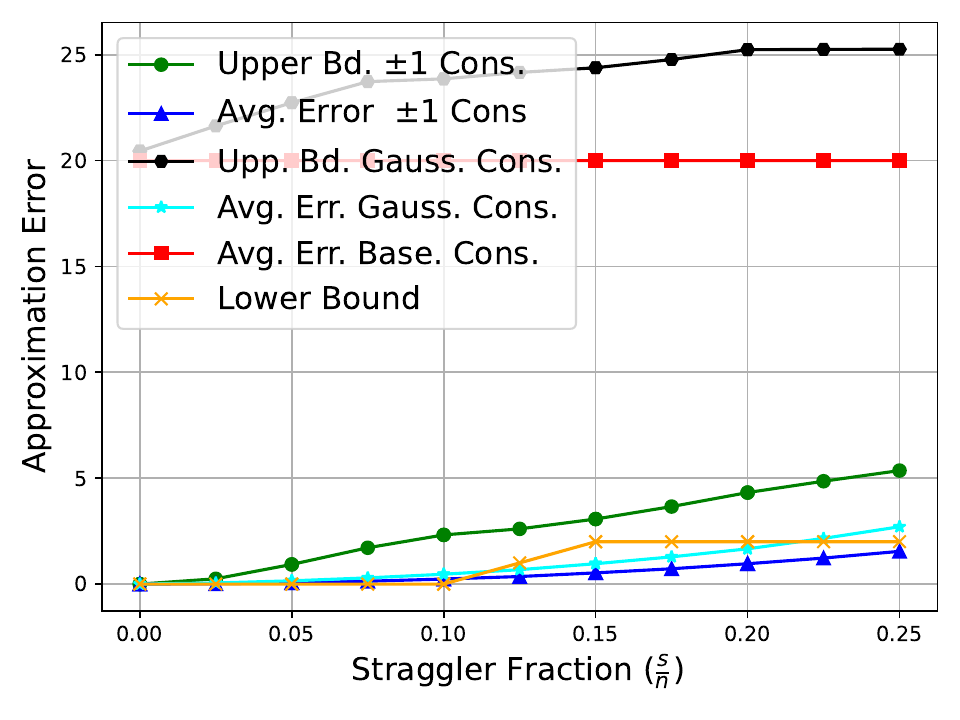}
        \caption{ Performance of the construction as in  $\eqref{eq:recons}$ for a $(40,20,3,6)$ bi-regular bipartite graph.}
        \label{fig:DiagDom}
        
    \end{subfigure}
    \hfill
    \begin{subfigure}[t]{0.45\textwidth}
        \centering
        \includegraphics[
            width=\linewidth,
          ,
            clip
        ]{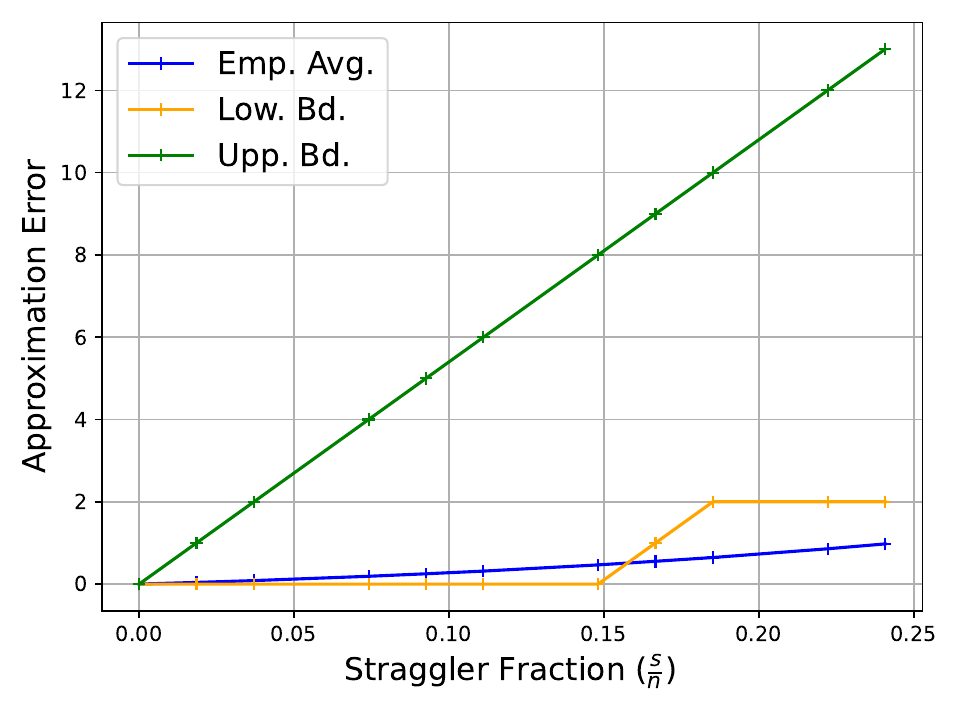}
        \caption{Performance of the construction as in  \eqref{eq:hada_coset_cons} for a $(27,2,5)$ coset bipartite graph.}
        \label{fig:hada_coset_cons}
        
    \end{subfigure}

    \caption{{\small Performance of the constructions in \eqref{eq:recons} and \eqref{eq:hada_coset_cons}.}}
    \label{fig:coset_cons}
   
\end{figure*}

    


%

Next, we chose a $(27, 2, 5)$ coset bipartite graph to  evaluate the construction as in \eqref{eq:effAD1AD2} with $\epsilon = 0.1$. In this case, $ m = 2$, $n = 54$, $k = 27$ and $\delta = 5$. The empirical average was calculated the same way as in \autoref{fig:bibd_31_exact}. Note that in this case the bound in \eqref{eq:AD1AD2err} depends on the specific sets of stragglers and the bound in \eqref{eq:AD1AD2err_Coset} depends on the number of stragglers only. Although the bound appears loose for zero stragglers, the approximation error is zero because the encoding matrix in this case is invertible (see \autoref{fig:rdm_cos_bound}). To examine how the approximation error translates to the actual
gradient computation, we also evaluate the relative squared gradient error,
defined as
$
\frac{\left\|\mathbf{Z}\mathbf{B}_{\mathcal{F}}\mathbf{R}
-\mathbf{Z}\mathbf{F}\right\|_F^2}
{\left\|\mathbf{Z}\mathbf{F}\right\|_F^2},
$
where $\mathbf{Z}\mathbf{F}$ is the true gradient and
$\mathbf{Z}\mathbf{B}_{\mathcal{F}}\mathbf{R}$ is the decoded gradient. 
This evaluation was performed using the MNIST dataset and a fully connected neural network with input dimension 784, one hidden layer with 128 neurons and ReLU activation, and an output layer with 10 neurons and softmax activation. We calculate the empirical average of this
relative squared gradient error over the sampled realizations. \autoref{fig:err_gradd_err} compares this empirical average relative squared gradient error with the empirical average approximation error. The two
quantities increase with the straggler fraction and exhibit a similar
overall trend.

We also evaluate the construction in \eqref{eq:hada_coset_cons} using the $(27,2,5)$
coset bipartite graph with $m=2$. For each number of stragglers, the
approximation error is averaged over 1000 random choices of the straggler set and
compared with the upper bound in Lemma \ref{lemma:hada_coset_bound} and the lower bound.
As shown in \autoref{fig:hada_coset_cons}, the empirical average error remains well below the
upper bound.

Next, we consider the construction in Section \ref{sec:construction_RHPNSC} (see Algorithm \ref{alg:example}). We chose the bi-adjacency matrix of a $(40,20,3,6)$ bi-regular bipartite graph as the assignment matrix and $m = 2$.
In Algorithm 1, we picked $\mathbf{v}_1 = \mathbf{1}_n$ and $\mathbf{v}_2$ such that the entries of $\mathbf{v}_2$ are $\pm 1$. Note that in this particular case, we were successful in finding a $\bfv_2$ that satisfied the null-space constraints. However, this is not always guaranteed. If no restrictions are placed on $\bfv_2$ other than the null-space constraints, then the existence of a nonzero null-space vector is guaranteed. 
For a given number of stragglers, the upper bound \eqref{eq:bounddiagdom} was calculated for each of 1000 random choices of stragglers and the maximum of them is shown (see \autoref{fig:DiagDom}). The average approximation error $\eqref{eq:ApproxError}$ was calculated for the same 1000 choices of stragglers, for this construction and the baseline construction. The upper bound in this case is significantly lower than the baseline scheme error in this case as well.
Next, we picked $\mathbf{v}_1$ from a Gaussian $\mathcal{N}(\mathbf{0}_n,\bfI_n)$ distribution and $\bfv_2$ according to the null-space constraints of Algorithm \ref{alg:example}. In this case, the simulation results show that the error is quite low. However, the upper bound based on \eqref{eq:bounddiagdom} is quite loose. For both cases, as expected, the error is zero when there are no stragglers. We note that in \autoref{fig:coset_rdm} and \autoref{fig:coset_cons}, the average errors are lower than the lower bound (obtained from \eqref{eq:lowerbound}) at slightly higher straggler fractions. This is because the lower bound corresponds to a constructed worst-case non-straggler set of size $n-s$ that has the corresponding approximation error, whereas the average error curves are generated by random choice of the non-stragglers. 



Finally, we compare the convergence of our scheme with the baseline construction. For training, we considered the MNIST dataset, and the fully connected neural network discussed above.
%
For comparison, we chose a  $(7, 7, 3,3, 1)$ BIBD and a $(27, 2, 5)$ coset bipartite graph with $m = 2$ for both cases. We set the straggling probability to $q = .25$. The encoding matrix follows construction  \eqref{eq:effAD1AD2}, with $\epsilon  = 0$ for the BIBD case and $\epsilon  = .1$ for the coset bipartite graph case, compared against the corresponding baseline construction, i.e., $\bfB^T = \begin{bmatrix} \bfA^T & \bfA^T\end{bmatrix}$. The results are averaged over 20 random initializations. It can be observed that in both cases, our construction converges faster (see  \autoref{fig:loss_coset} and \autoref{fig:loss_bibd}).


\begin{figure}[t!]
    \centering
     \begin{subfigure}[t]{0.48\linewidth}
        \centering
        \includegraphics[width= .9\linewidth]{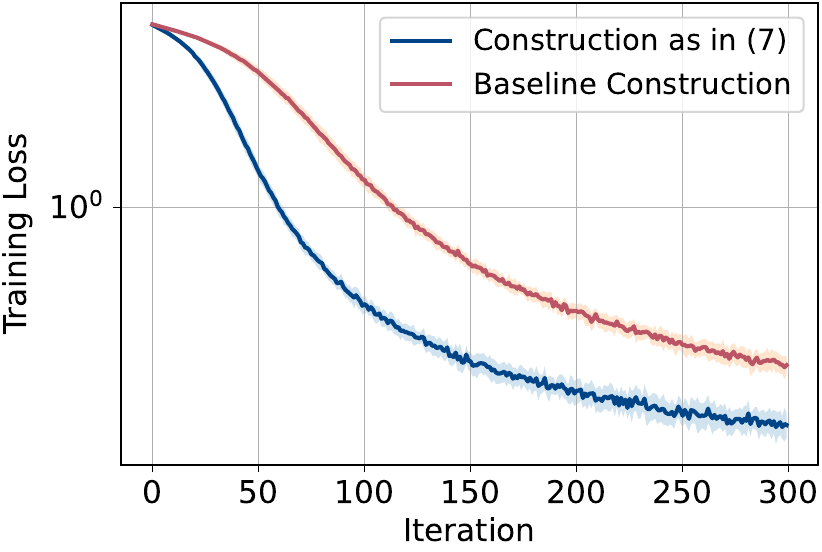}
    \caption{{\small $(27, 2, 5)$ coset bipartite graph}}
    \label{fig:loss_coset}
    
 \end{subfigure}
    \hfill
    \begin{subfigure}[t]{0.48\linewidth}
        \centering
        \includegraphics[width= .9\linewidth]{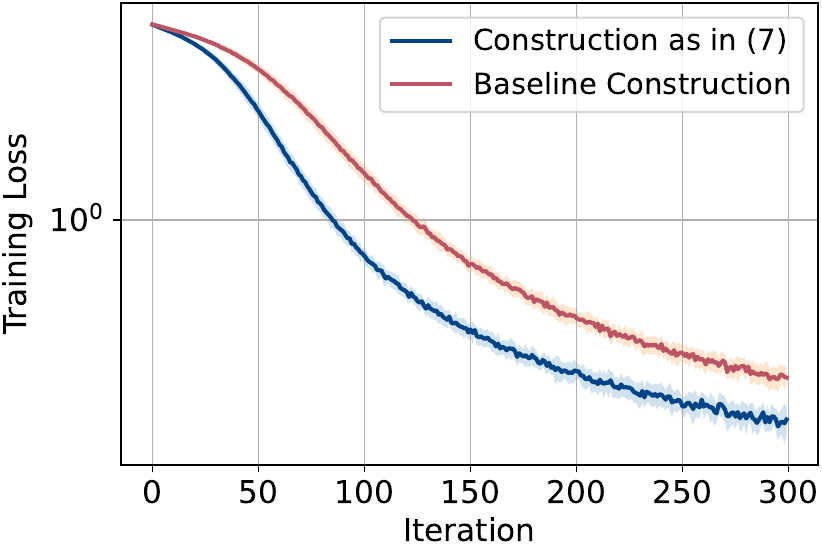}
    \caption{{\small  $(7, 7, 3,3, 1)$ BIBD }}
    \label{fig:loss_bibd}
    \end{subfigure}
    \caption{Training loss for a small neural network with different constructions, $m = 2$.}
    \label{fig:training_loss_comparison}

\end{figure}

\section{Conclusions and Future Work}
In this work, we propose approximate gradient coding schemes that are communication-efficient. Our constructions are based on post multiplication by diagonal matrices and Hadamard products with null-space constraints, applied to structured assignment matrices from BIBDs, strongly regular graphs, and coset bipartite graphs. Moreover, we prove convergence for some of our constructions with specific assignment matrices. We validate our constructions through numerical experiments, which show improved approximation error and faster convergence relative to the baseline constructions. 

There is ample scope for future work. We expect to develop constructions with tighter approximation error bounds and lower empirical error across a wide range of straggler sets. Additionally, schemes that incorporate partial computations from straggling workers (rather than ignoring the partial results entirely) are of particular interest. Extending our schemes to heterogeneous straggler settings while maintaining communication efficiency is also an important open problem.

\bibliographystyle{IEEEtran}
\bibliography{BG, merged_clean}



\appendices
\section{Expected Approximation Error Under Random Straggling}
\label{sec:exp_app_err_rand_strr}
We first average the fixed non-straggler set error bounds derived in Corollaries \ref{cor:BIBD}, \ref{cor:SRG}, and \ref{cor:bipartite} over the random set of non-straggling workers. The following lemma provides a common bound that is subsequently applied to the BIBD, SRG, and coset bipartite graph constructions.

\begin{lemma}
\label{lemma:expected_error_random_straggling}
Suppose that each worker straggles independently with probability $q$, and let $\mathcal{F}\subseteq[n]$ be the random set of non-straggling workers, and let $f = |\mathcal{F}|$. Suppose that there exist constants $C_0>0$, $a_0>0$, and $b_0\geq 0$ such that, for every nonempty $\cF$,
\[
\mathbb{E}_{\mathbf{D}}
\left[
\text{Err}_{\mathcal{F}}(\mathbf{B})
\right]
\leq
mk-\frac{C_0f}{a_0+b_0(f-1)}.
\]
Then,
\[
\mathbb{E}_{\mathcal{F}}
\left[
\mathbb{E}_{\mathbf{D}}
\left[
\text{Err}_{\mathcal{F}}(\mathbf{B})
\right]
\right]
\leq
mk-
\frac{C_0n(1-q)}
{a_0+b_0(n-1)}.
\]
\end{lemma}

\begin{proof}

When $\mathcal{F}=\emptyset$, no worker returns a message. So,
\[
\text{Err}_{\emptyset}(\mathbf{B})
=
\|\mathbf{F}\|_F^2
=
mk.
\]
Therefore,

\begin{align*}
&\mathbb{E}_{\mathcal{F}}
\left[
\mathbb{E}_{\mathbf{D}}
\left[
\operatorname{Err}_{\mathcal{F}}(\mathbf{B})
\right]
\right]
\\
&=
\sum_{{F}\subseteq[n]} \mathbb{P}(\mathcal{F} = F)
\mathbb{E}_{\mathbf{D}}
\left[
\operatorname{Err}_{{F}}(\mathbf{B})
\right]
\\
&=
q^n mk
+
\sum_{f=1}^{n}
\sum_{\substack{{F}\subseteq[n]\\|{F}|=f}} (1-q)^f q^{n-f}
\mathbb{E}_{\mathbf{D}}
\left[
\operatorname{Err}_{{F}}(\mathbf{B})
\right]
\\
&\le
q^n mk
+
\sum_{f=1}^{n}
\binom{n}{f} (1-q)^f q^{n-f}
\left(
mk-\frac{C_0f}{a_0+b_0(f-1)}
\right)
.
\end{align*}

For every $f\in\{1,\ldots,n\}$, since $b_0\geq0$,
\[
a_0+b_0(f-1)
\leq
a_0+b_0(n-1).
\]
Consequently,
\[
\frac{f}{a_0+b_0(f-1)}
\geq
\frac{f}{a_0+b_0(n-1)}.
\]
It follows that
\[
\begin{aligned}
&\mathbb{E}_{\mathcal{F}}
\left[
\mathbb{E}_{\mathbf{D}}
\left[
\text{Err}_{\mathcal{F}}(\mathbf{B})
\right]
\right]
\\
&\leq
q^nmk
+
\sum_{f=1}^{n}
\binom{n}{f} (1-q)^f q^{n-f}
\left(
mk-\frac{C_0f}{a_0+b_0(n-1)}
\right)
\\
&=
mk-
\frac{C_0}{a_0+b_0(n-1)}
\sum_{f=1}^{n}f\binom{n}{f} (1-q)^f q^{n-f}
\\
&=
mk-
\frac{C_0n(1-q) ((1-q)+q)^{n-1}}
{a_0+b_0(n-1)}.
\end{aligned}
\]

So,
\[
\mathbb{E}_{\mathcal{F}}
\left[
\mathbb{E}_{\mathbf{D}}
\left[
\text{Err}_{\mathcal{F}}(\mathbf{B})
\right]
\right]
\leq
mk-
\frac{C_0n(1-q)}
{a_0+b_0(n-1)}.
\]
\end{proof}

We now apply Lemma~\ref{lemma:expected_error_random_straggling} to the BIBD, SRG, and coset bipartite graph constructions.

\begin{cor}
\label{cor:expected_error_BIBD_strr}
Suppose that the assignment matrix $\mathbf{A}$ is the transpose of the incidence matrix of a $(n,k,\gamma,\delta,\lambda)$ BIBD. Then, for the construction as in \eqref{eq:effAD1AD2} with $\epsilon=0$,
\[
\mathbb{E}_{\mathcal{F}}
\left[
\mathbb{E}_{\mathbf{D}}
\left[
\text{Err}_{\mathcal{F}}(\mathbf{B})
\right]
\right]
\leq
mk-
\frac{m\delta^2n(1-q)}
{m\delta+\lambda(n-1)}.
\]
\end{cor}

\begin{proof}
For a fixed non-empty set $\mathcal{F}$ with $|\mathcal{F}|=f$, Corollary~\ref{cor:BIBD} gives
\[
\mathbb{E}_{\mathbf{D}}
\left[
\text{Err}_{\mathcal{F}}(\mathbf{B})
\right]
\leq
mk-
\frac{m\delta^2f}
{m\delta+\lambda(f-1)}.
\]
The result follows from Lemma~\ref{lemma:expected_error_random_straggling} by setting
\[
C_0=m\delta^2,
\qquad
a_0=m\delta,
\qquad
b_0=\lambda.
\]
For a BIBD, $\delta>0$ and $\lambda\geq0$. Hence, $a_0>0$, $b_0\geq0$ and $C_0 > 0$.
\end{proof}

\begin{cor}
\label{cor:expected_error_SRG_strr}
Suppose that the assignment matrix $\mathbf{A}$ is the adjacency matrix of a $(n,\delta,\lambda,\mu)$ SRG with $0<\delta\neq\mu$. Then, for the construction as in \eqref{eq:effAD1AD2} with $\epsilon=0$,
\[
\mathbb{E}_{\mathcal{F}}
\left[
\mathbb{E}_{\mathbf{D}}
\left[
\text{Err}_{\mathcal{F}}(\mathbf{B})
\right]
\right]
\leq
mk-
\frac{m\delta^2n(1-q)}
{m\delta+(\lambda-\mu)\theta+\mu(n-1)},
\]
where $\theta$ is defined as in Corollary~\ref{cor:SRG}.
\end{cor}

\begin{proof}
For a fixed non-empty set $\mathcal{F}$ with $|\mathcal{F}|=f$, Corollary~\ref{cor:SRG} gives
\[
\mathbb{E}_{\mathbf{D}}
\left[
\text{Err}_{\mathcal{F}}(\mathbf{B})
\right]
\leq
mk-
\frac{m\delta^2f}
{m\delta+(\lambda-\mu)\theta+\mu(f-1)}.
\]
The result follows from Lemma~\ref{lemma:expected_error_random_straggling} by setting
\[
C_0=m\delta^2,
\qquad
a_0=m\delta+(\lambda-\mu)\theta,
\qquad
b_0=\mu.
\]

By definition, $\mu\geq0$. We next verify that $a_0>0$. If $\lambda\geq\mu$, then $\theta=\delta$, and hence
\[
a_0
=
m\delta+(\lambda-\mu)\delta
>0.
\]
If $\lambda<\mu$, then
\[
\theta
=
\frac{1}{2}
\left[
(\lambda-\mu)
-
\sqrt{(\lambda-\mu)^2+4(\delta-\mu)}
\right]
<0.
\]
Since $\lambda-\mu<0$, we have
\[
(\lambda-\mu)\theta>0,
\]
and therefore
\[
a_0=m\delta+(\lambda-\mu)\theta>0.
\]
Thus, the conditions of Lemma~\ref{lemma:expected_error_random_straggling} are satisfied.
\end{proof}

\begin{cor}
\label{cor:expected_error_coset_strr}
Suppose that the assignment matrix $\mathbf{A}$ is the bi-adjacency matrix of a $(k,m,\delta)$ coset bipartite graph such that $k=p^a$, where $p$ is a prime, $a\in\mathbb{Z}_{\geq1}$, and $p\nmid\delta$. Then, for the construction as in \eqref{eq:effAD1AD2} with $\epsilon > 0$,
\[
\mathbb{E}_{\mathcal{F}}
\left[
\mathbb{E}_{\mathbf{D}}
\left[
\text{Err}_{\mathcal{F}}(\mathbf{B})
\right]
\right]
\leq
mk-
\frac{m\delta^2n(1-q)}
{m\delta^2+c(m-1)\delta},
\]
where $n=mk$.
\end{cor}

\begin{proof}
For a fixed non-empty set $\mathcal{F}$ with $|\mathcal{F}|=f$, Corollary~\ref{cor:bipartite} gives
\[
\mathbb{E}_{\mathbf{D}}
\left[
\text{Err}_{\mathcal{F}}(\mathbf{B})
\right]
\leq
mk-
\frac{m\delta^2f}
{m\delta^2+c(m-1)\delta}.
\]
The result follows from Lemma~\ref{lemma:expected_error_random_straggling} by setting
\[
C_0=m\delta^2,
\qquad
a_0=m\delta^2+c(m-1)\delta,
\qquad
b_0=0.
\]
Since $m,\delta,c>0$, we have $a_0>0$. Thus, the conditions of Lemma~\ref{lemma:expected_error_random_straggling} are satisfied.
\end{proof}

\begin{lemma}
\label{lemma:hada_coset_bound}
Suppose that the assignment matrix $\bfA$ is the bi-adjacency matrix of a
$(k,m,\delta)$ coset bipartite graph such that $k=p^a$, where $p$ is a prime,
$a\in\mathbb{Z}_{\geq 1}$, and $p\nmid\delta$. Then, for the construction as
in \eqref{eq:hada_coset_cons}, for any non-straggler set $\cF$ with
$|\cF|=n-s$,
\begin{equation}
\label{eq:Hada_Coset_err}
\text{Err}_{\cF}(\bfB)\leq s.
\end{equation}
Consequently,
\begin{equation}
\label{eq:Hada_Coset_exp_err}
\mathbb{E}_{\cF}\left[\text{Err}_{\cF}(\bfB)\right]\leq nq.
\end{equation}
\end{lemma}

\begin{proof}
First consider a non-empty set $\cF$. For $i\in[m]$, observe that
$\mathbf{f}_i^T\mathbf{f}_i=k$. Since $\bfB$ is invertible,
$\bfB_{\cF}^T\bfB_{\cF}$ is invertible.  Let
$\bfP_{\cF}:=\bfI_n(:,\cF)\in\mathbb{R}^{n\times(n-s)}$
be the selection matrix that selects the columns corresponding to $\cF$.
Thus, $\bfB_{\cF}=\bfB\bfP_{\cF}$. We have
\begin{align*}
&\mathbf{f}_i^T\mathbf{f}_i
-\mathbf{f}_i^T\bfB\bfP_{\cF}
(\bfB_{\cF}^T\bfB_{\cF})^{-1}
\bfP_{\cF}^T\bfB^T\mathbf{f}_i\\
&=k-\|\bfP_{\cF}^T\bfB^T\mathbf{f}_i\|_2^2
\frac{\mathbf{f}_i^T\bfB\bfP_{\cF}}
{\|\bfP_{\cF}^T\bfB^T\mathbf{f}_i\|_2}
(\bfB_{\cF}^T\bfB_{\cF})^{-1}
\frac{\bfP_{\cF}^T\bfB^T\mathbf{f}_i}
{\|\bfP_{\cF}^T\bfB^T\mathbf{f}_i\|_2}\\
&\leq
k-\|\bfP_{\cF}^T\bfB^T\mathbf{f}_i\|_2^2
\lambda_{\min}\left(
(\bfB_{\cF}^T\bfB_{\cF})^{-1}
\right).
\end{align*}

Recall that $\mathbf{f}_i=\mathbf{e}_i\otimes\mathbf{1}_k$. Since
$\bfB=\bfH^T\otimes\bfC$,
\begin{align*}
\bfB^T\mathbf{f}_i
&=(\bfH\otimes\bfC^T)
(\mathbf{e}_i\otimes\mathbf{1}_k)\\
&=(\bfH\mathbf{e}_i)\otimes
(\bfC^T\mathbf{1}_k)\\
&=\delta(\bfH\mathbf{e}_i)\otimes\mathbf{1}_k.
\end{align*}
Since every entry of $\bfH\mathbf{e}_i$ has magnitude one,
every entry of $\bfB^T\mathbf{f}_i$ has magnitude $\delta$. Therefore,
\begin{equation}
\label{eq:Hada_Coset_f_norm}
\|\bfP_{\cF}^T\bfB^T\mathbf{f}_i\|_2^2
=\delta^2(n-s).
\end{equation}

Moreover,
$
\bfB^T\bfB
=
\bfH\bfH^T\otimes\bfC^T\bfC
=
m\bfI_m\otimes\bfC^T\bfC.
$ Since $\bfC$ is $\delta$-regular,
$
\lambda_{\max}(\bfC^T\bfC)=\delta^2,
$
and hence
$
\lambda_{\max}(\bfB^T\bfB)=m\delta^2.
$
Also,
$
\bfB_{\cF}^T\bfB_{\cF}
=
\bfP_{\cF}^T\bfB^T\bfB\bfP_{\cF}.
$
Thus, by Cauchy's interlacing theorem,
$
\lambda_{\max}(\bfB_{\cF}^T\bfB_{\cF})
\leq m\delta^2
$. Consequently,
\[
\lambda_{\min}\left(
(\bfB_{\cF}^T\bfB_{\cF})^{-1}
\right)
=
\frac{1}{
\lambda_{\max}(\bfB_{\cF}^T\bfB_{\cF})
}
\geq \frac{1}{m\delta^2}.
\]

Using \eqref{eq:Hada_Coset_f_norm}, for each $i\in[m]$,
\begin{align*}
&\mathbf{f}_i^T\mathbf{f}_i
-\mathbf{f}_i^T\bfB\bfP_{\cF}
(\bfB_{\cF}^T\bfB_{\cF})^{-1}
\bfP_{\cF}^T\bfB^T\mathbf{f}_i\\
&\leq
k-\delta^2(n-s)\frac{1}{m\delta^2}\\
&=
k-\frac{n-s}{m}.
\end{align*}
Summing over $i\in[m]$ gives
\[
\text{Err}_{\cF}(\bfB)
\leq
mk-(n-s).
\]
Since $n=mk$ for a $(k,m,\delta)$ coset bipartite graph,
$
\text{Err}_{\cF}(\bfB)\leq s.
$ For $\cF=\emptyset$, we have
$\text{Err}_{\emptyset}(\bfB)=\|\bfF\|_F^2=mk=n$. So \eqref{eq:Hada_Coset_err} also holds since in this case $s=n$. Since each worker straggles independently with probability $q$, 
$
\mathbb{E}_{\cF}\left[\text{Err}_{\cF}(\bfB)\right]
\leq
\mathbb{E}_{\cF}[s]
=
nq,
$ which proves \eqref{eq:Hada_Coset_exp_err}.
\end{proof}

\section{Condition Number of the Decoding Step}
\label{sec:appendix_cond_number}

In this appendix, we analyze the numerical stability of the least-squares decoding step for the construction in $\eqref{eq:effAD1AD2}$ with BIBD and SRG assignment matrices. For a fixed non-straggler set $\mathcal{F}$, we derive upper bounds on the condition number $\kappa(\mathbf{B}_{\mathcal{F}}^T\mathbf{B}_{\mathcal{F}})$. The following claim first bounds this condition number in terms of $\boldsymbol{\Delta}_{\mathcal{F}}=\mathbf{A}_{\mathcal{F}}^{T}\mathbf{A}_{\mathcal{F}}$.


\begin{claim}
\label{claim:cond_num}
Suppose that $\mathbf{\Delta}_{\mathcal{F}}$ is invertible. Then, for the construction in~(7),
$$
\kappa\!\left(\mathbf{B}_{\mathcal{F}}^{T}\mathbf{B}_{\mathcal{F}}\right)
\leq
\kappa\!\left(\mathbf{\Delta}_{\mathcal{F}}\right)
\frac{\sum_{j=1}^{m}\left\|\tilde{\mathbf{D}}_{j}\right\|_{2}^{2}}
{\sum_{j=1}^{m}\sigma_{\min}^{2}\!\left(\tilde{\mathbf{D}}_{j}\right)}
\leq
\left(\frac{1+\epsilon}{1-\epsilon}\right)^{2}
\kappa\!\left(\mathbf{\Delta}_{\mathcal{F}}\right).
$$
\end{claim}


\begin{proof}
Recall that
\[
\mathbf{B}_{\mathcal F}^{T}\mathbf{B}_{\mathcal F}
=
\sum_{j=1}^{m}
\tilde{\mathbf{D}}_j
\mathbf{\Delta}_{\mathcal F}
\tilde{\mathbf{D}}_j,
\]
where each $\tilde{\mathbf{D}}_j$ is an invertible diagonal matrix. Since $\mathbf{\Delta}_{\mathcal{F}}=\mathbf{A}_{\mathcal{F}}^T\mathbf{A}_{\mathcal{F}}$ is positive semidefinite and invertible, it is positive definite. For any vector $\mathbf{y}$,
\[
\lambda_{\min}(\mathbf{\Delta}_{\mathcal F})\|\mathbf{y}\|_2^2
\le
\mathbf{y}^{T}\mathbf{\Delta}_{\mathcal F}\mathbf{y}
\le
\lambda_{\max}(\mathbf{\Delta}_{\mathcal F})\|\mathbf{y}\|_2^2.
\]
Let $\mathbf{y}= \tilde{\mathbf{D}}_j\mathbf{x}$. Then, for any unit-norm vector $\mathbf{x}$,
\[
\mathbf{x}^{T}
\tilde{\mathbf{D}}_j
\mathbf{\Delta}_{\mathcal F}
\tilde{\mathbf{D}}_j
\mathbf{x}
=
\mathbf{y}^{T}{\mathbf{\Delta}}_{\mathcal F}\mathbf{y}
\le
\lambda_{\max}({\mathbf{\Delta}}_{\mathcal F})
\|\mathbf{y}\|_2^2.
\]
Since
\[
\|\mathbf{y}\|_2
=
\|\tilde{\mathbf{D}}_j\mathbf{x}\|_2
\le
\|\tilde{\mathbf{D}}_j\|_2\|\mathbf{x}\|_2
=
\|\tilde{\mathbf{D}}_j\|_2,
\]
we get
\[
\mathbf{x}^{T}
\tilde{\mathbf{D}}_j
\mathbf{\Delta}_{\mathcal F}
\tilde{\mathbf{D}}_j
\mathbf{x}
\le
\lambda_{\max}(\mathbf{\Delta}_{\mathcal F})
\|\tilde{\mathbf{D}}_j\|_2^2.
\]
Taking the maximum over all unit-norm vectors $\mathbf{x}$ gives
\[
\lambda_{\max}
\!\left(
\tilde{\mathbf{D}}_j
\mathbf{\Delta}_{\mathcal F}
\tilde{\mathbf{D}}_j
\right)
\le
\|\tilde{\mathbf{D}}_j\|_2^2
\lambda_{\max}(\mathbf{\Delta}_{\mathcal F}).
\]

Similarly,
\[
\mathbf{x}^{T}
\tilde{\mathbf{D}}_j
\mathbf{\Delta}_{\mathcal F}
\tilde{\mathbf{D}}_j
\mathbf{x}
=
\mathbf{y}^{T}\mathbf{\Delta}_{\mathcal F}\mathbf{y}
\ge
\lambda_{\min}(\mathbf{\Delta}_{\mathcal F})
\|\mathbf{y}\|_2^2.
\]
Since
\[
\|\mathbf{y}\|_2
=
\|\tilde{\mathbf{D}}_j\mathbf{x}\|_2
\ge
\sigma_{\min}(\tilde{\mathbf{D}}_j)\|\mathbf{x}\|_2
=
\sigma_{\min}(\tilde{\mathbf{D}}_j),
\]
we get
\[
\mathbf{x}^{T}
\tilde{\mathbf{D}}_j
\mathbf{\Delta}_{\mathcal F}
\tilde{\mathbf{D}}_j
\mathbf{x}
\ge
\lambda_{\min}({\mathbf{\Delta}}_{\mathcal F})
\sigma_{\min}^2(\tilde{\mathbf{D}}_j).
\]
Taking the minimum over all unit vectors $\mathbf{x}$ gives
\[
\lambda_{\min}
\!\left(
\tilde{\mathbf{D}}_j
\mathbf{\Delta}_{\mathcal F}
\tilde{\mathbf{D}}_j
\right)
\ge
\sigma_{\min}^2(\tilde{\mathbf{D}}_j)
\lambda_{\min}(\mathbf{\Delta}_{\mathcal F}).
\]

Now, by Weyl's inequality,
\[
\lambda_{\max}
\!\left(
\mathbf{B}_{\mathcal F}^{T}
\mathbf{B}_{\mathcal F}
\right)
\le
\sum_{j=1}^{m}
\lambda_{\max}
\!\left(
\tilde{\mathbf{D}}_j
\mathbf{\Delta}_{\mathcal F}
\tilde{\mathbf{D}}_j
\right),
\]
and therefore
\[
\lambda_{\max}
\!\left(
\mathbf{B}_{\mathcal F}^{T}
\mathbf{B}_{\mathcal F}
\right)
\le
\lambda_{\max}(\mathbf{\Delta}_{\mathcal F})
\sum_{j=1}^{m}\|\tilde{\mathbf{D}}_j\|_2^2.
\]
Similarly,
\[
\lambda_{\min}
\!\left(
\mathbf{B}_{\mathcal F}^{T}
\mathbf{B}_{\mathcal F}
\right)
\ge
\sum_{j=1}^{m}
\lambda_{\min}
\!\left(
\tilde{\mathbf{D}}_j
\mathbf{\Delta}_{\mathcal F}
\tilde{\mathbf{D}}_j
\right),
\]
and hence
\[
\lambda_{\min}
\!\left(
\mathbf{B}_{\mathcal F}^{T}
\mathbf{B}_{\mathcal F}
\right)
\ge
\lambda_{\min}(\mathbf{\Delta}_{\mathcal F})
\sum_{j=1}^{m}\sigma_{\min}^2(\tilde{\mathbf{D}}_j).
\]

Thus,
\[
\kappa
\!\left(
\mathbf{B}_{\mathcal F}^{T}
\mathbf{B}_{\mathcal F}
\right)
=
\frac{
\lambda_{\max}
\!\left(
\mathbf{B}_{\mathcal F}^{T}
\mathbf{B}_{\mathcal F}
\right)
}{
\lambda_{\min}
\!\left(
\mathbf{B}_{\mathcal F}^{T}
\mathbf{B}_{\mathcal F}
\right)
}
\le
\kappa(\mathbf{\Delta}_{\mathcal F})
\frac{
\sum_{j=1}^{m}\|\tilde{\mathbf{D}}_j\|_2^2
}{
\sum_{j=1}^{m}\sigma_{\min}^2(\tilde{\mathbf{D}}_j)
}.
\]

Since the absolute value of every diagonal entry of $\tilde{\mathbf{D}}_j$ is between $1-\epsilon$ and $1+\epsilon$, we have
$$
\|\tilde{\mathbf{D}}_j\|_2\leq 1+\epsilon
$$
and
$$
\sigma_{\min}(\tilde{\mathbf{D}}_j)\geq 1-\epsilon.
$$
Therefore,
$$
\sum_{j=1}^{m}\|\tilde{\mathbf{D}}_j\|_2^2
\leq
m(1+\epsilon)^2
$$
and
$$
\sum_{j=1}^{m}\sigma_{\min}^2(\tilde{\mathbf{D}}_j)
\geq
m(1-\epsilon)^2.
$$
Consequently,
$$
\begin{aligned}
\kappa\left(\mathbf{B}_{\mathcal{F}}^T
\mathbf{B}_{\mathcal{F}}\right)
&\leq
\kappa\left(\mathbf{\Delta}_{\mathcal{F}}\right)
\frac{
\sum_{j=1}^{m}\|\tilde{\mathbf{D}}_j\|_2^2
}{
\sum_{j=1}^{m}\sigma_{\min}^2(\tilde{\mathbf{D}}_j)
}\\
&\leq
\left(\frac{1+\epsilon}{1-\epsilon}\right)^2
\kappa\left(\mathbf{\Delta}_{\mathcal{F}}\right).
\end{aligned}
$$
This completes the proof.

\end{proof}
\begin{cor}
Suppose that the assignment matrix $\mathbf{A}$ is the transpose of the incidence matrix of a $(n,k,\gamma,\delta,\lambda)$ BIBD. Then, for the construction as in (7), for every nonempty $\mathcal{F}\subseteq\{1,\ldots,n\}$ of size $n-s$,
$$
\kappa\left(\mathbf{B}_{\mathcal{F}}^T
\mathbf{B}_{\mathcal{F}}\right)
\leq
\left(\frac{1+\epsilon}{1-\epsilon}\right)^2
\frac{\delta+(n-s-1)\lambda}{\delta-\lambda}.
$$
\end{cor}

\begin{proof}
Recall that
$$
\mathbf{\Delta}_{\mathcal{F}}
= (\delta-\lambda)\mathbf{I}_{n-s}
+\lambda\mathbf{J}_{n-s}.$$
The largest eigenvalue of $\mathbf{\Delta}_{\mathcal{F}}$ is
$$
\lambda_{\max} \left(\mathbf{\Delta}_{\mathcal{F}}\right)
= \delta+(n-s-1)\lambda.
$$
Moreover,
$$
\lambda_{\min}\left(\mathbf{\Delta}_{\mathcal{F}}\right)
\geq
\delta-\lambda.
$$
Since $\delta>\lambda$ for a BIBD, $\mathbf{\Delta}_{\mathcal{F}}$ is positive definite. Therefore,
$$
\begin{aligned}
\kappa\left(\mathbf{\Delta}_{\mathcal{F}}\right)
&=
\frac{
\lambda_{\max}\left(\mathbf{\Delta}_{\mathcal{F}}\right)
}{
\lambda_{\min}\left(\mathbf{\Delta}_{\mathcal{F}}\right)
}
\
&\leq
\frac{\delta+(n-s-1)\lambda}{\delta-\lambda}.
\end{aligned}
$$

The result follows from Claim \ref{claim:cond_num}.
\end{proof}

\remark{For a $(91, 91, 10,10,1)$ BIBD with $\epsilon = 0$, we have $\kappa(\bfB_{\cF}^T\bfB_{\cF}) \le \frac{100}{9} \approx 11.11$.}

\begin{cor}
Suppose that the assignment matrix $\mathbf{A}$ is the adjacency matrix of an $(n,\delta,\lambda,\mu)$ SRG with $\delta\neq\mu$. Let
$$
\theta_{+}
= \frac{1}{2}
\left(
(\lambda-\mu)
+
\sqrt{(\lambda-\mu)^2+4(\delta-\mu)}
\right)
$$
and
$$
\theta_{-}
= \frac{1}{2}
\left(
(\lambda-\mu)
- \sqrt{(\lambda-\mu)^2+4(\delta-\mu)}
\right).
$$
Also, let
$$
\rho
= \min\left\{
|\theta_{+}|,
|\theta_{-}|
\right\}.
$$
Then, for the construction as in (7), for every nonempty $\mathcal{F}\subseteq\{1,\ldots,n\}$,
$$
\kappa\left(\mathbf{B}_{\mathcal{F}}^T
\mathbf{B}_{\mathcal{F}}\right)
\leq
\left(\frac{1+\epsilon}{1-\epsilon}\right)^2
\frac{\delta^2}{\rho^2}.
$$
\end{cor}

\begin{proof}
We have,
$$
\mathbf{\Delta}_{\mathcal{F}}
=\mathbf{A}_{\mathcal{F}}^T\mathbf{A}_{\mathcal{F}} = \mathbf{A}^2(\mathcal{F},\mathcal{F}).
$$

The eigenvalues of $\mathbf{A}$ are $\delta$, $\theta_{+}$, and $\theta_{-}$. Therefore, the eigenvalues of $\mathbf{A}^2$ are $\delta^2$, $\theta_{+}^2$, and $\theta_{-}^2$. Moreover,
$$
\theta_{+}\theta_{-}
= \mu-\delta.
$$
Since $\delta\neq\mu$, both $\theta_{+}$ and $\theta_{-}$ are nonzero. Hence,
$$
\rho>0.
$$
By the Cauchy interlacing theorem,
$$
\lambda_{\max}\left(\mathbf{\Delta}_{\mathcal{F}}\right)
\leq
\delta^2
$$
and
$$
\lambda_{\min}\left(\mathbf{\Delta}_{\mathcal{F}}\right)
\geq
\rho^2.
$$
Therefore, $\mathbf{\Delta}_{\mathcal{F}}$ is positive definite, and
$$
\begin{aligned}
\kappa\left(\mathbf{\Delta}_{\mathcal{F}}\right)
&=
\frac{
\lambda_{\max}\left(\mathbf{\Delta}_{\mathcal{F}}\right)
}{
\lambda_{\min}\left(\mathbf{\Delta}_{\mathcal{F}}\right)
}
\
&\leq
\frac{\delta^2}{\rho^2}.
\end{aligned}
$$
The result follows from Claim \ref{claim:cond_num}.
\end{proof}

\remark{For a $(126, 25, 8,4)$ SRG with $\epsilon = 0$, we have $\kappa(\bfB_{\cF}^T\bfB_{\cF}) \le \frac{625}{9}  \approx 69.44 $.}

\section{Proof of Proposition~\ref{prop:coset-full-support}}
\label{app:coset-full-support}

\begin{proof}
Since $\mathbf{A}$ is the bi-adjacency matrix of a $(k,m,\delta)$ coset bipartite graph, its columns can be ordered such that
\[
\mathbf{A}
=
\left[
\mathbf{C}\ \cdots\ \mathbf{C}
\right],
\]
where $\mathbf{C}\in\mathbb{R}^{k\times k}$. Choose pairwise orthogonal vectors
$\mathbf{q}_1,\ldots,\mathbf{q}_m\in\mathbb{R}^m$ such that
$\mathbf{q}_1=\mathbf{1}_m$ and every entry of each $\mathbf{q}_j$ is nonzero. Such a collection can be obtained, for example, from a suitably scaled Householder matrix \cite{10.1145/320941.320947}. For each $j\in[m]$, let $
\mathbf{v}_j
=
\mathbf{q}_j\otimes\mathbf{1}_k.$ Since every entry of $\mathbf{q}_j$ is nonzero, every entry of $\mathbf{v}_j$ is also nonzero. For each $\ell\in[m]$, let
$
\mathbf{D}_{\mathbf{v}_\ell}
=
\operatorname{diag}(\mathbf{v}_\ell),
$
and let $\mathbf{R}_\ell$ denote the diagonal row-normalization matrix whose $i$-th diagonal entry is
\[
\mathbf{R}_\ell(i,i)
=
\frac{1}
{\left\|\mathbf{v}_\ell(\mathcal{H}_i)\right\|_2^2}.
\]
The matrix constructed at the $\ell$-th iteration of Algorithm \ref{alg:example} can then be written as
\[
\mathbf{A}_\ell
=
\mathbf{R}_\ell
\mathbf{A}
\mathbf{D}_{\mathbf{v}_\ell}.
\]
It suffices to show that,
$
\mathbf{A}_\ell\mathbf{v}_j=\mathbf{0}_k$, for every $\ell < j$.
For any $\ell<j$, we have
\[
\begin{aligned}
\mathbf{A}_\ell\mathbf{v}_j
&=
\mathbf{R}_\ell
\mathbf{A}
\mathbf{D}_{\mathbf{v}_\ell}
\mathbf{v}_j\\
&=
\mathbf{R}_\ell
\mathbf{A}
\left(
\mathbf{v}_\ell\circ\mathbf{v}_j
\right).
\end{aligned}
\]
By the definition of $\mathbf{v}_\ell$ and $\mathbf{v}_j$,
\[
\mathbf{v}_\ell\circ\mathbf{v}_j
=
\begin{bmatrix}
\mathbf{q}_\ell(1)\mathbf{q}_j(1)\mathbf{1}_k\\
\vdots\\
\mathbf{q}_\ell(m)\mathbf{q}_j(m)\mathbf{1}_k
\end{bmatrix}.
\]
Therefore, 
\[
\begin{aligned}
\mathbf{A}
\left(
\mathbf{v}_\ell\circ\mathbf{v}_j
\right)
&=
\left[
\mathbf{C}\ \cdots\ \mathbf{C}
\right]
\begin{bmatrix}
\mathbf{q}_\ell(1)\mathbf{q}_j(1)\mathbf{1}_k\\
\vdots\\
\mathbf{q}_\ell(m)\mathbf{q}_j(m)\mathbf{1}_k
\end{bmatrix}\\
&=
\sum_{t=1}^{m}
\mathbf{q}_\ell(t)\mathbf{q}_j(t)\mathbf{C}\mathbf{1}_k\\
&=
\left(
\sum_{t=1}^{m}\mathbf{q}_\ell(t)\mathbf{q}_j(t)
\right)
\mathbf{C}\mathbf{1}_k\\
&=
\left(
\mathbf{q}_\ell^T\mathbf{q}_j
\right)
\mathbf{C}\mathbf{1}_k.
\end{aligned}
\]
Since $\mathbf{q}_\ell$ and $\mathbf{q}_j$ are orthogonal,
\[
\mathbf{q}_\ell^T\mathbf{q}_j=0.
\]
It follows that
\[
\mathbf{A}_\ell\mathbf{v}_j
=
\mathbf{0}_k
\qquad
\text{for every }\ell<j.
\]
Consequently,
\[
\mathbf{v}_j
\in
\bigcap_{\ell=1}^{j-1}
\text{Null}(\mathbf{A}_\ell)
=
\text{Null}\!\left(
\left[
\mathbf{A}_1^T\ \cdots\ \mathbf{A}_{j-1}^T
\right]^T
\right).
\]

Thus, there exist vectors
$\mathbf{v}_1,\ldots,\mathbf{v}_m$ satisfying the required null-space conditions, with every entry of each $\mathbf{v}_j$ nonzero.
\end{proof}



\section{Proof that  ${\|\mathbf{B}_{\mathcal{C}} \mathbf{R}- \mathbf{F}\|_F^2}  = \sum_{i = 1}^k {\|\mathbf{B}^{(i)} \mathbf{R}- \mathbf{I}_m\|_F^2}$.}
\label{sec:appendix_claim_frob}
\begin{proof}

Let $\tilde{\mathbf{b}}_i^T$ and $\tilde{\mathbf{f}_i}^T$ be the $i^{th}$ row of $\mathbf{B}_C$ and $\mathbf{F}$ respectively, for $i \in [mk]$. Since the squared Frobenius norm of a matrix is the sum of the squared Frobenius norm of each row of that matrix, 

\begin{align*}
\|\mathbf{B}_{\mathcal{C}} \mathbf{R}- \mathbf{F}\|_F^2 &= \sum_{i = 1}^{mk} \|\tilde{\mathbf{b}}_i^T\mathbf{R}- \tilde{\mathbf{f}}_i^T\|_F^2\\
&= \sum_{i = 1}^k \sum_{j = 1}^m \|\tilde{\mathbf{b}}_{i+ k(j-1)}^T\mathbf{R}- \tilde{\mathbf{f}}_{i+ k(j-1)}^T\|_F^2.
\end{align*}

By definition, for $i \in [k]$, 
\begin{align*}
\mathbf{B}^{(i)}  = \begin{bmatrix}\tilde{\mathbf{b}}_{i}^T \\ \tilde{\mathbf{b}}_{i+ k}^T\\ \vdots \\ \tilde{\mathbf{b}}_{i+ k(m-1)}^T \end{bmatrix}.
\end{align*}

Also, note that, 
\begin{align*}
\begin{bmatrix}
\tilde{\mathbf{f}}_{i}^T \\ \tilde{\mathbf{f}}_{i+ k}^T\\ \vdots \\ \tilde{\mathbf{f}}_{i+ k(m-1)}^T \end{bmatrix} = \mathbf{I}_m.
\end{align*}

Consequently, 
\begin{align*}
\sum_{j = 1}^m \|\tilde{\mathbf{b}}_{i+ k(j-1)}^T\mathbf{R}- \tilde{\mathbf{f}}_{i+ k(j-1)}^T\|_F^2 = \|\mathbf{B}^{(i)}\mathbf{R} - \mathbf{I}_m\|_F^2.
\end{align*}
\end{proof}

\section{Hadamard Matrix Based Construction with Coset Bipartite Graph}
\label{app:hada_cos_unbiased}
\normalfont
In this section, we establish the unbiasedness property of the construction in \eqref{eq:hada_coset_cons}.



\begin{lemma}
\label{lemma:Cost_Hadamard_Convergence}

Suppose that the assignment matrix $\mathbf{A}$ is the bi-adjacency matrix of a $(k, m, \delta)$ coset bipartite graph such that $k = p^a$, where $p$ is a prime, $a \in \mathbb{Z}_{\ge 1}$ and $p \nmid \delta$. Then, for the construction as in $\eqref{eq:hada_coset_cons}$, we have 

\[\mathbb{E}_{\mathcal{F}}[\mathbf{B}_{\cF} \mathbf{R}]= \alpha_H \mathbf{F},\]

where $\alpha_H$ is a nonzero constant.
\end{lemma}
\begin{proof}

Here, 
$\mathbf{B}^T \mathbf{B} = (\mathbf{H}_m^T \otimes \mathbf{C})^T  (\mathbf{H}_m^T \otimes \mathbf{C}) = \mathbf{H}_m\mathbf{H}_m^T \otimes \mathbf{C}^T \mathbf{C} = m \mathbf{I}_m \otimes \mathbf{C}^T \bfC$. For $\cF =\emptyset$, we assume $\bfB_{\cF} \bfR = \mathbf{0}_{mk \times m}$. For non-empty $\cF$,  let $\bfP_{\cF} = \bfI_n(:, \cF)$. We have, 
\begin{align*}
\mathbf{B}_{\mathcal{F}}\mathbf{R} &= \bfB_{\cF}(\bfB_{\cF}^T\bfB_{\cF})^{-1} \bfB_{\cF}^T \mathbf{F}\\
&=\mathbf{B} \mathbf{P}_{\mathcal{F}}(\mathbf{P}_\mathcal{F}^T \mathbf{B}^T\mathbf{B}\mathbf{P}_\mathcal{F})^{-1}\mathbf{P}_{\mathcal{F}}^T \bfB^T\mathbf{F}\\
\end{align*}

For each $i\in[m]$, let
\[
\mathcal{F}_i
:=
\left\{
j\in[k] : (i-1)k+j\in\mathcal{F}
\right\}.
\]
Thus, $\mathcal{F}_i$ denotes the indices of the non-straggling workers within the $i$-th block of $k$ workers. For non-empty $\cF_i$, let $\bfP'_{\cF_i} = \bfI_k(:, \cF_i)$. Thus, when $\cF_i \ne \emptyset$ for all $i \in [m]$,

\[
\mathbf{P}_{\mathcal{F}}
=
\begin{bmatrix}
\mathbf{P}'_{\mathcal{F}_1}
    & \mathbf{0}_{k\times |\mathcal{F}_2|}
    & \cdots
    & \mathbf{0}_{k\times |\mathcal{F}_m|}\\
\mathbf{0}_{k\times |\mathcal{F}_1|}
    & \mathbf{P}'_{\mathcal{F}_2}
    & \cdots
    & \mathbf{0}_{k\times |\mathcal{F}_m|}\\
\vdots
    & \vdots
    & \ddots
    & \vdots\\
\mathbf{0}_{k\times |\mathcal{F}_1|}
    & \mathbf{0}_{k\times |\mathcal{F}_2|}
    & \cdots
    & \mathbf{P}'_{\mathcal{F}_m}
\end{bmatrix}.
\]

Let $\mathbf{G}:=\mathbf{C}^T\mathbf{C}.$ Since $\bfB^T\bfB = m \bfI_m \otimes \bfC^T \bfC$, it follows that
\begin{align*}
\mathbf{P}_{\mathcal{F}}^T\mathbf{B}^T\mathbf{B}
\mathbf{P}_{\mathcal{F}}
& = m
\begin{bmatrix}
\mathbf{P}'^{T}_{\mathcal{F}_1} \mathbf{G}\mathbf{P}'_{\mathcal{F}_1}
& & \mathbf{0}\\
&
\ddots
&\\
\mathbf{0}
& &
\mathbf{P}'^{T}_{\mathcal{F}_m}\mathbf{G}\mathbf{P}'_{\mathcal{F}_m}
\end{bmatrix}
\end{align*}
 Therefore,
\[
\mathbf{P}_{\mathcal{F}}
(\mathbf{P}_{\mathcal{F}}^T\mathbf{B}^T\mathbf{B}
\mathbf{P}_{\mathcal{F}})^{-1}
\mathbf{P}_{\mathcal{F}}^T
=
\frac{1}{m}
\begin{bmatrix}
\mathbf{P}'_{\mathcal{F}_1}
\left(
\mathbf{P}'^{T}_{\mathcal{F}_1}
\mathbf{G}
\mathbf{P}'_{\mathcal{F}_1}
\right)^{-1}
\mathbf{P}'^{T}_{\mathcal{F}_1} & & \mathbf{0}\\
& \ddots &\\
\mathbf{0} & & \mathbf{P}'_{\mathcal{F}_m}
\left(
\mathbf{P}'^{T}_{\mathcal{F}_m}
\mathbf{G}
\mathbf{P}'_{\mathcal{F}_m}
\right)^{-1}
\mathbf{P}'^{T}_{\mathcal{F}_m}
\end{bmatrix}.
\]

For $i\in[m]$, let

\[
\mathbf{T}_{\mathcal{F}_i} :=
\begin{cases}
\mathbf{P}'_{\mathcal{F}_i}
\left(
\mathbf{P}'^{T}_{\mathcal{F}_i}
\mathbf{G}
\mathbf{P}'_{\mathcal{F}_i}
\right)^{-1}
\mathbf{P}'^{T}_{\mathcal{F}_i}, & \mathcal{F}_i \neq \emptyset, \\[0.5em]
\mathbf{0}_{k \times k}, & \mathcal{F}_i = \emptyset.
\end{cases}
\]


Since for $\cF_i = \emptyset$, $\mathbf{T}_{\cF_i} = \mathbf{0}_{k \times k}$, for any $\cF \subseteq [n]$, 

\[\mathbf{B}_{\cF} \bfR = \frac{1}{m}\mathbf{B}\begin{bmatrix}
\mathbf{T}_{\cF_1} & & \mathbf{0}\\
& \ddots &\\
\mathbf{0} & & \mathbf{T}_{\cF_m} 
\end{bmatrix} \bfB^T \bfF\]

Since each worker is independently non-straggling with probability
$1-q$, the random sets $\mathcal{F}_1,\ldots,\mathcal{F}_m$ are
independent. Let
$
\overline{\mathbf{T}}
:=
\mathbb{E}_{\mathcal{F}}
[\mathbf{T}_{\mathcal{F}_i}]
.$ Thus,
\[
\mathbb{E}_{\mathcal{F}}
\left[
\bfB_{\cF}\bfR
\right]
=
\frac{1}{m} \bfB (\mathbf{I}_m\otimes\overline{\mathbf{T}}) \mathbf{B}^T \bfF.
\]

Next, we show that $\overline{\mathbf{T}}$ is circulant. Let
$\mathbf{Q}\in\{0,1\}^{k\times k}$ be the permutation matrix
corresponding to a cyclic shift by one position. Since $\mathbf{C}$
is circulant,
$
\mathbf{Q}^T\mathbf{C}\mathbf{Q}=\mathbf{Q}\mathbf{C}\mathbf{Q}^T =\mathbf{C}.
$ From this we can infer that $\mathbf{Q}^T\mathbf{G}\mathbf{Q} = \mathbf{Q}\mathbf{G}\mathbf{Q}^T
= \mathbf{G}$.

For $\mathcal{F}_i\subseteq[k]$, let $\mathcal{F}_i'$
denote the set obtained by cyclically shifting the elements of
$\mathcal{F}_i$ by one position. Then,
$\mathbf{Q}\mathbf{P}'_{\mathcal{F}_i}$ is an equivalent selection
matrix for $\mathcal{F}_i'$. 
Now,
\begin{align*}
\mathbf{Q} \bfT_{\mathcal{F}_i} \mathbf{Q}^T &= \mathbf{Q} \mathbf{P}'_{\mathcal{F}_i}
\left(
\mathbf{P}'^{T}_{\mathcal{F}_i}
\mathbf{G}
\mathbf{P}'_{\mathcal{F}_i}
\right)^{-1}
\mathbf{P}'^{T}_{\mathcal{F}_i} \mathbf{Q}^T\\
&= \mathbf{Q} \mathbf{P}'_{\mathcal{F}_i}
\left(
\mathbf{P}'^{T}_{\mathcal{F}_i}\mathbf{Q}^T\mathbf{Q}
\mathbf{G}\mathbf{Q}^T\mathbf{Q}
\mathbf{P}'_{\mathcal{F}_i}
\right)^{-1}
\mathbf{P}'^{T}_{\mathcal{F}_i} \mathbf{Q}^T\\
&= \mathbf{Q} \mathbf{P}'_{\mathcal{F}_i}
\left(
\mathbf{P}'^{T}_{\mathcal{F}_i}\mathbf{Q}^T
\mathbf{G}\mathbf{Q}
\mathbf{P}'_{\mathcal{F}_i}
\right)^{-1}
\mathbf{P}'^{T}_{\mathcal{F}_i} \mathbf{Q}^T\\
& = \bfT_{\cF_i'}.
\end{align*}


Since $\cF_i  \overset{d}{=} \cF_i'$, we have $\bfT_{\cF_i} \overset{d} {=} \bfT_{\cF_i'}$.   Thus, $\mathbb{E}_{\cF}\left[\bfT_{\cF_i}\right] = \mathbb{E}_{\cF}\left[\bfT_{\cF_i'}\right]$. Consequently, $
\mathbf{Q}\overline{\mathbf{T}}\mathbf{Q}^T
=
\overline{\mathbf{T}}.
$
So, by noting that $\mathbf{Q}$ corresponds to a cyclic shift by one position, we can conclude that $\overline{\mathbf{T}}$ is circulant. So,
\begin{align*}
\mathbb{E}_{\mathcal{F}}
[\mathbf{B}_{\mathcal{F}}\mathbf{R}]
&=
\frac{1}{m} \bfB (\mathbf{I}_m\otimes\overline{\mathbf{T}}) \mathbf{B}^T \bfF\\
&=
\frac{1}{m}(\mathbf{H}_m^T\otimes\mathbf{C})
\left(\mathbf{I}_m\otimes\overline{\mathbf{T}}
\right)
(\mathbf{H}_m\otimes\mathbf{C}^T)(\bfI_m \otimes \mathbf{1}_k)\\
&=\frac{1}{m} \bfH_m^T\bfH_m \otimes \bfC \overline{\mathbf{T}} \bfC^T\mathbf{1}_k = \bfI_m \otimes  \bfC \overline{\mathbf{T}} \bfC^T\mathbf{1}_k .
\end{align*}

Let
$
\tilde{\mathbf{T}}:=
\mathbf{C}\overline{\mathbf{T}}\mathbf{C}^T.$
Since $\mathbf{C}$ and $\overline{\mathbf{T}}$ are circulant,
$\tilde{\mathbf{T}}$ is also circulant. Therefore, $\mathbf{1}_k$ is an
eigenvector of $\tilde{\mathbf{T}}$, and there exists a constant $\alpha_H$
such that $\tilde{\mathbf{T}}\mathbf{1}_k=\alpha_H\mathbf{1}_k.$

It remains to show that $\alpha_H$ is nonzero. Since $\mathbf{C}$ is
invertible, $\mathbf{G}=\mathbf{C}^T\mathbf{C}$ is positive definite.
Thus, for any non-empty $\mathcal{F}_i$,
$
\mathbf{1}_k^T\mathbf{T}_{\mathcal{F}_i}\mathbf{1}_k
=
\mathbf{1}_{|\mathcal{F}_i|}^T
\left(
\mathbf{P}'^{T}_{\mathcal{F}_i}
\mathbf{G}
\mathbf{P}'_{\mathcal{F}_i}
\right)^{-1}
\mathbf{1}_{|\mathcal{F}_i|}
>0.
$
Since $q<1$, $\mathcal{F}_i$ is non-empty with positive probability.
Therefore,
$
\mathbf{1}_k^T\overline{\mathbf{T}}\mathbf{1}_k>0
.$ Moreover, since $\mathbf{C}$ is $\delta$-regular,
$\mathbf{C}\mathbf{1}_k=\mathbf{C}^T\mathbf{1}_k
=\delta\mathbf{1}_k$. Hence,
$
\alpha_H k
=
\mathbf{1}_k^T\widetilde{\mathbf{T}}\mathbf{1}_k
=
\delta^2
\mathbf{1}_k^T\overline{\mathbf{T}}\mathbf{1}_k
>0.
$ Consequently, $\alpha_H$ is nonzero.



\end{proof}

\end{document}